\PassOptionsToPackage{notheorems,nomaths}{jmlrutils}
\PassOptionsToPackage{hyphens}{url}
\documentclass[pmlr,11pt,twoside]{jmlr}

\usepackage[T1]{fontenc}
\usepackage{newtxtext}
\usepackage{amsthm,mathtools}
\usepackage{newtxmath}
\let\vv\relax

\usepackage{bm,mathrsfs}
\usepackage{xcolor}
\usepackage{graphicx}
\usepackage{algorithm,algorithmic,placeins}
\usepackage{needspace}
\usepackage{booktabs,longtable,multirow,array}
\usepackage{pifont,enumitem,microtype,comment}
\setcitestyle{round,authoryear}
\definecolor{referencecolor}{HTML}{236B63}
\hypersetup{colorlinks=true,linkcolor=referencecolor,
            citecolor=blue,urlcolor=blue}
\usepackage[capitalize,nameinlink,noabbrev]{cleveref}
\jmlrSuppressPackageChecks

\newtheoremstyle{coltplain}{6pt}{6pt}{\itshape}{}%
  {\bfseries}{}{0.5em}{}
\newtheoremstyle{coltdefinition}{6pt}{6pt}{\normalfont}{}%
  {\bfseries}{}{0.5em}{}
\theoremstyle{coltplain}
\newtheorem{theorem}{Theorem}
\newtheorem{lemma}[theorem]{Lemma}
\newtheorem{proposition}[theorem]{Proposition}
\newtheorem{corollary}[theorem]{Corollary}

\theoremstyle{coltdefinition}

\theoremstyle{coltdefinition}
\newtheorem{remark}[theorem]{Remark}
\theoremstyle{coltplain}

\crefname{theorem}{Theorem}{Theorems}
\crefname{example}{Example}{Examples}
\crefname{assumption}{Assumption}{Assumptions}
\crefname{problem}{Problem}{Problems}
\makeatletter
\renewenvironment{proof}[1][\proofname]{%
  \par\pushQED{\qed}\normalfont
  \topsep6\p@\@plus6\p@\relax
  \trivlist\item[\hskip\labelsep\bfseries #1]\ignorespaces
}{%
  \popQED\endtrivlist\@endpefalse
}
\makeatother

\newcommand{\R}{\mathbb{R}}
\newcommand{\E}{\mathbb{E}}

\newcommand{\epsappr}{\varepsilon_{\mathrm{appr}}}
\newcommand{\epsobj}{\varepsilon_{\mathrm{obj}}}

\newcommand{\epsppr}{\varepsilon_{\mathrm{ppr}}}

\newcommand{\trans}{\mathsf{T}}

\DeclarePairedDelimiter{\norm}{\lVert}{\rVert}

\DeclarePairedDelimiterX{\inner}[2]{\langle}{\rangle}{#1, #2}

\newcommand{\softO}{\widetilde{\mathcal{O}}}

\newcommand{\va}{\bm{a}}
\newcommand{\vb}{\bm{b}}
\newcommand{\vc}{\bm{c}}
\newcommand{\vd}{\bm{d}}
\newcommand{\ve}{\bm{e}}
\newcommand{\vf}{\bm{f}}
\newcommand{\vg}{\bm{g}}
\newcommand{\vh}{\bm{h}}

\newcommand{\vn}{\bm{n}}

\newcommand{\vp}{\bm{p}}
\newcommand{\vq}{\bm{q}}
\newcommand{\vr}{\bm{r}}
\newcommand{\vs}{\bm{s}}
\newcommand{\vt}{\bm{t}}
\newcommand{\vu}{\bm{u}}
\newcommand{\vv}{\bm{v}}
\newcommand{\vw}{\bm{w}}
\newcommand{\vx}{\bm{x}}
\newcommand{\vy}{\bm{y}}
\newcommand{\vz}{\bm{z}}

\newcommand{\vdelta}{\bm{\delta}}

\newcommand{\vlambda}{\bm{\lambda}}

\newcommand{\vomega}{\bm{\omega}}

\newcommand{\vpi}{\bm{\pi}}

\newcommand{\vxi}{\bm{\xi}}
\newcommand{\vell}{\bm{\ell}}

\newcommand{\vzero}{\bm{0}}

\newcommand{\one}{\mathbf{1}}

\newcommand{\eunit}[1]{\bm{e}_{#1}}

\newcommand{\mA}{\bm{A}}
\newcommand{\mB}{\bm{B}}
\newcommand{\mC}{\bm{C}}
\newcommand{\mD}{\bm{D}}

\newcommand{\mH}{\bm{H}}
\newcommand{\mI}{\bm{I}}

\newcommand{\mK}{\bm{K}}
\newcommand{\mL}{\bm{L}}
\newcommand{\mM}{\bm{M}}

\newcommand{\mP}{\bm{P}}
\newcommand{\mQ}{\bm{Q}}
\newcommand{\mR}{\bm{R}}
\newcommand{\mS}{\bm{S}}

\newcommand{\mW}{\bm{W}}

\newcommand{\mDelta}{\bm{\Delta}}

\newcommand{\mPi}{\bm{\Pi}}

\newcommand{\gB}{\mathcal{B}}
\newcommand{\gC}{\mathcal{C}}

\newcommand{\gF}{\mathcal{F}}

\newcommand{\gI}{\mathcal{I}}

\newcommand{\gK}{\mathcal{K}}

\newcommand{\gS}{\mathcal{S}}

\newcommand{\gU}{\mathcal{U}}
\newcommand{\gV}{\mathcal{V}}

\newcommand{\sN}{\mathbb{N}}

\DeclareMathOperator{\diag}{diag}

\DeclareMathOperator{\nnz}{nnz}

\DeclareMathOperator{\supp}{supp}

\DeclareMathOperator{\vol}{vol}

\DeclareMathOperator{\Work}{Work}

\DeclareMathOperator*{\argmin}{arg\,min}

\title[Accelerated Local Algorithms for PageRank]{Accelerated Local Algorithms for\\Personalized and Regularized PageRank}
\author[Baojian Zhou]{%
  \Name{Baojian Zhou} \Email{\href{mailto:bjzhou@fudan.edu.cn}{bjzhou@fudan.edu.cn}}\\
  \addr School of Data Science, Fudan University\\
  Shanghai, China
}
\hypersetup{pdftitle={Accelerated Local Algorithms for Personalized and Regularized PageRank},
            pdfauthor={Baojian Zhou}}

\begin{document}
\maketitle
\thispagestyle{plain}

\begin{abstract}
Local PageRank algorithms seek sparse approximations with work independent
of graph size. We give a deterministic algorithm for regularized
personalized PageRank with additive objective accuracy $\epsobj$ in
$\widetilde{\mathcal{O}}(1/(\rho\sqrt\alpha))$ local work, where $\alpha$
is the lazy teleportation parameter and $\rho$ is the regularizer.
Accuracy enters only polylogarithmically. The bound charges discovery,
repeated neighborhood scans, numerical updates, certification, and output,
without graph-wide preprocessing or a supplied solution support.

The algorithm combines regularization continuation with accelerated
corrections constrained by a degree-scaled box and a mass cap. Two energies
for the same recurrence control objective convergence and the response
that activates coordinates. A selected-flow argument bounds cumulative
scanned volume, and a sparse threshold reporter realizes the bound.
We also specify a bounded-arithmetic implementation for rational inputs.

A second, randomized algorithm uses support-safe threshold batches.
A block-Cholesky and Chebyshev argument bounds their depth, and certified
SDD solves give expected work
$\widetilde{\mathcal{O}}(V_*\min\{k_*,\alpha^{-1/2}\})$, where $k_*$
and $V_*$ are the optimal support's cardinality and degree volume.
Both methods imply the corresponding accelerated degree-normalized PPR
approximation. The concurrent September 2026 preprint of Cui, Wei, and
Yang also attains the randomized work scale. Our principal distinction
is deterministic local acceleration with only polylogarithmic overhead
and no SDD oracle.
\end{abstract}

\section{Introduction}\label{sec:introduction}
Personalized PageRank (PPR) measures the relevance of graph vertices to a
seed through a geometrically stopped random walk. Local PPR algorithms
produce sparse approximations and support graph clustering and diffusion
without requiring access to the entire graph
\citep{andersen2006local,andersen2007using}. The central question is whether
acceleration can reduce total local work while preserving independence
from the graph size.

The classical push algorithm computes a sparse approximation in
$\mathcal{O}(1/(\alpha\epsappr))$ degree-weighted work, where $\alpha$ is the
teleportation parameter and $\epsappr$ is its normalized residual
tolerance. Positive residual mass pays for each neighborhood scan. Global
accelerated methods suggest a square-root improvement: the lazy symmetric
PageRank system has spectrum in $[\alpha,1]$, so standard accelerated methods need only
$\widetilde{\mathcal{O}}(1/\sqrt\alpha)$ iterations. These iteration bounds,
however, do not control the cost of accessing the graph. A local
running-time guarantee must charge both the vertices discovered and every
repeated update.

\paragraph{Regularized PageRank and the acceleration question.}
Weighted $\ell_1$ regularization gives a sparse optimization formulation
called regularized personalized PageRank (RPPR). For a seed vertex $v$ in
a graph with $n$ vertices, consider
\begin{equation}
 \min_{\vx\in\R^n}
 F_\rho(\vx):=
 \tfrac12\vx^\trans\mQ\vx-\alpha\eunit{v}^\trans\mD^{-1/2}\vx
           +\alpha\rho\norm{\mD^{1/2}\vx}_1,
 \qquad
 \mQ=\alpha\mI+\tfrac{1-\alpha}{2}\bm{\mathcal L}.
 \label{eq:intro-rppr-objective}
\end{equation}
Here $\mD$ is the diagonal matrix of original graph degrees,
$\bm{\mathcal L}$ is the symmetric normalized graph Laplacian, and $\rho>0$
is the regularization parameter. The optimum is nonnegative and has
support volume at most $1/\rho$ \citep{fountoulakis2019variational}.
Local proximal gradient exploits this structure in
$\widetilde{\mathcal{O}}(1/(\alpha\rho))$ work.
\citet{fountoulakis2022open} asked whether a local algorithm can achieve
\begin{equation*}
 \widetilde{\mathcal{O}}\!\left(\frac1{\rho\sqrt\alpha}\right)
\end{equation*}
work with only polylogarithmic dependence on the requested accuracy and no dependence on the ambient graph size.
We attain this target deterministically and also give a randomized algorithm.
A concurrent preprint by \citet{cui2026accelerating}, submitted to arXiv on
September~10, 2026, attains the same randomized work scale. Our principal
distinction is a deterministic bound with only polylogarithmic overhead,
proved without a deterministic nearly-linear SDD oracle.

\paragraph{The obstacle to local acceleration.}
Earlier approaches face two computational difficulties: expensive updates
outside the optimal support due to momentum and repeated solves on growing
active sets. Momentum-based methods can activate vertices outside the
optimal support, including vertices of very large degree; the fast iterative
shrinkage-thresholding algorithm (FISTA) can fail support confinement
\citep{fountoulakis2026complexity}. Meanwhile, methods that preserve safe
support expansion avoid these excursions by solving a sequence of restricted
problems, but their guarantees retain a factor for the number of support
discoveries \citep{martinezrubio2023accelerated,wei2026simple}. The concurrent
analysis of \citet{cui2026accelerating} overcomes this second difficulty
using a potential function for blocks of active-set expansions. The first
difficulty calls for a bound on cumulative update volume; the second calls
for a bound on the number of restricted systems that must be solved. Our two
algorithms address these difficulties through different forms of local
acceleration.

\paragraph{Our contributions.}
We consider point sources on finite connected simple unweighted undirected
graphs with at least two vertices, under the local-access model of
\cref{sec:problem-formulation}. The RPPR work bounds below concern the
nonzero regime $0<\rho<1/d_v$, where $d_v$ is the seed degree; for
$\rho\geq1/d_v$, the exact zero solution is returned in constant work.
\begin{itemize}[leftmargin=*]
\item \textbf{Deterministic local acceleration.}
Regularization continuation and constrained accelerated corrections
attain additive RPPR objective accuracy $\epsobj$ in
$\widetilde{\mathcal{O}}(1/(\rho\sqrt\alpha))$ work
(\cref{thm:deterministic-rppr}). The algorithm decreases regularization in
stages, accelerating each correction within a degree-scaled box under a
mass constraint. The total work bound includes the cumulative scanning
cost $\sum_j d_{u_j}$, where $u_j$ is the vertex scanned at the $j$th
neighborhood scan, as well as all remaining operations.

\item \textbf{Randomized acceleration with support adaptation.}
Certified threshold batching attains the same objective guarantee with
expected work
\begin{equation}
 \widetilde{\mathcal{O}}\!\left(
      V_*\min\left\{k_*,\frac1{\sqrt\alpha}\right\}\right)
 \leq\widetilde{\mathcal{O}}\!\left(
      \min\left\{\frac1{\rho^2},\frac1{\rho\sqrt\alpha}\right\}\right),
 \label{eq:intro-random-adaptive}
\end{equation}
where $k_*$ and $V_*$ are the cardinality and original-degree volume
of the optimal support (\cref{thm:randomized-rppr-main}). The minimum
retains the small-support regime of the active-set bound while adding an
accelerated bound on the number of batches. Restarting after a reported
failure gives a Las Vegas algorithm: every returned vector is certified.

\item \textbf{A PPR approximation consequence.}
Regularization bias and objective error together yield a sparse
approximation $\widehat\vpi$ to the seed's lazy PPR vector $\vpi$ with
$\norm{\mD^{-1}(\widehat\vpi-\vpi)}_\infty\leq\epsppr$ in
$\widetilde{\mathcal{O}}(1/(\epsppr\sqrt\alpha))$ deterministic work, or the
same expected work for the randomized method
(\cref{cor:main-ppr}). The conversion specifies separate regularization
and objective-accuracy budgets. We also give a certified conversion from
the vertex sets discovered by these algorithms to PPR approximations,
using principal systems with the original graph degrees.
\end{itemize}

\paragraph{Why the two methods remain local.}
The deterministic algorithm solves a sequence of decreasing regularization
levels. Each level starts from a safe sparse baseline and accelerates the
remaining correction inside a degree-scaled box with a mass constraint.
Two comparison energies for the same recurrence connect objective
convergence to local work. A signed-flow argument charges the cumulative
volume of all updated coordinates, including repeated visits and updates
outside the optimal support.

The randomized algorithm repeatedly solves on its current active set,
certifies the result, and admits boundary vertices whose violations exceed
a common threshold. The sign structure of the PageRank matrix certifies
that admitted vertices belong to the optimal support. A block-Cholesky and
Chebyshev argument bounds the number of batches by
$\widetilde{\mathcal{O}}(1/\sqrt\alpha)$, while every nonempty batch adds a
vertex and hence there are at most $k_*$ batches. Combining the smaller
bound with fully charged solves of local symmetric diagonally dominant
(SDD) systems yields
\cref{eq:intro-random-adaptive}.
The same support-adaptive soft bound also follows from the proof of
\citet[Theorem~1.3]{cui2026accelerating}; we do not claim that bound alone
as a distinction. Our randomized proof instead supplies a block-Cholesky
and Chebyshev explanation of the batch-depth phenomenon, with a
residual-certified Las Vegas implementation.

All bounds count graph access, local computation, certification, and
output under the model in \cref{sec:problem-formulation}. General seed
inputs and their additional costs are treated in \cref{app:seed-scope}.

\subsection{Related work}\label{sec:related-work}
\citet{fountoulakis2019variational} developed a monotone local version of
the iterative shrinkage-thresholding algorithm (ISTA) for regularized
PageRank. The regularization-path results of \citet{ha2021statistical}
provide a related foundation for continuation.

\paragraph{The local RPPR algorithms.}
\citet{martinezrubio2023accelerated} introduced Conjugate Directions for
PageRank (CDPR), an exact method, and Accelerated Sparse PageRank (ASPR),
an approximate method, on expanding supports.
\citet{wei2026simple} replaced the restricted numerical solve by a
nearly-linear SDD solve, obtaining a bound proportional to the number
of support discoveries. Their output also supplies the Andersen--Chung--Lang
(ACL) positive-residual certificate \citep{andersen2006local}.
For the comparison in \cref{tab:lineage}, write $\mathcal S_\rho^*:=\supp(\vx_\rho^*)$ for the RPPR optimum's
support and let $k_*=|\mathcal S_\rho^*|$,
$V_*=\vol(\mathcal S_\rho^*)$, and
$\widetilde V_*=\nnz(\mQ_{\mathcal S_\rho^*,\mathcal S_\rho^*})$.
The first volume uses original graph degrees; the last quantity counts
internal nonzeros. \Cref{app:related-work} gives the exact normalization
maps, theorem pointers, and output qualifications. Wei and Yang also discuss a deterministic SDD substitution with
additional factors; the specific soft bound, rather than the existence
of a deterministic active-set method, is the distinction here.

\citet{cui2026accelerating} subsequently prove the accelerated randomized
RPPR and ACL approximation bounds by controlling blocks of active-set
expansions. Their September~10 preprint directly overlaps our randomized
running-time result, including the support-adaptive soft bound obtained
from their proof. Its deterministic substitution incurs an additional
subpolynomial factor, rather than only polylogarithmic overhead.
Our deterministic continuation theorem is the main distinction;
the randomized proof here uses a different block-Cholesky decay argument.

\paragraph{Classical acceleration and evolving sets.}
The local iterative methods of
\citet{zhou2024iterative} and the accelerated evolving
set framework of \citet{huang2025accelerated} develop the link between
numerical acceleration and changing local supports. Their guarantees and
accuracy dependence differ from the graph-uniform RPPR bound studied here.
\citet{fountoulakis2026complexity} give conditional work bounds for FISTA
and construct $m$-edge stars on which it requires $\Omega(m)$ work while
ISTA's work is independent of $m$. For fixed $\alpha$, these instances have
$\rho=\Theta(1/m)$, so this lower bound does not rule out the
$\widetilde{\mathcal{O}}(1/(\rho\sqrt\alpha))$ target. Our two algorithms
attain this target through explicit control of cumulative update volume
or support expansion.

\Needspace{3\baselineskip}
\paragraph{Other approximation and access models.}
Some PPR estimators use the same output norm as our PPR corollary
\citep{wei2024absolute}; teleportation dependence, preprocessing, and access
assumptions also matter. Recent scalar PageRank estimation
\citep{thorup2026instance} and $\ell_1$-regularized resistance
\citep{li2026resistance} have different outputs or computational models.
\Cref{app:related-work} gives detailed comparisons with these and other
local diffusion, estimation, and obstacle methods.

\begin{table}[t]
\centering
\small
\renewcommand{\arraystretch}{1.15}
\caption{RPPR work in the canonical graph normalization. Additive
approximation refers to the RPPR objective gap. Soft bounds suppress the
logarithmic factors specified in the access model; probability qualifications
are shown explicitly. The comparison retains original-degree volume $V_*$
and internal nonzeros $\widetilde V_*$.}
\label{tab:lineage}
\begin{tabular}{@{}
 >{\raggedright\arraybackslash}p{0.35\textwidth}
 >{\raggedright\arraybackslash}p{0.23\textwidth}
 >{\raggedright\arraybackslash}p{\dimexpr0.42\textwidth-4\tabcolsep\relax}@{}}
\toprule
Method & Output & Work \\
\midrule
Local ISTA \citep{fountoulakis2019variational}
 & Additive approximation & $\softO(V_*/\alpha)$ \\
\midrule
CDPR \citep{martinezrubio2023accelerated}
 & Exact optimum & $\mathcal{O}(k_*^3+k_*V_*)$ \\
\midrule
ASPR \citep{martinezrubio2023accelerated}
 & Additive approximation & $\softO(k_*\widetilde V_*/\sqrt\alpha+k_*V_*)$ \\
\midrule
SDD active set \citep{wei2026simple}
 & Additive approximation and ACL certificate & $\softO(k_*V_*)$ with high probability \\
\midrule
Accelerated active set \citep{cui2026accelerating}
 & Additive approximation and ACL certificate & $\softO(V_*\min\{k_*,\alpha^{-1/2}\})$ with high probability \\
\midrule
\cref{alg:deterministic-rppr} (this paper)
 & Additive approximation, safe subsolution & $\softO(1/(\rho\sqrt\alpha))$ \\
\midrule
\cref{alg:threshold-batch-rppr} (this paper)
 & Additive approximation, safe support & $\softO(V_*\min\{k_*,\alpha^{-1/2}\})$ expected \\
\bottomrule
\end{tabular}
\end{table}

\paragraph{Organization.}
\Cref{sec:problem-formulation} fixes the problem, access model, and main
results. \Cref{sec:obstacle-active-sets} develops the common geometry and
local certificates. \Cref{sec:deterministic-method} presents deterministic
continuation, its two-energy analysis, and its sparse implementation.
\Cref{sec:randomized-method} gives threshold batching and proves its depth
and total-work bounds. \Cref{sec:two-stage-ppr} develops the PPR consequences
and discusses further questions. Detailed literature comparisons, bounded
arithmetic, and seed extensions appear in the appendices.

\section{Problem Formulation and Main Results}\label{sec:problem-formulation}
We use the lazy symmetric normalization of \citet[Section~2, equation~(3)]{fountoulakis2022open}.

\subsection{Graph and optimization problems}
\label{subsec:shared-source-aligned-problem}

We use boldface for vectors and matrices and regard vectors as columns.
Let \(\mathcal{G}=(\mathcal{V},\mathcal{E})\) be a finite connected simple
unweighted undirected graph with \(\mathcal{V}=[n]:=\{1,\ldots,n\}\),
\(n\geq2\), and adjacency matrix \(\bm{A}\in\{0,1\}^{n\times n}\).
For \(i\in\mathcal{V}\), define
\[
    \mathcal N(i):=\{j\in\mathcal{V}:\{i,j\}\in\mathcal{E}\},
    \qquad
    d_i:=|\mathcal N(i)|,
    \qquad
    \bm{D}:=\diag(d_1,\ldots,d_n).
\]
Degrees always refer to the original graph, including in restricted
systems. Connectedness gives \(d_i\geq1\), so \(\bm{D}^{-1/2}\) is well
defined. For \(\mathcal{S}\subseteq\mathcal{V}\), its neighborhood,
external boundary, and volume are
\begin{equation*}
    \mathcal N(\mathcal{S}):=\bigcup_{i\in\mathcal{S}}\mathcal N(i),
    \qquad
    \partial\mathcal{S}:=\mathcal N(\mathcal{S})\setminus\mathcal{S},
    \qquad
    \vol(\mathcal{S}):=\sum_{i\in\mathcal{S}}d_i.
\end{equation*}
Write \(\supp(\bm{x}):=\{i\in\mathcal{V}:x_i\neq0\}\).
We use \(\preceq,\succeq\) for Loewner order and \(\leq,\geq\) for
entrywise inequalities between vectors or matrices. The seed distribution
satisfies
\begin{equation*}
    \bm{s}\in\R_+^n,
    \qquad
    \inner{\one}{\bm{s}}=1.
\end{equation*}
Here \(\one\) is the all-ones vector. Our main results use a point source
\(\bm{s}=\bm{e}_v\), where \(\bm{e}_v\) is the \(v\)-th standard basis
vector; general seeds and their input and work costs are treated in
\cref{app:seed-scope}. For a teleportation parameter \(\alpha\in(0,1]\),
define
\begin{equation}
\begin{aligned}
    \bm{\mathcal L}&:=\bm{I}-\bm{D}^{-1/2}\bm{A}\bm{D}^{-1/2},\\
    \bm{Q}&:=\alpha\bm{I}+\frac{1-\alpha}{2}\bm{\mathcal L}
      =\frac{1+\alpha}{2}\bm{I}
       -\frac{1-\alpha}{2}\bm{D}^{-1/2}\bm{A}\bm{D}^{-1/2},\\
    \bm{b}&:=\alpha\bm{D}^{-1/2}\bm{s}.
\end{aligned}
    \label{eq:shared-pagerank-matrices}
\end{equation}
The normalized Laplacian satisfies
\(\bm{0}\preceq\bm{\mathcal L}\preceq2\bm{I}\), hence
\(\alpha\bm{I}\preceq\bm{Q}\preceq\bm{I}\). The PageRank quadratic is
\begin{equation}
    f(\bm{x}):=\frac12\inner{\bm{x}}{\bm{Q}\bm{x}}-\inner{\bm{b}}{\bm{x}},
    \qquad
    \nabla f(\bm{x})=\bm{Q}\bm{x}-\bm{b}.
    \label{eq:shared-pagerank-objective}
\end{equation}
Its unique minimizer and the corresponding lazy personalized PageRank
(PPR) vector are
\begin{equation*}
    \bm{x}_0^*:=\bm{Q}^{-1}\bm{b},
    \qquad
    \bm{\pi}:=\bm{D}^{1/2}\bm{x}_0^*.
\end{equation*}
For \(\rho>0\), the regularized PageRank (RPPR) objective and its unique
minimizer are defined by
\begin{equation}
    g_\rho(\bm{x}):=\alpha\rho\norm{\bm{D}^{1/2}\bm{x}}_1,
    \qquad
    F_\rho(\bm{x}):=f(\bm{x})+g_\rho(\bm{x}),
    \qquad
    \bm{x}_\rho^*:=\argmin_{\bm{x}\in\R^n}F_\rho(\bm{x}).
    \label{eq:shared-rppr-objective}
\end{equation}
We abbreviate the minimizer as \(\bm{x}^*\) when the objective is clear.
Writing \(\mathcal{S}_\rho^*:=\supp(\bm{x}_\rho^*)\), standard RPPR
properties give \(\bm{x}_\rho^*\geq\bm{0}\),
\(\vol(\mathcal{S}_\rho^*)\leq1/\rho\), and the coordinatewise
optimality conditions
\begin{equation}
    \nabla_i f(\bm{x}_\rho^*)
    \in
    \begin{cases}
        \{-\alpha\rho\sqrt{d_i}\}, & (\bm{x}_\rho^*)_i>0,\\
        [-\alpha\rho\sqrt{d_i},0], & (\bm{x}_\rho^*)_i=0.
    \end{cases}
    \label{eq:shared-rppr-kkt}
\end{equation}
For a point source $\vs=\eunit{v}$, $\vx_\rho^*\neq\vzero$ exactly when
\begin{equation}
 0<\rho<1/d_v.
 \label{eq:nonzero-rppr-regime}
\end{equation}
Indeed, the optimality condition at zero is
$\vb\leq\alpha\rho\mD^{1/2}\one$, which reduces to $\rho\geq1/d_v$.
A single degree query identifies the zero case. When $\alpha=1$, the
solution is explicit:
$\vx_\rho^*=((1-\rho d_v)_+/\sqrt{d_v})\eunit{v}$, where
$(t)_+:=\max\{t,0\}$.

\subsection{Accuracy guarantees}

The RPPR task takes $\rho>0$ and an additive tolerance
$\epsobj>0$, and asks for a sparse vector $\widehat\vx$ satisfying
\begin{equation}
 F_\rho(\widehat\vx)-F_\rho(\vx_\rho^*)\leq\epsobj.
 \label{eq:rppr-output-target}
\end{equation}
The output lists the nonzero coordinates of $\widehat\vx$ and their values;
unlisted coordinates are zero. For PPR, an output $\widehat\vx$ represents
$\widehat\vpi=\mD^{1/2}\widehat\vx$. The semantic target is
\begin{equation}
 \norm{\mD^{-1}(\widehat\vpi-\vpi)}_\infty
 =\norm{\mD^{-1/2}(\widehat\vx-\vx_0^*)}_\infty
 \leq\epsppr.
 \label{eq:semantic-ppr-target}
\end{equation}
Thus $\epsppr$ measures solution error per unit degree. The PageRank matrix $\mQ$ is a \emph{Stieltjes matrix}: symmetric positive
definite with nonpositive off-diagonal entries.
With $\vomega=\mD^{1/2}\one$, the identities
\begin{equation*}
 \mQ\vomega=\alpha\vomega,\qquad
 \mQ^{-1}\geq\vzero,
 \qquad |\mQ|\vomega=\vomega
\end{equation*}
will be used throughout; $|\mQ|$ denotes entrywise absolute values.
Since $\mA\one=\mD\one$, we have
$\mD^{-1/2}\mA\mD^{-1/2}\vomega=\vomega$, which gives the first identity
and, using the zero diagonal of $\mA$, the third. Also,
$\mI-\mQ\geq\vzero$ entrywise and
$\norm{\mI-\mQ}_2\leq1-\alpha<1$, so the Neumann series
$\mQ^{-1}=\sum_{k=0}^{\infty}(\mI-\mQ)^k$ converges and is entrywise
nonnegative.

\begin{proposition}[Regularization bias and objective-to-PPR conversion]
\label{prop:rppr-ppr-bridge}
For every $\rho>0$,
\begin{equation}
 \vzero\leq\vx_0^*-\vx_\rho^*\leq\rho\vomega.
 \label{eq:rppr-bias}
\end{equation}
Every nonnegative output satisfying \cref{eq:rppr-output-target} obeys
\begin{equation}
 \norm{\mD^{-1/2}(\widehat\vx-\vx_0^*)}_\infty
 \leq\rho+\sqrt{2\epsobj/\alpha}.
 \label{eq:objective-to-ppr}
\end{equation}
\end{proposition}
\begin{proof}
The KKT conditions give
$\vzero\leq\vb-\mQ\vx_\rho^*\leq\alpha\rho\vomega$.
Multiply by $\mQ^{-1}\geq\vzero$ and use
$\mQ^{-1}\vomega=\vomega/\alpha$ to obtain \cref{eq:rppr-bias}.
Strong convexity gives
$\norm{\widehat\vx-\vx_\rho^*}_2\leq\sqrt{2\epsobj/\alpha}$.
Since $d_i\geq1$, its degree-normalized maximum norm is no larger.
The triangle inequality proves \cref{eq:objective-to-ppr}.
\end{proof}
For example, to achieve the semantic PPR target \cref{eq:semantic-ppr-target}, it suffices to set
\begin{equation}
 \rho=\epsppr/2,\qquad \epsobj=\alpha\epsppr^2/8
 \label{eq:direct-ppr-budgets}
\end{equation}
is sufficient for \cref{eq:semantic-ppr-target}. No identification of the
two tolerances is involved. A second, residual-based sufficient condition is
\begin{equation}
 \norm{\mD^{-1/2}(\vb-\mQ\widehat\vx)}_\infty\leq\alpha\epsppr
 \quad\Longrightarrow\quad
 \norm{\mD^{-1/2}(\widehat\vx-\vx_0^*)}_\infty\leq\epsppr.
 \label{eq:ppr-residual-certificate}
\end{equation}
It follows from inverse positivity applied to the two coordinatewise
residual bounds. Our main RPPR algorithm uses its proved objective budget;
it does not need to scan the whole graph to check this sufficient PPR
certificate.

\subsection{Local access and computational cost}

The algorithms receive the seed label $v$, their numerical parameters, and
degree and adjacency-list access to $\mathcal G$. All other vertex labels
must be discovered through incident entries. No support information or
auxiliary graph data are supplied, and both algorithms use no graph-wide
preprocessing. Unless stated otherwise, we use an exact-real algebraic word model and
define the \emph{total work} of a complete execution by
\begin{equation}
 \Work:=N_{\mathrm{adj}}+N_{\mathrm{deg}}+N_{\mathrm{op}}+N_{\mathrm{out}}.
 \label{eq:total-work}
\end{equation}
Here $N_{\mathrm{adj}}$ counts adjacency-list entry inspections,
$N_{\mathrm{deg}}$ degree queries, $N_{\mathrm{op}}$ all remaining scalar
arithmetic operations, comparisons, random choices, and state reads or
writes, and $N_{\mathrm{out}}$ the words written to return the output. Each such operation
costs one unit. Composite operations, including dictionary updates,
local matrix and solver construction, threshold reporting, certification,
and cleanup, are charged for all their elementary operations.

Every inspection is counted, including repeated access to cached
incidences. Scanning a set $\mathcal S$ reads only its vertices' lists and
contributes $\vol(\mathcal S)$ to $N_{\mathrm{adj}}$. Full scans therefore
contribute $\sum_j d_{u_j}$, where $u_j$ is the vertex scanned at the
$j$th scan and repeated visits are included. For randomized algorithms,
$\Work$ includes failed trials and restarts; expected bounds control
$\E[\Work]$. Peak storage is reported separately.

An exact scalar occupies one word regardless of its encoding length.
For rational inputs, a density $f_i$ and original degree $d_i$ represent
$x_i=f_i\sqrt{d_i}$ and $\pi_i=d_if_i$, avoiding square-root computations.
The deterministic schedule uses dyadic halving, arithmetic, and comparisons.
\Cref{app:bounded-arithmetic} separately charges numerical and label
encoding costs for a specified bounded-arithmetic implementation;
floating-point stability requires separate analysis.

In soft bounds, $\widetilde{\mathcal{O}}$ hides only polylogarithms in inverse numerical
parameters, requested accuracy, explicitly supplied failure probability, and
the number of exposed records. It hides no ambient graph-size factor, no
unknown minimum sign margin, and no uncharged setup. In the nonzero regime
\cref{eq:nonzero-rppr-regime}, set
\begin{equation}
 L_{\mathrm{par}}:=\log\frac{2}{\alpha\rho\min\{1,\epsobj\}}.
 \label{eq:main-parameter-log}
\end{equation}

\subsection{Main results}

The next two theorems establish the RPPR guarantee, and
\cref{cor:main-ppr} gives its PPR consequence.
A nonnegative vector $\widehat\vx$ satisfying $\mQ\widehat\vx\leq\vb$
is called a \emph{subsolution}. Both algorithms can return subsolutions
lying coordinatewise below $\vx_\rho^*$ at the stated work bounds.

\begin{theorem}[Deterministic local acceleration]
\label{thm:deterministic-rppr}
For every graph and point source above, $0<\alpha\leq1$,
$0<\rho<1/d_v$, and $\epsobj>0$, \cref{alg:deterministic-rppr} returns
\begin{equation*}
 \vzero\leq\widehat\vx\leq\vx_\rho^*,\qquad
 F_\rho(\widehat\vx)-F_\rho(\vx_\rho^*)\leq\epsobj,
 \qquad\vol(\supp\widehat\vx)\leq1/\rho.
\end{equation*}
Its total work, defined in \cref{eq:total-work}, satisfies
\begin{equation}
 \Work=\mathcal{O}\!\left(\frac{L_{\mathrm{par}}}{\rho\sqrt\alpha}
           \log^2\!\left(2+\frac{L_{\mathrm{par}}}{\rho\sqrt\alpha}\right)\right)
 =\softO\!\left(\frac1{\rho\sqrt\alpha}\right),
 \label{eq:deterministic-main-work}
\end{equation}
and its peak storage is $\mathcal{O}(1+L_{\mathrm{par}}/(\rho\sqrt\alpha))$ words.
For $\rho\geq1/d_v$ it returns the exact zero solution in constant word work.
\end{theorem}

The proof occupies \cref{sec:deterministic-method}. The bound improves
the graph-independent worst-case work from
$\softO(1/(\alpha\rho))$ to $\softO(1/(\rho\sqrt\alpha))$.
The principal new estimate controls the cumulative kinetic volume of an
entire stage while allowing accelerated updates outside the optimal support.

\begin{theorem}[Randomized local acceleration, with support adaptation]
\label{thm:randomized-rppr-main}
Let $k_*=|\mathcal S_\rho^*|$ and $V_*=\vol(\mathcal S_\rho^*)$.
For $0<\alpha\leq1$, $0<\rho<1/d_v$, and $\epsobj>0$, there is a Las
Vegas local algorithm returning $\vzero\leq\widehat\vx\leq\vx_\rho^*$
satisfying \cref{eq:rppr-output-target}. It can also guarantee
$\vzero\leq\vb-\mQ\widehat\vx\leq2\alpha\rho\vomega$.
Its expected total work satisfies
\begin{equation}
 \E[\Work]=\softO\!\left(V_*\min\{k_*,\alpha^{-1/2}\}\right)
 \leq\softO\!\left(
       \min\left\{\frac1{\rho^2},\frac1{\rho\sqrt\alpha}\right\}\right).
 \label{eq:randomized-adaptive-work}
\end{equation}
Alternatively, a single capped run reports an explicit failure with
probability at most $\zeta\in(0,1)$, has the same expected work up to
failure-probability logarithms, and always satisfies the output guarantee
when it returns a vector. All support admissions are deterministically
certified.
\end{theorem}

The factors suppressed in \cref{eq:randomized-adaptive-work} include
logarithmic objective-accuracy dependence in the batch count and in the SDD
solves. The refined bound follows because every nonempty batch adds a new
vertex, while \cref{thm:threshold-batch-depth} bounds the number of batches
spectrally. Thus the randomized method retains the small-support bound as
well as the accelerated worst-case envelope: up to logarithmic factors,
the smaller of support size and the inverse square root of the teleportation
parameter determines how many restricted systems need to be solved.
Its graph and boundary records have size $\mathcal{O}(V_*)$; numerical workspace is
that of one SDD call on the current active set. Allocation is charged, so expected peak
space is in particular bounded by the displayed expected work.

\begin{corollary}[Local PPR at semantic accuracy]
\label{cor:main-ppr}
For $0<\epsppr<1$, degree-normalized PPR error $\epsppr$ can be attained
in $\softO(1/(\epsppr\sqrt\alpha))$ deterministic work, or in the same
expected work by a Las Vegas algorithm. The deterministic output obeys
$\vzero\leq\widehat\vx\leq\vx_0^*$, and both outputs have support volume
at most $2/\epsppr$. If $\epsppr\geq1$, zero suffices.
\end{corollary}
\begin{proof}
Use \cref{eq:direct-ppr-budgets} and
\cref{prop:rppr-ppr-bridge}. The logarithm of the requested objective
accuracy remains a sum of $\log(1/\alpha)$ and $\log(1/\epsppr)$.
For $\epsppr\geq1$, use $\vpi\geq\vzero$, $\one^\trans\vpi=1$, and
$d_i\geq1$.
\end{proof}

\paragraph{Bounded arithmetic.}
For rational numerical inputs, \cref{thm:bounded-deterministic-rppr} gives
an implementation with the same $\softO(1/(\rho\sqrt\alpha))$ local
operation bound using integers of controlled length. With schoolbook
integer arithmetic its bit cost is
$\softO(B^2/(\rho\sqrt\alpha))$, where $B$ explicitly includes the
input numerical encodings, encountered label and degree encodings, and
parameter and precision logarithms. This statement has no hidden dependence
on a minimum activation margin. It concerns a specified rounded algorithm,
not an arbitrary numerical implementation of the exact recurrence.

\section{Stieltjes Geometry and Local Certificates}\label{sec:obstacle-active-sets}
This section collects structural properties and local certificates used by
the two algorithms.
Monotonicity and mass bounds justify deterministic continuation in
\cref{sec:deterministic-method}. Safe active-set expansion and a local
objective certificate support the randomized method in
\cref{sec:randomized-method}. An electrical-flow interpretation is given
in \cref{app:grounded-flow-dual}.

\subsection{Nonnegative formulation and monotonicity}

On nonnegative vectors, the weighted $\ell_1$ penalty is linear. Absorb it
into the load by defining
\begin{equation}
 \vc_\rho:=\vb-\alpha\rho\vomega,
 \qquad
 \phi_\rho(\vx):=\frac12\vx^\trans\mQ\vx-\vc_\rho^\trans\vx.
 \label{eq:shifted-obstacle}
\end{equation}
Thus $F_\rho(\vx)=\phi_\rho(\vx)$ for $\vx\geq\vzero$.
Only the seed coordinate of $\vc_\rho$ can be positive, which will imply
connectivity of the optimal support.
Principal submatrices inherit the spectral bounds and off-diagonal signs
of $\mQ$, so the same Neumann-series argument gives nonnegative principal
inverses.

The reformulation and KKT conditions below follow
\citet[Section~2]{martinezrubio2023accelerated}, building on the
nonnegativity analysis of \citet[Theorem~1]{fountoulakis2019variational}.
The latter paper also discusses connected support for connected seed sets
(p.~567). We include a short proof in our point-source setting for
completeness.

\begin{lemma}[Nonnegative reformulation and connected support]
\label{thm:obstacle-reduction}
The RPPR problem has the same unique optimum and optimum value as
\begin{equation}
 \min_{\vx\geq\vzero}\phi_\rho(\vx).
 \label{eq:obstacle-problem}
\end{equation}
Its KKT system is
\begin{equation}
 \vx\geq\vzero,
 \qquad
 \vw:=\mQ\vx-\vc_\rho\geq\vzero,
 \qquad
 x_iw_i=0\quad(i\in[n]).
 \label{eq:obstacle-kkt}
\end{equation}
In the nontrivial regime \eqref{eq:nonzero-rppr-regime}, the support
$\mathcal S_\rho^*$ is connected and contains $v$.
\end{lemma}

\begin{proof}
The signs $Q_{ij}\leq0$ for $i\neq j$ and $\vb\geq\vzero$ give
$F_\rho(|\vx|)\leq F_\rho(\vx)$. Uniqueness therefore implies
$\vx_\rho^*\geq\vzero$. Equality of the objectives on the nonnegative
orthant proves the reformulation. Its optimality conditions are
\cref{eq:obstacle-kkt}, and strong convexity gives uniqueness.

At the seed,
\[
 (\vc_\rho)_v=\frac{\alpha}{\sqrt{d_v}}(1-\rho d_v)>0.
\]
If $(\vx_\rho^*)_v=0$, its slack is strictly negative, contradicting
\cref{eq:obstacle-kkt}.  Thus $v\in\mathcal S_\rho^*$.  If an active connected
component $\gC$ did not contain $v$, no support edge would join $\gC$ to the
other active components.  Stationarity on $\gC$ would give
$\mQ_{\gC\gC}(\vx_\rho^*)_\gC=(\vc_\rho)_\gC<\vzero$.  Multiplication by
$\mQ_{\gC\gC}^{-1}\geq\vzero$ contradicts positivity on $\gC$.
\end{proof}

\begin{lemma}[Order, mass, and support volume]
\label{lem:rppr-order-mass}
If $\vz\geq\vzero$ and $\mQ\vz\geq\vc_\rho$, then
$\vx_\rho^*\leq\vz$. Thus $\vx_\rho^*$ is the least nonnegative vector
satisfying this shifted-load inequality. Consequently, for $0<r\leq r'$,
\begin{equation}
 \vzero\leq\vx_{r'}^*\leq\vx_r^*\leq\vx_0^*,
 \qquad \vomega^\trans\vx_r^*+r\vol(\mathcal S_r^*)\leq1.
 \label{eq:rppr-order-mass}
\end{equation}
In particular $\vol(\mathcal S_r^*)\leq1/r$.
\end{lemma}
\begin{proof}
For the first assertion let $\gI=\{i:(\vx_\rho^*)_i>z_i\}$ and
$\ve=\vx_\rho^*-\vz$. If $\gI$ is nonempty, stationarity on
its positive optimum coordinates gives
$\mQ_{\gI\gI}\ve_{\gI}
 \leq-\mQ_{\gI\gI^c}\ve_{\gI^c}\leq\vzero$.
The last inequality uses $\ve_{\gI^c}\leq\vzero$ and Stieltjes signs.
Multiplication by the nonnegative principal inverse contradicts
$\ve_{\gI}>\vzero$. Since
$\mQ\vx_r^*\geq\vc_r\geq\vc_{r'}$ and
$\mQ\vx_0^*=\vb\geq\vc_r$, the first assertion proves the order.

The source deficit $\vb-\mQ\vx_r^*$ is nonnegative: it equals
$\alpha r\omega_i$ on the positive support, and at a zero coordinate
the off-diagonal signs give $(\mQ\vx_r^*)_i\leq0\leq b_i$.
Its weighted mass is
$\alpha(1-\vomega^\trans\vx_r^*)$, by
$\mQ\vomega=\alpha\vomega$ and $\vomega^\trans\vb=\alpha$.
Its contribution on the positive support alone is
$\alpha r\vol(\mathcal S_r^*)$, proving the mass inequality.
\end{proof}

\subsection{Safe active-set expansion}
\label{subsec:reachable-faces}

We now describe the exact active-set steps underlying the randomized
method. Assume the nonzero point-source regime in
\cref{eq:nonzero-rppr-regime}. For an active set $\gU\subseteq[n]$, define
the restricted linear-system solution, extended by zero, and its slack by
\begin{equation*}
 x_i^\gU=
 \begin{cases}
  (\mQ_{\gU\gU}^{-1}(\vc_\rho)_\gU)_i,&i\in \gU,\\
  0,&i\notin \gU,
 \end{cases}
 \qquad
 \vw^\gU:=\mQ\vx^\gU-\vc_\rho.
\end{equation*}
Set $\vx^\emptyset=\vzero$. For an arbitrary $\gU$, the restricted
solution need not be nonnegative. An active set is \emph{reachable} if it
is obtained from $\{v\}$ by repeatedly adding one or more coordinates
whose exact current slack is negative. The next theorem specializes the
restricted-minimizer geometry of
\citet[Proposition~2 and Section~3]{martinezrubio2023accelerated}.
We give a direct proof using Stieltjes inverses and Schur complements.

\begin{theorem}[Safe active-set expansion]
\label{thm:safe-batched-pivots}
For every reachable active set $\gU$,
\begin{equation*}
 \gU\subseteq\mathcal S_\rho^*,
 \qquad
 \vzero<\vx_\gU^\gU\leq(\vx_\rho^*)_\gU.
\end{equation*}
Every $j\notin \gU$ with $w_j^\gU<0$ lies in
$\mathcal S_\rho^*\setminus \gU$.  If $\gB$ is any nonempty collection of such
coordinates, then $\gU\cup \gB$ is reachable and
\begin{equation*}
 \vx_{\gU\cup \gB}^{\gU\cup \gB}>\vzero,
 \qquad
 \vx_\gU^{\gU\cup \gB}\geq\vx_\gU^\gU.
\end{equation*}
If no outside slack is negative, $\vx^\gU=\vx_\rho^*$.
\end{theorem}

\begin{proof}
The initial seed slack is negative, so the first scalar solution is
positive. Suppose the claim holds at $\gU$. Restricted stationarity and the
optimal equations give
\begin{equation*}
 \mQ_{\gU\gU}\bigl((\vx_\rho^*)_\gU-\vx_\gU^\gU\bigr)
 =-\mQ_{\gU,\mathcal S_\rho^*\setminus \gU}
   (\vx_\rho^*)_{\mathcal S_\rho^*\setminus \gU}\geq\vzero.
\end{equation*}
Inverse positivity proves the coordinatewise comparison.  If
$j\notin\mathcal S_\rho^*$, then
\begin{align*}
 w_j^\gU-w_j^*
 &=\mQ_{j\gU}\bigl(\vx_\gU^\gU-(\vx_\rho^*)_\gU\bigr)
   -\mQ_{j,\mathcal S_\rho^*\setminus \gU}
     (\vx_\rho^*)_{\mathcal S_\rho^*\setminus \gU}\geq0.
\end{align*}
Hence a negative violator cannot be a false support coordinate.

For $\gU'=\gU\cup \gB$, block elimination gives
\begin{equation*}
 \mS_\gB\vx_\gB^{\gU'}=-\vw_\gB^\gU,
 \qquad
 \mS_\gB:=\mQ_{\gB\gB}-\mQ_{\gB\gU}\mQ_{\gU\gU}^{-1}\mQ_{\gU\gB}.
\end{equation*}
Moreover,
\begin{equation*}
 \vx_\gU^{\gU'}-\vx_\gU^\gU=-\mQ_{\gU\gU}^{-1}\mQ_{\gU\gB}\vx_\gB^{\gU'}.
\end{equation*}
The Schur complement is Stieltjes and has a nonnegative inverse with
positive diagonal.  Since
$-\vw_\gB^\gU>\vzero$, the new coordinates are positive; the second display and
the off-diagonal sign make every old correction nonnegative.  This closes the
induction.  If no outside slack is negative, inside stationarity, positivity,
and outside nonnegativity satisfy the global KKT system, whose solution is
unique.
\end{proof}

Thus $\vx^\gU$ minimizes $\phi_\rho$ over nonnegative vectors supported
on a reachable $\gU$. We also call this restricted computation a
\emph{face solve}.

For any $\vx\geq\vzero$, set $\vw=\mQ\vx-\vc_\rho$ and
$y_i=x_i/\sqrt{d_i}$. At a coordinate $j\neq v$ with $x_j=0$,
\begin{equation}
 \frac{w_j}{\sqrt{d_j}}
 =\alpha\rho-\frac{1-\alpha}{2d_j}
   \sum_{i\in\mathcal N(j)\cap\supp(\vx)}y_i.
 \label{eq:boundary-slack}
\end{equation}
This expression makes support discovery local.

\begin{corollary}[Boundary-only discovery]
\label{cor:boundary-only-discovery}
At a reachable active set $\gU$, every negative outside slack lies in
$\partial\gU$. Exposing each admitted vertex's adjacency list once costs
$\vol(\gU)$ entry inspections and creates at most $\vol(\gU)$ boundary
incidences. No adjacency list outside $\gU$ is needed to evaluate the
boundary slacks.
\end{corollary}

\begin{proof}
If $j\notin \gU\cup\partial \gU$ and $j\neq v$, sparsity gives
$w_j^\gU=-(\vc_\rho)_j=\alpha\rho\sqrt{d_j}>0$.  The seed is already active.
The incidence count charges each exposed edge to its endpoint in $\gU$.
Boundary degrees can be queried without opening their adjacency lists.
Any later rescan is charged again under \cref{eq:total-work}.
\end{proof}

\subsection{Local objective certification}

We next turn slack values into a stopping test for
\cref{eq:rppr-output-target}. Let $I_{\R_+^n}$ be the indicator function,
equal to zero on $\R_+^n$ and $+\infty$ elsewhere.
For feasible $\vx\geq\vzero$, let $\vw=\mQ\vx-\vc_\rho$ and define
\begin{equation}
 g_i(\vx)=
 \begin{cases}
  w_i,&x_i>0,\\
  \min\{w_i,0\},&x_i=0.
 \end{cases}
 \label{eq:minimal-kkt-subgradient}
\end{equation}

\begin{lemma}[KKT subgradient certificate]
\label{lem:kkt-gap-certificate}
The vector $\vg(\vx)$ has the smallest Euclidean norm among the
subgradients of $\phi_\rho+I_{\R_+^n}$ at $\vx$, and
\begin{equation*}
 F_\rho(\vx)-F_\rho(\vx_\rho^*)
 \leq\frac{\norm{\vg(\vx)}_2^2}{2\alpha}.
\end{equation*}
If $\vx$ is supported on a known set $\gU$ containing $v$, the certificate
and its squared norm can be computed with $\vol(\gU)$ adjacency-entry
inspections and $\softO(\vol(\gU))$ total work under
\cref{eq:total-work}.
\end{lemma}

\begin{proof}
At $x_i>0$ the normal cone is zero.  At $x_i=0$ it is $(-\infty,0]$, so
\cref{eq:minimal-kkt-subgradient} is the closest point to zero in the
coordinate subdifferential.  Strong convexity gives, for every subgradient
$\vg$,
\[
 (\phi_\rho+I_{\R_+^n})(\vy)\geq(\phi_\rho+I_{\R_+^n})(\vx)
 +\vg^\trans(\vy-\vx)+\frac\alpha2\norm{\vy-\vx}_2^2.
\]
Taking $\vy=\vx_\rho^*$ and bounding the last two terms below by
$-\norm{\vg}_2^2/(2\alpha)$ proves the objective bound.

One scan of the rows in $\gU$ accumulates $\mQ\vx$ on
$\gU\cup\partial\gU$. Every farther coordinate has positive slack by
\cref{eq:boundary-slack}, so its certificate entry is zero.
There are at most $|\gU|+\vol(\gU)\leq2\vol(\gU)$ records.
Degree queries, diagonal and source contributions, clipping, the squared
norm, and output take $\mathcal{O}(\vol(\gU))$ elementary operations.
Accumulating entries in a balanced search tree adds a factor
$\mathcal{O}(\log(2+\vol(\gU)))$, giving the stated total work.
Degree densities permit the same computation without square roots:
the squared norm is a degree-weighted sum of squared certificate densities.
\end{proof}

In particular, $\norm{\vg(\vx)}_2^2\leq2\alpha\epsobj$ certifies the
requested RPPR accuracy using only the active set and its boundary.

\section{Deterministic Accelerated Continuation}\label{sec:deterministic-method}
\subsection{Algorithm and stage invariants}\label{sec:deterministic-algorithm}
The deterministic method uses a changing regularization parameter to make
each remaining correction diffuse. It then accelerates that correction on
an explicit convex set. The set contains the unknown correction optimum,
but its definition requires only the current source and the degrees.
\Cref{alg:deterministic-rppr} gives the complete procedure after the stage
recurrence and its repair map are defined.

\subsubsection{A diffuse-source stage}

Write
\begin{equation*}
 \vomega:=\mD^{1/2}\one,
 \qquad \lambda_r:=\alpha r.
\end{equation*}
Here $r\geq\rho$ is the stage regularizer; $\rho$ remains the final target.
At the start of a stage we have a sparse baseline $\overline\vx$ satisfying
\begin{equation*}
 \vzero\leq\overline\vx\leq\vx_r^*,\qquad
 \vh:=\vb-\mQ\overline\vx,\qquad
 \vzero\leq\vh\leq4\lambda_r\vomega.
\end{equation*}
The vector $\vh$ is a residual source, distinct from the seed distribution
$\vs$. Its weighted mass is
\begin{equation*}
 m_h:=\vomega^\trans\vh
      =\alpha(1-\vomega^\trans\overline\vx),\qquad
 m_r:=m_h/\alpha\leq1.
\end{equation*}
We use two comparison vectors only in the analysis:
\begin{equation*}
 \vxi^*:=\vx_r^*-\overline\vx,
 \qquad \vt:=\mQ^{-1}\vh=\vx_0^*-\overline\vx.
\end{equation*}
The vector $\vxi^*$ minimizes the correction objective below; $\vt$
provides the comparison point for the second energy controlling local work.
The algorithm computes neither inverse nor optimum.

\begin{lemma}[An explicit correction domain]
\label{lem:det-correction-domain}
Both $\vxi^*$ and $\vt$ belong to
\begin{equation*}
 \gK_r:=\{\vu:\vzero\leq\vu\leq4r\vomega,
                    \ \vomega^\trans\vu\leq m_r\}.
\end{equation*}
In fact $\vzero\leq\vxi^*\leq\vt\leq4r\vomega$ and
$\vomega^\trans\vt=m_r$.
\end{lemma}
\begin{proof}
The order follows from \cref{eq:rppr-bias}, the baseline invariant,
inverse positivity, and
$\mQ^{-1}\vomega=\vomega/\alpha$:
$\vt\leq4\alpha r\mQ^{-1}\vomega=4r\vomega$.
Multiplying $\mQ\vt=\vh$ by $\vomega^\trans$ proves the mass identity.
\end{proof}

The correction objective is the smooth quadratic
\begin{equation*}
 J_r(\vxi):=\tfrac12\vxi^\trans\mQ\vxi
                 -(\vh-\lambda_r\vomega)^\trans\vxi.
\end{equation*}
On nonnegative corrections it satisfies
$J_r(\vxi)=F_r(\overline\vx+\vxi)-F_r(\overline\vx)$.
Thus $\vxi^*$ minimizes $J_r$ on $\gK_r$, and correction gaps are
exactly the desired gaps for the full candidate.

Choose a dyadic $\theta$ by halving $1/2$ until $\theta^2\leq\alpha$, and set
\begin{equation}
 \mu_{\mathrm{c}}:=\theta^2,\qquad \chi:=1-\theta,
 \qquad \alpha/4\leq\mu_{\mathrm{c}}\leq\alpha.
 \label{eq:det-acceleration-parameters}
\end{equation}
The parameter $\mu_{\mathrm{c}}$ is the curvature used by the accelerated
recurrence. Start with $\vxi_0=\vz_0=\vzero$ and iterate
\begin{equation}
 \begin{aligned}
 \vy_k&=\frac{\vxi_k+\theta\vz_k}{1+\theta},\\
 \vq_k&=\chi\vz_k+\theta\vy_k
       -\frac{\mQ\vy_k-\vh+\lambda_r\vomega}{\theta},\\
 \vz_{k+1}&=\operatorname{Proj}_{\gK_r}(\vq_k),\\
 \vxi_{k+1}&=\chi\vxi_k+\theta\vz_{k+1}.
 \end{aligned}
 \label{eq:det-iteration}
\end{equation}
Projection is Euclidean. We call $\vz_k$ the \emph{kinetic vector} to
distinguish the coordinates updated by the projection from the longer-lived
primal correction $\vxi_k$. Projection keeps $\vz_k$ in $\gK_r$, and
convexity then keeps $\vxi_k$ there.

\paragraph{Why both constraints are present.}
The upper box controls the sign of a projection normal at a saturated
coordinate. The mass cap controls the common normal introduced by the
weighted sum constraint. Together they allow the \emph{same} Euclidean
projection to satisfy a comparison inequality in the $\mQ$ metric
(\cref{lem:det-sector}). That second comparison is the additional ingredient
needed to bound local work. Neither constraint requires knowing the support
of $\vx_r^*$, and the kinetic support need not be contained in it.

\paragraph{Explicit projection.}
There is a scalar $\gamma\geq0$ such that
\begin{equation}
 \frac{z_{k+1,i}}{\omega_i}
 =\min\{4r,(q_{k,i}/\omega_i-\gamma)_+\}.
 \label{eq:det-waterfill}
\end{equation}
If the right side at $\gamma=0$ has weighted mass at most $m_r$, use zero;
otherwise choose $\gamma$ so that $\vomega^\trans\vz_{k+1}=m_r$.
These are exactly the projection KKT conditions. The finite threshold search
and the cost of reporting its positive coordinates are proved in
\cref{sec:deterministic-implementation}.

\subsubsection{A certified stage handoff}

A small objective gap does not by itself make the candidate a subsolution.
We use one projected-gradient step and a downward density correction at the
\emph{end} of a stage. Intermediate iterations remain precisely
\cref{eq:det-iteration}.

\begin{lemma}[Safe repair]
\label{lem:det-repair}
Let $0<\delta\leq\lambda_r/2$. Suppose a nonnegative candidate
$\widetilde\vx$ has $F_r(\widetilde\vx)-F_r(\vx_r^*)\leq\tau$, where
$\tau\leq\alpha\delta^2/8$. Let the old baseline satisfy
$\vzero\leq\overline\vx\leq\vx_r^*$ and
$\mQ\overline\vx\leq\vb$; it may be zero. Define
\begin{equation}
 \vp=[\widetilde\vx-(\mQ\widetilde\vx-\vb+\lambda_r\vomega)]_+,
 \qquad \vu=[\vp-\delta\vomega]_+.
 \label{eq:det-repair-map}
\end{equation}
Then $\overline\vx^+:=\max\{\overline\vx,\vu\}$ satisfies
\begin{equation}
 \begin{gathered}
 \vzero\leq\overline\vx\leq\overline\vx^+\leq\vx_r^*,\qquad
 \vzero\leq\vb-\mQ\overline\vx^+\leq2\lambda_r\vomega,\\
 \vol(\supp\overline\vx^+)\leq1/r,\qquad
 F_r(\overline\vx^+)-F_r(\vx_r^*)\leq2\delta^2/r.
 \end{gathered}
 \label{eq:det-repaired-invariant}
\end{equation}
The coordinatewise maximum can be omitted if monotone baselines are not
needed.
\end{lemma}

\subsubsection{The complete deterministic algorithm}

At each stage, \cref{alg:deterministic-rppr} halves the regularization
level, accelerates the remaining correction, and repairs the result to
obtain the next safe baseline. At the final level it chooses the repair
accuracy to meet the requested objective tolerance.

\begin{algorithm}[htbp]
\caption{Deterministic accelerated local RPPR}
\label{alg:deterministic-rppr}
\begin{algorithmic}[1]
\REQUIRE Degree and adjacency-list access, seed $v$, $\alpha,\rho,\epsobj>0$
\STATE Query $d_v$. If $\rho d_v\geq1$ or $\epsobj\geq\alpha/2$, return $\vzero$.
\STATE If $\alpha=1$, return the single coordinate
       $\widehat x_v=(1-\rho d_v)/\sqrt{d_v}$.
\STATE Choose $\theta,\mu_{\mathrm{c}},\chi$ as in
       \cref{eq:det-acceleration-parameters}; set
       $r_{\mathrm{old}}=1/d_v$ and $\overline\vx=\vzero$.
\REPEAT
  \STATE $r\gets\max\{\rho,r_{\mathrm{old}}/2\}$; set $\delta\gets\alpha r/2$.
  \IF{$r=\rho$}
    \STATE Halve $\delta$ until $2\delta^2/r\leq\epsobj$.
  \ENDIF
  \STATE Assemble $\vh=\vb-\mQ\overline\vx$ locally; compute $m_r$.
  \STATE Set $\vxi=\vz=\vzero$, $a=1$, and $\tau=\alpha\delta^2/8$.
  \WHILE{$a>\tau$}
    \STATE Apply \cref{eq:det-iteration} using the sparse threshold reporter.
    \STATE $a\gets\chi a$.
  \ENDWHILE
  \STATE Materialize $\widetilde\vx=\overline\vx+\vxi$; apply
         \cref{eq:det-repair-map} and set
         $\overline\vx\gets\max\{\overline\vx,\vu\}$.
  \STATE $r_{\mathrm{old}}\gets r$.
\UNTIL{$r=\rho$}
\RETURN The sparse coordinate list of $\overline\vx$.
\end{algorithmic}
\end{algorithm}
\FloatBarrier

The scalar $a$ is a computable upper bound on the current comparison
energy, as proved in \cref{sec:deterministic-analysis}. In particular, the stopping test uses no optimum,
global gradient, or unknown complementarity margin.
The first stage has $r\geq1/(2d_v)$, so its source obeys
$\vb\leq2\alpha r\vomega$. At subsequent stages,
\cref{eq:det-repaired-invariant} and $r\geq r_{\mathrm{old}}/2$ imply
$\vh\leq4\alpha r\vomega$. Monotonicity of $\vx_r^*$ as $r$ decreases
preserves baseline safety. These facts close the continuation induction.

\begin{proof}[Proof of \cref{lem:det-repair}]
Write $\vdelta=\widetilde\vx-\vp$. Projection optimality and
$\mQ\preceq\mI$ give
\begin{equation}
 F_r(\widetilde\vx)-F_r(\vp)\geq\tfrac12\norm{\vdelta}_2^2,
 \qquad (\mI-\mQ)\vdelta\in
 \partial(\phi_r+I_{\R_+^n})(\vp).
 \label{eq:det-pg-certificate}
\end{equation}
For the second assertion, the projection normal is
$\vdelta-\nabla\phi_r(\widetilde\vx)$; adding
$\nabla\phi_r(\vp)$ gives $(\mI-\mQ)\vdelta$.
Strong convexity and descent imply
\[
 \norm{\vp-\vx_r^*}_2\leq\sqrt{2\tau/\alpha}\leq\delta/2,
 \qquad \norm{(\mI-\mQ)\vdelta}_2\leq\sqrt{2\tau}\leq\delta/2.
\]
At every positive coordinate of $\vp$ the normal cone vanishes. Hence
$(\mQ\vp-\vb)_i\leq\delta/2-\lambda_r\omega_i\leq0$.
At a zero coordinate, $(\mQ\vp-\vb)_i\leq0$ follows directly from
the Stieltjes signs.
Clipping gives $\vzero\leq\vu\leq\vx_r^*$ and
$\vzero\leq\vx_r^*-\vu\leq2\delta\vomega$, using $\omega_i\geq1$.
If $u_i>0$, then $u_i=p_i-\delta\omega_i$ and
$\vu\geq\vp-\delta\vomega$. Therefore
$(\mQ\vu)_i\leq(\mQ\vp)_i-\alpha\delta\omega_i\leq b_i$;
zero coordinates again follow by the off-diagonal sign.

The maximum of two nonnegative subsolutions is a subsolution: on row $i$,
choose the vector that attains the maximum there; the diagonal term is
unchanged and every other coordinate can only increase. Both vectors are
below $\vx_r^*$, and the maximum retains the displayed error bound.
The identity $|\mQ|\vomega=\vomega$ then gives
\[
 \vb-\mQ\overline\vx^+
 \leq\vb-\mQ\vx_r^*+2\delta\vomega
 \leq(\lambda_r+2\delta)\vomega\leq2\lambda_r\vomega.
\]
Support containment gives the volume bound and makes the linear KKT term
in the objective expansion zero. Thus the gap is
$\tfrac12\norm{\overline\vx^+-\vx_r^*}_{\mQ}^2
 \leq\tfrac12(2\delta)^2\vol(\mathcal S_r^*)\leq2\delta^2/r$.
\end{proof}

\subsection{Two energies and cumulative local work}\label{sec:deterministic-analysis}
Fix one stage and suppress its regularizer in the iteration indices.
The analysis has three steps. An ordinary accelerated energy bounds the
objective gap. The same comparison argument in a different metric bounds
the squared response of the actual projected trajectory. Finally, a
signed-flow inequality converts that response bound into cumulative degree
volume. The second and third steps account for locality.

\subsubsection{One comparison identity, two metrics}

\begin{lemma}[Accelerated comparison]
\label{lem:det-comparison}
In a Hilbert inner product, let $\varphi$ be one-smooth and
$\mu_{\mathrm{c}}$-strongly convex, where $\mu_{\mathrm{c}}=\theta^2$ and
$0<\theta<1$. Write $\chi=1-\theta$ and set
\[
 \vy=\frac{\vx+\theta\vz}{1+\theta},\quad
 \vg=\nabla\varphi(\vy),\quad
 \vq=\chi\vz+\theta\vy-\vg/\theta,\quad
 \vx^+=\chi\vx+\theta\vp.
\]
Here the gradient and norms use the Hilbert metric. If a comparison vector
$\vu$ satisfies
\begin{equation}
 \inner{\vp-\vu}{\vq-\vp}\geq0,
 \label{eq:det-comparison-sector}
\end{equation}
then
\begin{align*}
 &\varphi(\vx^+)-\varphi(\vu)
       +\frac{\mu_{\mathrm{c}}}2\norm{\vp-\vu}^2\\
 &\quad\leq\chi\left(\varphi(\vx)-\varphi(\vu)
       +\frac{\mu_{\mathrm{c}}}2\norm{\vz-\vu}^2\right)
       -\frac{\mu_{\mathrm{c}}\theta(1-\mu_{\mathrm{c}})}2
        \norm{\vz-\vy}^2.
\end{align*}
The vector $\vu$ need not minimize $\varphi$, and the quantity in
parentheses need not be nonnegative.
\end{lemma}
\begin{proof}
The update satisfies $\vx^+=\vy-\vg+\theta(\vp-\vq)$.
Smoothness and expansion of its squared displacement imply
\[
 \varphi(\vx^+)\leq\varphi(\vy)-\tfrac12\norm{\vg}^2
                     +\tfrac{\mu_{\mathrm{c}}}2\norm{\vp-\vq}^2.
\]
The sector inequality gives
$\norm{\vp-\vu}^2\leq\norm{\vq-\vu}^2-\norm{\vp-\vq}^2$.
Add the two inequalities and expand
$\vq-\vu=\chi(\vz-\vu)+\theta(\vy-\vu)-\vg/\theta$.
Strong convexity at $\vx$ and $\vu$, with weights $\chi$ and $\theta$,
respectively, gives
\begin{align*}
 \chi\varphi(\vx)+\theta\varphi(\vu)
 &\geq\varphi(\vy)
    +\inner{\vg}{\chi(\vx-\vy)+\theta(\vu-\vy)}\\
 &\quad+\tfrac{\mu_{\mathrm{c}}}2
       \bigl(\chi\norm{\vx-\vy}^2+\theta\norm{\vu-\vy}^2\bigr).
\end{align*}
The gradient terms cancel because
\[
 \chi(\vx-\vy)+\theta(\vu-\vy)
 =-\theta\{\chi(\vz-\vu)+\theta(\vy-\vu)\}.
\]
Use $\vx-\vy=\theta(\vy-\vz)$ and the convex-combination identity
\[
 \norm{\chi(\vz-\vu)+\theta(\vy-\vu)}^2
 =\chi\norm{\vz-\vu}^2+\theta\norm{\vy-\vu}^2
       -\chi\theta\norm{\vz-\vy}^2.
\]
The remaining decrease is
$-\mu_{\mathrm{c}}\chi\theta(1+\theta)\norm{\vz-\vy}^2/2$.
Since $\chi(1+\theta)=1-\mu_{\mathrm{c}}$, this proves the claim.
\end{proof}

\Needspace{5\baselineskip}
\begin{proposition}[Accelerated objective convergence]
\label{prop:det-stage-rate}
The exact stage recurrence satisfies
\begin{equation}
 E_{k+1}\leq\chi E_k,\qquad
 E_k:=J_r(\vxi_k)-J_r(\vxi^*)
          +\tfrac{\mu_{\mathrm{c}}}2\norm{\vz_k-\vxi^*}_2^2,
 \qquad E_0\leq1.
 \label{eq:det-first-energy}
\end{equation}
Consequently $K\leq1+\theta^{-1}\log(1/\tau)$ steps suffice for
correction gap at most $\tau\in(0,1)$.
\end{proposition}
\begin{proof}
The Hessian of $J_r$ lies between $\alpha\mI$ and $\mI$, and
$\mu_{\mathrm{c}}\leq\alpha$. Euclidean projection onto $\gK_r$ supplies
\cref{eq:det-comparison-sector} with $\vu=\vxi^*$.
Apply \cref{lem:det-comparison} and discard its nonpositive last term.
The correction slack
\begin{equation*}
 \vell_r:=\mQ\vxi^*-\vh+\lambda_r\vomega
         =\mQ\vx_r^*-\vb+\lambda_r\vomega
\end{equation*}
is nonnegative and satisfies $\vell_r^\trans\vxi^*=0$.
Indeed $\vxi^*$ can be positive only where $\vx_r^*$ is positive.
It follows that
\[
 E_0=\tfrac12\norm{\vxi^*}_{\mQ}^2
               +\tfrac{\mu_{\mathrm{c}}}2\norm{\vxi^*}_2^2\leq1,
\]
because $\vomega^\trans\vxi^*\leq1$, $\omega_i\geq1$, and
$\mQ\preceq\mI$. Finally $\chi^K\leq\exp(-\theta K)$.
\end{proof}

\subsubsection{The projection inequality in the PageRank metric}

Euclidean projection onto a box with a mass cap is not generally
nonexpansive in an unrelated matrix metric. We need only the following
inequality for the particular inverse-source comparator $\vt$.

\begin{lemma}[The second projection comparison]
\label{lem:det-sector}
For any $\vq$, put $\vp=\operatorname{Proj}_{\gK_r}(\vq)$ and
$\vn=\vq-\vp$. Then
\begin{equation*}
 \inner{\vp-\vt}{\mQ\vn}
       =(\mQ\vp-\vh)^\trans\vn\geq0.
\end{equation*}
\end{lemma}
\begin{proof}
The Euclidean normal has a decomposition
$\vn=\gamma\vomega+\vn^{\mathrm{up}}-\vn^{\mathrm{low}}$,
where all multipliers are nonnegative and obey complementary slackness.
At a lower face $p_i=0$, Stieltjes signs give
$(\mQ\vp-\vh)_i\leq0$. At an upper face $p_i=4r\omega_i$,
the upper box and the same signs give
$(\mQ\vp)_i\geq4r(\mQ\vomega)_i=4\lambda_r\omega_i\geq h_i$.
Finally, if $\gamma>0$, the mass constraint is tight and
\[
 \vomega^\trans(\mQ\vp-\vh)
 =\alpha\vomega^\trans\vp-m_h=\alpha m_r-m_h=0.
\]
Each normal component has nonnegative pairing with $\mQ\vp-\vh$.
\end{proof}

To bound local work, we need a response estimate that retains the source's
diffuseness. We choose an auxiliary quadratic whose gradient in the $\mQ$
metric is the same update direction as for $J_r$. Comparing it with $\vt$
will make its initial energy proportional to $\lambda_r m_h$. Define
\begin{equation}
 \begin{aligned}
 \mathcal A_r(\vxi)&:=\tfrac12\norm{\mQ\vxi-\vh}_2^2
                  +\alpha\lambda_r\vomega^\trans\vxi,\\
 B_k&:=\mathcal A_r(\vxi_k)-\mathcal A_r(\vt)
                  +\tfrac{\mu_{\mathrm{c}}}2\norm{\vz_k-\vt}_{\mQ}^2.
 \end{aligned}
 \label{eq:det-second-energy}
\end{equation}
Its gradient in the $\mQ$ inner product is
\begin{equation}
 \nabla_{\mQ}\mathcal A_r(\vxi)
 =\mQ^{-1}\{\mQ(\mQ\vxi-\vh)+\alpha\lambda_r\vomega\}
 =\mQ\vxi-\vh+\lambda_r\vomega.
 \label{eq:det-metric-gradient}
\end{equation}
This is exactly the vector already used in \cref{eq:det-iteration}.
The metric Hessian is $\mQ$, so its eigenvalues are in $[\alpha,1]$.

\begin{lemma}[Uniform squared-response bound]
\label{lem:det-response}
For the actual projected trajectory, at every iteration,
\begin{equation*}
 B_{k+1}\leq\chi B_k,\qquad
 \norm{\mQ(\vxi_k-\vxi^*)}_2^2
       \leq18\lambda_r m_h\leq18\alpha^2r.
\end{equation*}
\end{lemma}
\begin{proof}
Apply \cref{lem:det-comparison} in the $\mQ$ metric, using
\cref{lem:det-sector} and \cref{eq:det-metric-gradient}.
The comparator $\vt$ need not minimize $\mathcal A_r$; the comparison
lemma explicitly permits this.
Diffuseness and \cref{lem:det-correction-domain} give
\[
 \norm{\vh}_2^2\leq4\lambda_r m_h,
 \quad \vt^\trans\mQ\vt=\vt^\trans\vh\leq4r m_h,
 \quad \mathcal A_r(\vt)=\lambda_r m_h.
\]
Thus $B_0\leq3\lambda_r m_h$ and
$B_k\leq\chi^k B_0\leq3\lambda_r m_h$, even if $B_0$ is negative.
The linear and kinetic terms in \cref{eq:det-second-energy} are nonnegative,
which yields $\norm{\mQ\vxi_k-\vh}_2^2\leq8\lambda_r m_h$.
Meanwhile
\[
 \vq^*:=\vh-\mQ\vxi^*=\vb-\mQ\vx_r^*
 \quad\hbox{satisfies}\quad
 \vzero\leq\vq^*\leq\lambda_r\vomega,
 \quad \vomega^\trans\vq^*\leq m_h.
\]
Hence $\norm{\vq^*}_2^2\leq\lambda_r m_h$.
The squared triangle inequality gives the factor $2(8+1)=18$.
\end{proof}

\subsubsection{An analytical set with a uniform outside margin}

Let $\gC_r:=\supp(\vx_{r/2}^*)$. This set is used only in the proof.
By \cref{lem:rppr-order-mass},
\begin{equation}
 \vol(\gC_r)\leq2/r,\qquad
 i\notin\gC_r\quad\Longrightarrow\quad
 (\vx_r^*)_i=0,\quad (\vell_r)_i\geq\lambda_r\omega_i/2.
 \label{eq:det-comparison-margin}
\end{equation}
To verify the margin, both optima vanish at such $i$.
The order $\vx_{r/2}^*\geq\vx_r^*$ and the nonpositive off-diagonals
imply $(\mQ\vx_r^*)_i\geq(\mQ\vx_{r/2}^*)_i$.
Subtracting the two regularization loads and applying KKT at $r/2$
proves the claim. Thus even when complementarity is arbitrarily weak on
the optimal boundary, every excursion outside $\gC_r$ pays a known margin.

\begin{theorem}[Cumulative kinetic volume]
\label{thm:det-kinetic-work}
For every horizon $K\geq1$ of an exact stage,
\begin{equation}
 \sum_{k=0}^{K-1}\vol(\supp\vz_{k+1})\leq\frac{76K}{r}.
 \label{eq:det-kinetic-work}
\end{equation}
The sum includes all repeated visits; no confinement assumption is imposed
on the kinetic supports.
\end{theorem}

Thus $K$ accelerated iterations incur $\mathcal{O}(K/r)$ total degree volume,
counting every repeated visit. The sparse implementation in
\cref{sec:deterministic-implementation} accounts for the remaining
state and reporting operations.
\begin{proof}
Put $\mK_0=(\mI+\mP)/2$, where $\mP=\mA\mD^{-1}$ is nonnegative
and column stochastic. In mass coordinates set
\[
 \vv_k=\theta\mD^{1/2}\vz_k,\quad
 \vv^*=\theta\mD^{1/2}\vxi^*,\quad
 \beta_0=\frac{1-\alpha}{1+\theta}\leq\chi.
\]
The inequality $\beta_0\leq\chi$ follows from $\theta^2\leq\alpha$.
Let $\vv_k^{\mathrm{raw}}:=\theta\mD^{1/2}\vq_k$ be the proposed mass
before projection. Substitution into the raw update gives the exact identity
\begin{equation*}
 \begin{aligned}
 \vv_k^{\mathrm{raw}}-\vv^*
   &=\beta_0\mK_0(\vv_k-\vv^*)+\vh_k^{\mathrm{f}}-\vr_{\mathrm{f}}^*,\\
 \vh_k^{\mathrm{f}}&=-\frac{\mD^{1/2}(\mQ-\mu_{\mathrm{c}}\mI)
                         (\vxi_k-\vxi^*)}{1+\theta},\qquad
 \vr_{\mathrm{f}}^*=\mD^{1/2}\vell_r.
 \end{aligned}
\end{equation*}
For clarity, the coefficient of $\vv_k$ is
$(\mI-\mD^{1/2}\mQ\mD^{-1/2})/(1+\theta)=\beta_0\mK_0$;
the constant comparison terms cancel by the optimum's stationarity with
slack $\vell_r$.

Select $\gI_k=\{i:v_{k+1,i}>v_i^*\}$ and let
$\ve_k=(\vv_k-\vv^*)_+$. On selected coordinates the projected value is
positive, so lower-face multipliers vanish. Upper-face and mass multipliers
only subtract nonnegative mass. Column stochasticity therefore gives
\begin{equation}
 \norm{\ve_{k+1}}_1+\sum_{i\in \gI_k}r_{{\mathrm{f}},i}^*
 \leq\beta_0\norm{\ve_k}_1+\sum_{i\in \gI_k}h_{k,i}^{\mathrm{f}}.
 \label{eq:det-selected-flow}
\end{equation}
The forcing sum is restricted to the selected coordinates $\gI_k$.
Their degree volume will enter the Cauchy--Schwarz estimate below;
summing all positive forcing over the graph would lose this connection.
Since $\ve_0=\vzero$, summing leaves the nonnegative term
$\norm{\ve_K}_1+(1-\beta_0)\sum_{k=1}^{K-1}\norm{\ve_k}_1$ on the left.
Dropping it bounds cumulative selected slack by cumulative signed forcing.
Every positive kinetic coordinate outside $\gC_r$ is selected because
$\vxi^*$ vanishes there. With
\[
 V_{\mathrm{out}}:=\sum_{k<K}\vol(\supp\vz_{k+1}\setminus\gC_r),
 \qquad H_2:=\sum_{k<K}\norm{\mD^{-1/2}\vh_k^{\mathrm{f}}}_2^2,
\]
we obtain
\begin{equation}
 \tfrac{\lambda_r}2 V_{\mathrm{out}}
 \leq\sum_{k<K}\sum_{i\in \gI_k}h_{k,i}^{\mathrm{f}}
 \leq\sqrt{(V_{\mathrm{out}}+2K/r)H_2}.
 \label{eq:det-selected-cauchy}
\end{equation}
The last inequality is weighted Cauchy--Schwarz: the total selected
degree volume is at most $V_{\mathrm{out}}+2K/r$.
Since $\mQ-\mu_{\mathrm{c}}\mI$ commutes with $\mQ$ and has eigenvalues
between zero and those of $\mQ$, \cref{lem:det-response} implies
$H_2\leq18\alpha^2rK$.
Writing $Y=V_{\mathrm{out}}+2K/r$, \cref{eq:det-selected-cauchy} and
Young's inequality give
\[
 Y\leq\frac{2K}{r}+\frac{2\sqrt{YH_2}}{\lambda_r}
 \leq\frac{2K}{r}+\frac Y2+\frac{2H_2}{\lambda_r^2}.
\]
Thus $Y\leq4K/r+4H_2/\lambda_r^2\leq76K/r$.
The total kinetic volume is at most $Y$.
\end{proof}

An accelerated iteration bound and a small final support would not prove
\cref{eq:det-kinetic-work}. The bound concerns the actual projected
trajectory, including every excursion and revisit. Its usefulness depends
on implementing each iteration in time proportional to kinetic updates
and sparse source records, rather than scanning the historical primal
support. We give that implementation next.

\subsection{Sparse implementation}\label{sec:deterministic-implementation}
We now realize \cref{eq:det-iteration} without scanning the entire primal
history at every step. All maps, sets, and order statistics below use
deterministic balanced comparison trees. Bounds therefore include their
logarithmic operations; no expected-time hashing assumption is used.

\subsubsection{Common scaling and sparse response updates}

The primal update multiplies every old coordinate by $\chi$. One common
scale represents this decay; only new kinetic contributions require
individual updates. Let $U=4r$ and write the degree densities of the primal
correction as
\begin{equation*}
 \frac{\xi_{k,i}}{\omega_i}=\sigma_k X_{k,i},\qquad
 \sigma_k=\chi^k.
\end{equation*}
Initially $\sigma_0=1$, and all primal and kinetic response records are
zero. The fixed source is stored separately.
Maintain normalized response densities and the current kinetic response
\begin{equation*}
 \mathscr R_{k,i}:=
   [(\mQ-\mu_{\mathrm{c}}\mI)\vxi_k]_i/(\sigma_k\omega_i),
 \qquad
 \mathscr T_{k,i}:=
   [\mD^{-1/2}(\mQ-\mu_{\mathrm{c}}\mI)\vz_k]_i.
\end{equation*}
The primal averaging step changes normalized records only where the new
kinetic vector is nonzero:
\begin{equation*}
 \sigma_{k+1}=\chi\sigma_k,\qquad
 X_{k+1,i}=X_{k,i}+
       \frac{\theta z_{k+1,i}}{\sigma_{k+1}\omega_i}.
\end{equation*}
Opening the original adjacency list of each such vertex updates its diagonal
response and scatters the change to its neighbors. The work is proportional
to $\vol(\supp\vz_{k+1})$, including every repeated scan.

The raw projection density can be represented by a key $k_i$:
\begin{equation}
 \begin{aligned}
 k_i={}&-\frac{\mathscr R_{k,i}}{\theta(1+\theta)}
       +\frac1{\sigma_k}\left(
          \frac{\chi z_{k,i}}{\omega_i}
          -\frac{\mathscr T_{k,i}}{1+\theta}
          +\frac{h_i}{\theta\omega_i}\right),\\
 \frac{q_{k,i}}{\omega_i}={}&\sigma_k k_i-\lambda_r/\theta.
 \end{aligned}
 \label{eq:det-raw-key}
\end{equation}
The first term is an ordinary key. The bracketed term is an exception only
on the source records and on the current kinetic support and its boundary.
When the common scale changes, remove the old exceptions, update responses
on the newly scanned rows and their neighbors, and install the new
exceptions. This touches only the old and new kinetic incidence sets and
the fixed source records. Stored old keys permit exact deletion even when
their new values have changed. No threshold query is made during a partial
update.

The fixed source is assembled by scanning the baseline support. Its record
set lies in
\begin{equation*}
 \gF_r:=\{v\}\cup\supp\overline\vx\cup\partial(\supp\overline\vx),
 \qquad |\gF_r|=\mathcal{O}(1+1/r).
\end{equation*}
Indeed the baseline is inside $\mathcal S_r^*$, and the number of boundary
records is at most its degree volume. Refreshing source exceptions every
iteration is therefore affordable, including records that are inactive and
whose adjacency lists are never opened.

\subsubsection{An exact finite threshold reporter}

Store each represented key in a balanced tree ordered by $(k_i,i)$, with
subtree cardinalities, degree sums, and weighted key sums. Vertex labels
break ties only; mass comparisons group equal numerical keys correctly.
For a real threshold $t$, define
\begin{equation*}
 D(t)=\sum_{k_i>t}d_i,\quad
 S(t)=\sum_{k_i>t}d_i k_i,\quad
 T(t)=S(t)-tD(t)=\sum_i d_i(k_i-t)_+.
\end{equation*}
Each tail query costs $\mathcal{O}(\log(N+2))$ operations when $N$ records are stored.
Writing $t=(\lambda_r/\theta+\gamma)/\sigma_k$, the projected mass is
\begin{equation}
 M(t)=\sigma_k\{T(t)-T(t+U/\sigma_k)\}.
 \label{eq:det-mass-query}
\end{equation}
This formula includes both the lower truncation and upper clipping in
\cref{eq:det-waterfill}.

\begin{lemma}[Deterministic projection reporting]
\label{lem:det-reporter}
Given the maintained state, the exact multiplier and every positive
coordinate of $\operatorname{Proj}_{\gK_r}(\vq_k)$ can be obtained in
\begin{equation*}
 \mathcal{O}\bigl(\log^2(N+2)+|\supp\vz_{k+1}|\bigr)
\end{equation*}
operations, in addition to the charged point updates and row scans.
The search uses finitely many order statistics, with no dependence on a
smallest positive coordinate or on a complementarity margin.
\end{lemma}
\begin{proof}
Evaluate \cref{eq:det-mass-query} at $t_0=\lambda_r/(\theta\sigma_k)$.
If $M(t_0)\leq m_r$, then $\gamma=0$ is valid. Otherwise the desired root
lies between $t_0$ and the maximum stored key, where the mass is zero.
The only breakpoints of $M$ are the two sorted lists
\begin{equation*}
 \{k_i\}_{i=1}^N,\qquad \{k_i-U/\sigma_k\}_{i=1}^N.
\end{equation*}
Binary search the first list by rank, evaluating $M$ at each queried
breakpoint and maintaining a bracket around the target mass. If equality
holds, return that breakpoint. Otherwise stop when no breakpoint from the
first list lies strictly inside the bracket. Repeat for the second list.
Narrowing a bracket cannot restore an excluded breakpoint from the first
list. Each search uses $\mathcal{O}(\log(N+2))$ rank selections and mass queries,
each of logarithmic cost. The remaining interval has no breakpoint, so
$M$ is affine there and one division gives a valid root. Endpoint equality
handles a flat interval; a constant affine function cannot strictly
straddle the target. If a zero cap occurs, zero is returned directly.

Finally enumerate keys strictly above the root. They are exactly the
positive projected coordinates because $U>0$; each value follows from
\cref{eq:det-waterfill}. A tree successor traversal takes one logarithmic
initial search and linear output work. Exact equalities produce zero
coordinates and do not trigger an adjacency scan.
\end{proof}

\subsubsection{Discovery, completion, and the full ledger}

Before every projection, each vertex with positive primal or kinetic state
has had its row scanned, and all its neighbors have been exposed. Every
source coordinate is also represented. A vertex never exposed has zero
state, zero source, and zero neighbor response. Its raw density is therefore
$-\lambda_r/\theta<0$, so it cannot be selected. The finite reporter thus
agrees exactly with the full-space projection. Degree replies are requested
only for the seed or a discovered neighbor and are cached; their number is
bounded by the same incidences. Inactive boundary records incur scalar
work but no recursive row scan.

Let $K$ be the current stage's iteration count and set
$V_k^{\mathrm{row}}=\vol(\supp\vz_k)$, with $V_0^{\mathrm{row}}=0$.
Removing old exceptions and installing new ones costs
$\mathcal{O}(|\gF_r|+V_k^{\mathrm{row}}+V_{k+1}^{\mathrm{row}})$ point operations up to
logarithms. Each new primal response and kinetic response is assembled from
the emitted rows at the same charge. Thus \cref{lem:det-reporter} and
\cref{thm:det-kinetic-work} account for the entire iteration loop.

At completion, materialize the baseline and the union of kinetic supports.
Their cumulative row volume is already charged. One scan of this candidate
support computes the exact projected-gradient input on the candidate and
its boundary. Apply the downward correction before opening any newly
positive neighbor's row. The repaired support lies in $\mathcal S_r^*$,
so its next source assembly costs at most $1/r$ additional incidences.
At each stage handoff, discard obsolete records after constructing the new
baseline and source; charge disposal to the number of discarded records.
These facts
give the stage bound
\begin{equation}
 \mathcal{O}\!\left(
  \left[K(1+|\gF_r|)+\sum_{k=1}^KV_k^{\mathrm{row}}
              +\vol(\supp\overline\vx)+1\right]
  \log^2(N+2)\right)
 =\mathcal{O}\!\left(\frac Kr\log^2(N+2)\right).
 \label{eq:det-stage-ledger}
\end{equation}
Here $N=\mathcal{O}(1+K/r)$ suffices for the current stage's records and cached
adjacencies. In particular, lazy arithmetic does not hide a dense
materialization or a global terminal certificate.

\begin{proof}[Proof of \cref{thm:deterministic-rppr}]
The constant-output branches follow from the zero-solution criterion and
$F_\rho(\vzero)-F_\rho(\vx_\rho^*)\leq\alpha/2$.
For all other branches, \cref{prop:det-stage-rate} supplies the candidate
accuracy and \cref{lem:det-repair} closes the stage invariant. The final
choice of $\delta$ gives the requested objective gap and safe support.

For work, let $L_{\mathrm{par}}=\log(2/(\alpha\rho\min\{1,\epsobj\}))$ as in
\cref{eq:main-parameter-log}. Every stage uses
$K=\mathcal{O}(L_{\mathrm{par}}/\sqrt\alpha)$ iterations. Indeed its prescribed tolerance is
$\tau=\alpha\delta^2/8$, where halving from $\alpha r/2$ makes
$\delta^2$ within an absolute constant of
$\min\{\alpha^2r^2,\epsobj r\}$ at the final stage; intermediate
stages use the first value. Thus $\log(1/\tau)=\mathcal{O}(L_{\mathrm{par}})$.
The halving schedule has
\begin{equation}
 \sum_j\frac1{r_j}\leq\frac4\rho.
 \label{eq:det-geometric-work}
\end{equation}
Summing \cref{eq:det-stage-ledger} therefore gives
\[
 \mathcal{O}\!\left(\frac{L_{\mathrm{par}}}{\rho\sqrt\alpha}
       \log^2\!\left(2+\frac{L_{\mathrm{par}}}{\rho\sqrt\alpha}\right)\right).
\]
The same history bounds peak state by $\mathcal{O}(1+L_{\mathrm{par}}/(\rho\sqrt\alpha))$ words
and output by $1/\rho$ coordinates. All preprocessing is confined to
discovered rows and included in this charge.
\end{proof}

\paragraph{Arithmetic representation.}
The exact algorithm can use rational degree densities when its numerical
inputs are rational: the matrix acting on densities is
$\mD^{-1/2}\mQ\mD^{1/2}=((1+\alpha)/2)\mI-((1-\alpha)/2)\mD^{-1}\mA$.
An output stores $\widehat f_i:=\widehat x_i/\omega_i$ and $d_i$,
representing $\widehat x_i=\widehat f_i\sqrt{d_i}$ exactly. This observation
alone gives no bound on growing rational denominators.
\Cref{app:bounded-arithmetic} supplies directed rounding, controlled
neighbor-response error, scalar rebasing, and a separate encoding analysis.

\section{Randomized Threshold Batching}\label{sec:randomized-method}
\subsection{Algorithm and guarantee}\label{sec:threshold-batch-algorithm}
The algorithm keeps a discovered active set, solves the corresponding
principal system, and checks a local residual certificate. It then scans
the boundary and admits all vertices whose violations exceed a common
threshold. The certificate makes these admissions safe; the depth theorem
in \cref{sec:threshold-batch-depth} bounds how often the cycle repeats.
We first give the parameters and complete procedure, then prove the depth
bound and combine it with the numerical and graph-access costs.

Remove the constant-output cases first. For
$0<\epsobj<\alpha/2$, $0<\rho<1/d_v$, and $0<\alpha<1$, choose
$\vartheta$ by halving one until
\begin{equation}
 \frac1{16}\sqrt{\alpha\rho\epsobj}<\vartheta
       \leq\frac18\sqrt{\alpha\rho\epsobj},\qquad
 \nu:=\vartheta/4.
 \label{eq:algorithm-threshold}
\end{equation}
Equivalently, stop when $\vartheta^2\leq\alpha\rho\epsobj/64$;
no square root is computed. Choose a dyadic $\theta_{\mathrm{b}}$ as in
\cref{eq:det-acceleration-parameters}. Fix a failure budget
$\zeta\in(0,1)$ and define the integer caps
\begin{equation}
 \begin{aligned}
 T_{\mathrm{b}}&=1/\theta_{\mathrm{b}},&
 q_{\mathrm{acc}}&=\min\{j\in\sN:2^j\geq64/\epsobj\},&
 J&=T_{\mathrm{b}}q_{\mathrm{acc}},\\
 N_{\mathrm{try}}&=\min\{j\in\sN:4^j\geq(J+1)/\zeta\}.&&&
 \end{aligned}
 \label{eq:algorithm-parameters}
\end{equation}
All these choices use halving, integer multiplication, and comparisons.
Each iteration constructs a fresh supplied-face solver. The depth bound
in \cref{thm:threshold-batch-depth} controls the cost of these repeated
solves.

\begin{algorithm}[htbp]
\caption{Threshold-batch local RPPR solver}
\label{alg:threshold-batch-rppr}
\begin{algorithmic}[1]
\REQUIRE Local graph access, seed $v$, $\alpha,\rho,\epsobj>0$, $\zeta\in(0,1)$
\STATE Query $d_v$. If $\rho d_v\geq1$ or $\epsobj\geq\alpha/2$, return zero.
\STATE If $\alpha=1$, return the exact seed coordinate.
\STATE Set the parameters in \cref{eq:algorithm-threshold,eq:algorithm-parameters}.
\STATE $\gU\gets\{v\}$; expose its row and cache the discovered degrees.
\FOR{$k=0,1,\ldots,J$}
  \STATE Make at most $N_{\mathrm{try}}$ independent local SDD calls at relative energy
         tolerance $\nu$.  After each call, scan the active rows and accept
         the first $\vz$ satisfying
         $\norm{\mQ_{\gU\gU}\vz_\gU-(\vc_\rho)_\gU}_2\leq\nu\sqrt\alpha$.
  \STATE If no call is certified, \textbf{return} an explicit failure flag.
  \STATE Extend $\vz$ by zero outside $\gU$.
  \STATE Scan all rows in $\gU$ once and accumulate
         $\widehat r_i=(\vc_\rho)_i-(\mQ\vz)_i$ for every boundary candidate.
  \IF{$k=J$}
    \STATE \textbf{return} $[\vz]_+$ and $\gU$.
  \ENDIF
  \STATE $\gB\gets\{i\in\partial \gU:\widehat r_i>\vartheta/2\}$.
  \IF{$\gB=\emptyset$}
    \STATE \textbf{return} $[\vz]_+$ and $\gU$.
  \ENDIF
  \STATE Expose the adjacency lists of $\gB$ and set $\gU\gets \gU\cup \gB$.
\ENDFOR
\end{algorithmic}
\end{algorithm}
\FloatBarrier

\begin{theorem}[Graph-uniform local RPPR]
\label{thm:local-rppr}
On every point-source instance in the nontrivial regime specified above,
\cref{alg:threshold-batch-rppr} reports failure with probability at most
$\zeta$ and otherwise returns a nonnegative $\widehat\vx$ satisfying
\begin{equation}
 F_\rho(\widehat\vx)-F_\rho(\vx_\rho^*)\leq\epsobj.
 \label{eq:local-rppr-gap}
\end{equation}
Using independent capped trials of a randomized nearly-linear SDD solver on
each exposed face, its expected fully charged work is
\begin{equation}
 \softO\!\left(
   \frac{1}{\rho\sqrt\alpha}
   \log\frac1{\epsobj}\log\frac2\zeta
 \right).
 \label{eq:local-rppr-work}
\end{equation}
The bound includes construction of every principal system, all adjacency and
candidate incidences, every fresh solver state, repeated active-row and
boundary scans, certification, projection, materialization, and output.  It
uses no supplied support or ambient preprocessing.  If
$\epsobj\geq\alpha/2$, returning zero is sufficient; if
$\rho\geq1/d_v$, one seed-degree query identifies the exact zero solution.
\end{theorem}

The returned vector is supported inside $\mathcal S_\rho^*$.
The deterministic repair in \cref{cor:randomized-safe-repair} additionally
makes it a subsolution below $\vx_\rho^*$, with the same asymptotic work bound.

The next subsection establishes the batch bound used in this theorem.
Its complete correctness and work proof follows in
\cref{sec:threshold-batch-analysis}.

\subsection{Threshold-batch depth}\label{sec:threshold-batch-depth}
Write the outside face residual as
\begin{equation*}
 \vr^\gU:=\vc_\rho-\mQ\vx^\gU=-\vw^\gU.
\end{equation*}
Fix $\vartheta\geq0$.  A nested sequence
$\gU_k=\gB_0\mathbin{\dot\cup}\cdots\mathbin{\dot\cup}\gB_k$ is
\emph{$\vartheta$-batched} if $\gB_0=\{v\}$, every later $\gB_k$ is nonempty
and satisfies $r_i^{\gU_{k-1}}>0$ for every $i\in\gB_k$, and, immediately
before admitting $\gB_k$,
\[
 r_i^{\gU_{k-1}}\leq\vartheta
 \qquad\text{for every }i\in\gS_\rho^*\setminus\gU_k.
\]
The condition includes coordinates not yet assigned to a later block.
Any finite such sequence can be extended to partition $\gS_\rho^*$ by
repeatedly admitting all exact positive residuals. After each added batch,
every remaining coordinate had nonpositive residual on the preceding face,
so the threshold condition is preserved. This is an analytical condition;
the algorithm need not know $\gS_\rho^*$. The residual certificate will
establish it for all unreported coordinates in
\cref{sec:threshold-batch-analysis}.

\begin{lemma}[Chebyshev inverse approximation]
\label{lem:chebyshev-inverse}
Let $\alpha\mI\preceq\mM\preceq b_0\mI$, with $b_0\geq\alpha>0$.
For each $j\geq1$, there is a polynomial $p_{j-1}$ of degree at most
$j-1$ such that
\begin{equation*}
 \norm{\mM^{-1}-p_{j-1}(\mM)}_2
 \leq\frac2\alpha
       \left(\frac{\sqrt{b_0/\alpha}-1}
                  {\sqrt{b_0/\alpha}+1}\right)^j.
\end{equation*}
The corresponding residual polynomial has absolute value at most twice
the displayed geometric factor on $[\alpha,b_0]$.
\end{lemma}
\begin{proof}
If $b_0=\alpha$, take $p_{j-1}(t)=1/\alpha$; both errors vanish.
Otherwise $b_0>\alpha$.
Let $T_j$ be the first-kind Chebyshev polynomial and set
\[
 r_j(t)=
 \frac{T_j((b_0+\alpha-2t)/(b_0-\alpha))}
      {T_j((b_0+\alpha)/(b_0-\alpha))}.
\]
Its numerator has absolute value at most one on the interval.
For $u>1$,
$T_j(u)=((u+\sqrt{u^2-1})^j+(u-\sqrt{u^2-1})^j)/2$.
Substitution bounds $|r_j(t)|$ by twice the stated geometric factor.
Also $r_j(0)=1$, so $p_{j-1}(t)=(1-r_j(t))/t$ is a polynomial of
degree at most $j-1$ and
$|1/t-p_{j-1}(t)|=|r_j(t)|/t\leq|r_j(t)|/\alpha$.
Diagonalizing the symmetric matrix proves the operator-norm bound.
\end{proof}

\begin{theorem}[Threshold-batch energy-depth bound]
\label{thm:threshold-batch-depth}
Let $\gU_0\subset \gU_1\subset\cdots$ be a $\vartheta$-batched reachable
sequence, extended to the full optimal support.  Put
\begin{equation*}
 q_\alpha:=
 \frac{\sqrt{2/\alpha}-1}{\sqrt{2/\alpha}+1}.
\end{equation*}
Then every included face $\gU_J$ satisfies
\begin{equation}
 \phi_\rho(\vx^{\gU_J})-\phi_\rho(\vx_\rho^*)
 \leq8q_\alpha^{2J}+\frac{\vartheta^2}{\alpha\rho}.
 \label{eq:batch-depth-main}
\end{equation}
For $0<\epsobj<1$, the decaying term is at most $\epsobj/2$ after
$\mathcal{O}(\alpha^{-1/2}\log(2/\epsobj))$ batches. If the sequence ends sooner,
its final face solution is exact. The estimate is independent of
$|\gS_\rho^*|$, graph diameter, and any complementarity margin.
\end{theorem}

The proof has three steps. Block Cholesky elimination expresses the
objective gap as a tail energy. A block-tridiagonal comparison system
preserves spectral bounds and encodes delayed admissions. Chebyshev
approximation then separates the decaying seed contribution from the
threshold error. These factors are analytical objects and are not
constructed by the local algorithm.

\begin{proof}
Let $\gS=\gS_\rho^*$ and order it by the blocks
$\gB_0,\gB_1,\ldots,\gB_K$.  Block elimination gives
\begin{equation*}
 \mQ_{\gS\gS}=\mL\mDelta\mL^\trans,
 \qquad
 \mDelta=\operatorname{diag}(\mDelta_0,\ldots,\mDelta_K),
 \qquad
 \vg=\mL^{-1}(\vc_\rho)_\gS.
\end{equation*}
Here $\mL$ is unit block lower triangular, every pivot $\mDelta_k$ is
Stieltjes, and $\vg_k$ is the positive Schur residual on $\gB_k$ when it is
admitted. Let $\mR_k$ be its lower Cholesky factor, so
$\mDelta_k=\mR_k\mR_k^\trans$, and set
\begin{equation*}
 \vz_k:=\mR_k^{-1}\vg_k.
\end{equation*}
Each $\mR_k$ is a triangular nonsingular M-matrix, so $\vz_k\geq\vzero$.
Completion of squares in the block elimination gives the exact identity
\begin{equation}
 2\bigl(\phi_\rho(\vx^{\gU_J})-\phi_\rho(\vx_\rho^*)\bigr)
 =\sum_{k>J}\norm{\vz_k}_2^2.
 \label{eq:batch-exact-gap-tail}
\end{equation}

Put
$\mC=\mL\operatorname{diag}(\mR_0,\ldots,\mR_K)$, so that
$\mQ_{\gS\gS}=\mC\mC^\trans$.  This is the scalar lower Cholesky factor in the
block-respecting order.  Its recursion shows that every off-diagonal entry is
nonpositive.  Let $\widehat{\mC}$ retain only the diagonal blocks and first
block subdiagonal.  Both factors are triangular nonsingular M-matrices and
$\mC\leq\widehat{\mC}$ entrywise.  Hence
\begin{equation}
 \vzero\leq\widehat{\mC}^{-1}\leq\mC^{-1},
 \qquad
 \norm{\widehat{\mC}^{-1}}_2
 \leq\norm{\mC^{-1}}_2\leq\frac1{\sqrt\alpha}.
 \label{eq:batch-bidiagonal-inverse}
\end{equation}
The inverse ordering follows directly from
\[
 \mC^{-1}-\widehat{\mC}^{-1}
 =\mC^{-1}(\widehat{\mC}-\mC)\widehat{\mC}^{-1}\geq\vzero.
\]
The spectral-norm comparison is valid because both inverses are nonnegative:
$|\widehat{\mC}^{-1}\vz|\leq\mC^{-1}|\vz|$ for every $\vz$.

There is also a graph-independent upper bound.  The $k$th block row of
$\mC$ satisfies
\begin{equation*}
 \sum_{i\leq k}\mC_{ki}\mC_{ki}^\trans
 =\mQ_{\gB_k\gB_k}\preceq\mI.
\end{equation*}
Deleting all but two input blocks cannot increase its row-map norm, and each
input block occurs in at most two retained rows.  Therefore
\begin{equation*}
 \norm{\widehat{\mC}}_2\leq\sqrt2.
\end{equation*}
Consequently
$\mH_{\mathrm{ch}}:=\widehat{\mC}\widehat{\mC}^\trans$ is block tridiagonal and
\begin{equation*}
 \alpha\mI\preceq\mH_{\mathrm{ch}}\preceq2\mI.
\end{equation*}

We next encode causality.  For $k\geq2$, immediately before $\gB_{k-1}$ is
admitted, the still-future block $\gB_k$ has residual at most
$\vartheta\one$; this is its residual after elimination through $\gB_{k-2}$.
For $k=1$, the same bound follows from
$(\vc_\rho)_i=-\alpha\rho\sqrt{d_i}<0$ away from the seed.  Eliminating
$\gB_{k-1}$ subtracts $\mL_{k,k-1}\vg_{k-1}$, so
\begin{equation*}
 \vg_k+\mL_{k,k-1}\vg_{k-1}\leq\vartheta\one
 \qquad(k\geq1).
\end{equation*}
Using
$\mC_{k,k-1}=\mL_{k,k-1}\mR_{k-1}$, this becomes
\begin{equation*}
 (\widehat{\mC}\vz)_0=\vg_0,
 \qquad
 (\widehat{\mC}\vz)_k\leq\vartheta\one\quad(k\geq1).
\end{equation*}
Since $\widehat{\mC}^{-1}\geq\vzero$,
\begin{equation}
 \vzero\leq\vz\leq
 \widehat{\mC}^{-1}
 \begin{bmatrix}\vg_0\\ \vartheta\one_{\gS\setminus \gB_0}\end{bmatrix}.
 \label{eq:batch-chain-domination}
\end{equation}

The distributed threshold source is bounded by
\begin{equation}
 \norm{\widehat{\mC}^{-1}[\vzero;\vartheta\one]}_2^2
 \leq\frac{\vartheta^2|\gS|}{\alpha}
 \leq\frac{\vartheta^2}{\alpha\rho},
 \label{eq:batch-threshold-source}
\end{equation}
where the last step uses
$|\gS|\leq\vol(\gS)\leq1/\rho$.

For the seed source and $J\geq1$, use \cref{lem:chebyshev-inverse} to take a
degree-$(J-1)$ polynomial
$p_{J-1}$ for $1/t$ on $[\alpha,2]$, with
\begin{equation*}
 \norm{\mH_{\mathrm{ch}}^{-1}-p_{J-1}(\mH_{\mathrm{ch}})}_2
 \leq\frac2\alpha q_\alpha^J.
\end{equation*}
Use $\widehat{\mC}^{-1}=\widehat{\mC}^{\trans}\mH_{\mathrm{ch}}^{-1}$.
Because $\mH_{\mathrm{ch}}$ is block tridiagonal,
$p_{J-1}(\mH_{\mathrm{ch}})[\vg_0;\vzero]$ reaches only blocks through
$J-1$. Multiplication by the block upper-bidiagonal
$\widehat{\mC}^\trans$ does not increase the largest block index.  If
$\mPi_{>J}$ projects onto the block tail, then
\begin{align}
 \norm{\mPi_{>J}\widehat{\mC}^{-1}[\vg_0;\vzero]}_2
 &\leq\norm{\widehat{\mC}}_2
       \norm{\mH_{\mathrm{ch}}^{-1}-p_{J-1}(\mH_{\mathrm{ch}})}_2
       \norm{\vg_0}_2\notag\\
 &\leq2\sqrt2q_\alpha^J.
 \label{eq:batch-seed-tail}
\end{align}
Here
$\vg_0=(\vc_\rho)_v=\alpha(1-\rho d_v)/\sqrt{d_v}\leq\alpha$.
For $J=0$, the same bound follows from
\cref{eq:batch-bidiagonal-inverse} and $\alpha\leq1$.
Combining
\cref{eq:batch-exact-gap-tail,eq:batch-chain-domination,eq:batch-threshold-source,eq:batch-seed-tail}
with $\norm{\va+\vb}_2^2/2\leq\norm{\va}_2^2+\norm{\vb}_2^2$ proves
\cref{eq:batch-depth-main}.
\end{proof}

The two terms have different meanings.  The seed term decays spectrally down
the block chain.  The threshold term charges every delayed coordinate once,
at the block where it is eventually eliminated.  A vertex can therefore
remain below threshold for many faces without generating a repeated,
uncharged error.

\subsection{Correctness and total local work}\label{sec:threshold-batch-analysis}
\begin{proof}[Proof of \cref{thm:local-rppr}]
We first verify the face-solve wrapper.  On a reachable face,
$(\vc_\rho)_\gU\leq\vb_\gU$ and $\mQ_{\gU\gU}^{-1}\geq\vzero$, so
\begin{equation*}
 \vzero\leq\vx_\gU^\gU\leq\mQ_{\gU\gU}^{-1}\vb_\gU,
 \qquad
 \norm{\vx^\gU}_{\mQ}^2
 =(\vc_\rho)_\gU^\trans\vx_\gU^\gU
 \leq\vb_\gU^\trans\mQ_{\gU\gU}^{-1}\vb_\gU\leq\alpha.
\end{equation*}
The last inequality uses
$\norm{\vb}_2^2=\alpha^2/d_v\leq\alpha^2$ and
$\mQ_{\gU\gU}\succeq\alpha\mI$.

The solver is applied in degree coordinates, where the supplied matrix is
\begin{equation*}
 \mH_{\gU\gU}:=\mD_\gU^{1/2}\mQ_{\gU\gU}\mD_\gU^{1/2}
 =\frac{1+\alpha}{2}\mD_\gU-\frac{1-\alpha}{2}\mA_{\gU\gU}.
\end{equation*}
Its diagonal dominates its absolute off-diagonal row sum. Solving
$\mH_{\gU\gU}\vy_\gU=\mD_\gU^{1/2}(\vc_\rho)_\gU$ and returning
$\vz_\gU=\mD_\gU^{1/2}\vy_\gU$ preserves the energy norm.
Vectors on $\gU$ are extended by zero outside $\gU$.

Use the supplied-matrix solver of
\citet[Lemma~4.5 and Theorem~4.6]{koutis2011nearly}, with a fresh
chain having success probability at least $3/4$ on each call. All of its
construction and iteration work is charged on the supplied local matrix.  On a good call at relative energy tolerance $\nu$, the error
$\ve=\vz-\vx^\gU$ satisfies
$\norm{\ve}_{\mQ}\leq\nu\norm{\vx^\gU}_{\mQ}\leq\nu\sqrt\alpha$ and hence
passes the residual test.  Conversely, every accepted call, whether or not
its random chain was good, obeys the deterministic bound
\begin{equation*}
 \norm{\vz-\vx^\gU}_{\mQ}^2
 =\vr_\gU^\trans\mQ_{\gU\gU}^{-1}\vr_\gU
 \leq\frac1\alpha\norm{\vr_\gU}_2^2
 \leq\nu^2,
 \qquad
 \vr_\gU:=\mQ_{\gU\gU}\vz_\gU-(\vc_\rho)_\gU.
\end{equation*}
With constant call-failure probability at most $1/4$, $N_{\mathrm{try}}$ fresh trials
leave a phase uncertified with probability at most
$4^{-N_{\mathrm{try}}}\leq\zeta/(J+1)$.  The face depends on earlier accepted calls. Conditional on that entire
history, each fresh call has the stated success guarantee for the now fixed
supplied matrix. The conditional failure bound and a union bound over at
most $J+1$ faces prove the reported-failure claim.

The boundary residual is uniformly accurate.  Since
$\mQ^2\preceq\mQ$,
\begin{equation*}
 |(\mQ\ve)_i|\leq\norm{\mQ\ve}_2
 \leq\norm{\ve}_{\mQ}\leq\nu=\vartheta/4.
\end{equation*}
Every admitted coordinate therefore has exact positive residual greater than
$\vartheta/4$ and is support-safe by
\cref{thm:safe-batched-pivots}.  Every unreported boundary coordinate has
exact residual at most $3\vartheta/4<\vartheta$.  Thus the nonempty batches
form a $\vartheta$-batched sequence and every face satisfies
\begin{equation}
 \gU\subseteq\mathcal S_\rho^*,
 \qquad
 \vol(\gU)\leq\vol(\mathcal S_\rho^*)\leq1/\rho.
 \label{eq:algorithm-face-volume}
\end{equation}

If the loop reaches its cap, \cref{thm:threshold-batch-depth} and
\cref{eq:algorithm-threshold,eq:algorithm-parameters} give
\begin{equation*}
 \phi_\rho(\vx^\gU)-\phi_\rho(\vx_\rho^*)
 \leq8q_\alpha^{2J}+\frac{\vartheta^2}{\alpha\rho}
 \leq\frac{\epsobj}{8}+\frac{\epsobj}{64}.
\end{equation*}
If $\gB=\emptyset$, every exact positive outside residual is at most
$3\vartheta/4$.  Only boundary coordinates can be positive and
$|\partial \gU|\leq\vol(\gU)\leq1/\rho$, so
\cref{lem:kkt-gap-certificate} gives
\begin{equation*}
 \phi_\rho(\vx^\gU)-\phi_\rho(\vx_\rho^*)
 \leq\frac{9\vartheta^2}{32\alpha\rho}
 \leq\frac{9\epsobj}{2048}.
\end{equation*}

Finally, $\vx^\gU\geq\vzero$, so Euclidean orthant projection is nonexpansive
relative to $\vx^\gU$.  Restricted stationarity and
$\alpha\mI\preceq\mQ\preceq\mI$ imply
\begin{equation*}
 \phi_\rho([\vz]_+)-\phi_\rho(\vx^\gU)
 =\frac12\norm{[\vz]_+-\vx^\gU}_{\mQ}^2
 \leq\frac{\nu^2}{2\alpha}
 \leq\frac{\rho\epsobj}{2048}
 \leq\frac{\epsobj}{2048}.
\end{equation*}
The last inequality uses $\rho<1/d_v\leq1$. Combining either exact-face
gap bound with the projection-error bound proves \cref{eq:local-rppr-gap}.

It remains to charge work. The supplied matrix $\mH_{\gU\gU}$ has
$\mathcal{O}(\vol(\gU))$ stored nonzeros.  The requested relative
accuracy is
\begin{equation*}
 \eta_{\mathrm{sdd}}=\nu
 \in\left(\frac1{64}\sqrt{\alpha\rho\epsobj},
          \frac1{32}\sqrt{\alpha\rho\epsobj}\right].
\end{equation*}
The supplied-matrix theorem of \citet{koutis2011nearly} is nearly linear in
the local nonzeros and logarithmic in $1/\eta_{\mathrm{sdd}}$.  Including capped
retries and certification, one face costs
\[
 \softO\!\left(
   \vol(\gU)\log\frac1{\eta_{\mathrm{sdd}}}\,N_{\mathrm{try}}
 \right).
\]
Degree replies for newly exposed endpoints are cached. Candidate records
and sparse indices use deterministic balanced trees, contributing local-size
logarithms. Every other operation in the phase is therefore
$\widetilde{\mathcal{O}}(\vol(\gU))$, including candidate maintenance and all repeated
reads of cached incidences.  Batch exposure is
charged to the next face and remains $\mathcal{O}(1/\rho)$ by
\cref{eq:algorithm-face-volume}.  Finally,
\begin{equation*}
 \log(1/q_\alpha)
 =\log\frac{\sqrt{2/\alpha}+1}{\sqrt{2/\alpha}-1}
 \geq\sqrt{2\alpha}.
\end{equation*}
Since $\theta_{\mathrm{b}}\leq\sqrt\alpha$, the chosen
$J=T_{\mathrm{b}}q_{\mathrm{acc}}$ satisfies
$q_\alpha^{2J}\leq\exp(-2\sqrt2q_{\mathrm{acc}})
 \leq2^{-q_{\mathrm{acc}}}\leq\epsobj/64$,
justifying the capped-gap estimate above.
Also $J=\mathcal{O}(\alpha^{-1/2}\log(1/\epsobj))$.
Multiplying the per-face charge by this count proves
\cref{eq:local-rppr-work}.  The bound is unconditional:
only certified batches are exposed, even on an execution that later reports
failure. More precisely, every nonempty batch adds at least one new vertex.
Thus the number of face solves is at most
$\min\{k_*,J+1\}$, and each face has volume at most $V_*$.
The same ledger gives
\begin{equation}
 \E\Work\leq
 \softO\!\left(V_*\min\{k_*,\alpha^{-1/2}\}\right),
 \label{eq:randomized-refined-ledger}
\end{equation}
with the accuracy and failure logarithms already specified. Graph and
boundary records have $\mathcal{O}(V_*)$ words; only one supplied-face numerical
state is live at a time. All its allocation and disposal are charged.

For the constant-output cases, $F_\rho(\vzero)-F_\rho(\vx_\rho^*)\leq
\alpha/2$ follows from $\norm{\vb}_2^2\leq\alpha^2$ and strong convexity of
the smooth quadratic.  The seed-load sign proves the zero-RPPR test. The case $\alpha=1$ is diagonal and has the explicit solution stated in \cref{sec:problem-formulation}.
\end{proof}

\begin{remark}[Las Vegas wrapper and exact scope]
\label{rem:las-vegas-scope}
All accepted face solves and all admissions are verified deterministically;
randomness can only cause an explicit failure flag.  Running with a constant
$\zeta$ and restarting after a flag therefore gives a Las Vegas exact-real
algorithm with constant expected restart overhead.  This statement does not
provide deterministic bit complexity or guarantee the stability of a
particular floating-point implementation.
\end{remark}

\begin{corollary}[A safe subsolution output at the same work bound]
\label{cor:randomized-safe-repair}
The randomized algorithm can additionally return
\begin{equation}
 \vzero\leq\widehat\vx\leq\vx_\rho^*,\qquad
 \vzero\leq\vb-\mQ\widehat\vx\leq2\alpha\rho\vomega,
 \label{eq:randomized-safe-output}
\end{equation}
while satisfying the requested objective gap and retaining the expected
work in \cref{eq:randomized-refined-ledger} up to parameter logarithms.
\end{corollary}
\begin{proof}
Choose $\delta$ by halving $\alpha\rho/2$ until
$2\delta^2/\rho\leq\epsobj$, and call the threshold-batch algorithm
with the smaller objective budget $\tau=\alpha\delta^2/8$.
On its returned candidate perform the projected-gradient and downward
correction in \cref{eq:det-repair-map}, with $r=\rho$ and old baseline
zero. \Cref{lem:det-repair} gives
\cref{eq:randomized-safe-output}, support safety, and the required gap.
The new budget has
$\log(1/\tau)=\mathcal{O}(\log(2/(\alpha\rho\min\{1,\epsobj\})))$.
The terminal product scans only the already discovered candidate face.
Apply the downward correction before opening newly positive neighbors;
the repaired support lies inside $\mathcal S_\rho^*$.
Hence any additional scan and output cost is $\widetilde{\mathcal{O}}(V_*)$.
The repair is deterministic and introduces no new failure event.
\end{proof}

\begin{proof}[Proof of \cref{thm:randomized-rppr-main}]
Use \cref{thm:local-rppr,eq:randomized-refined-ledger} and, for the stronger
subsolution output, \cref{cor:randomized-safe-repair}.
Set $\zeta=1/4$ and restart after an explicit failure. Each entire run
uses fresh random choices and succeeds with probability at least $3/4$.
The usual renewal identity bounds total expected work by at most $4/3$
times the unconditional expected work of one run; it does not require
independence between that run's work and its success event. Every returned
vector is certified, so this is Las Vegas. Keeping the requested $\zeta$
and making a single run gives the alternative failure-reporting guarantee.
\end{proof}

\section{PPR Consequences and Discussion}\label{sec:two-stage-ppr}
The objective-to-PPR conversion in \cref{prop:rppr-ppr-bridge} already
proves \cref{cor:main-ppr}. We record two further consequences: a positive
residual certificate compatible with the ACL approximation definition, and
a set-only handoff to an ordinary principal PPR solve.

\subsection{A positive residual certificate}

Write $\mP_{\mathrm{L}}=(\mI+\mA\mD^{-1})/2$ and, for any load $\va$,
let
\begin{equation*}
 \operatorname{PPR}_\alpha(\va)
   :=\alpha(\mI-(1-\alpha)\mP_{\mathrm{L}})^{-1}\va.
\end{equation*}
The ACL approximation condition is stronger than an arbitrary small
solution error: a nonnegative vector $\widehat\vpi$ is an ACL
$\epsappr$-approximation if
\begin{equation}
 \widehat\vpi=\operatorname{PPR}_\alpha(\eunit{v}-\vr^{\mathrm{mass}}),
 \qquad \vzero\leq\vr^{\mathrm{mass}}\leq\epsappr\mD\one.
 \label{eq:acl-approximation-contract}
\end{equation}
This is the residual representation in
\citet{andersen2007using} and \citet[Definition~1.1]{wei2026simple},
translated to the lazy parameter.

\begin{corollary}[ACL-compatible local acceleration]
\label{cor:acl-compatible-output}
For $0<\epsappr<1$, both our deterministic and randomized methods can
return an ACL $\epsappr$-approximation, with support volume at most
$2/\epsappr$, in respectively
$\softO(1/(\epsappr\sqrt\alpha))$ deterministic and expected Las Vegas
work. In particular $\vzero\leq\widehat\vpi\leq\vpi$ and the
semantic error is at most $\epsappr$.
\end{corollary}
\begin{proof}
Set $\rho=\epsappr/2$ and request objective tolerance
$\alpha\epsappr^2/8$. The nontrivial deterministic execution reaches a
stage handoff, since this tolerance is below $\alpha/2$.
\Cref{eq:det-repaired-invariant} gives
$\vzero\leq\vb-\mQ\widehat\vx\leq2\alpha\rho\vomega$;
for the randomized method use \cref{cor:randomized-safe-repair}.
Define
\[
 \vr^{\mathrm{mass}}:=\alpha^{-1}\mD^{1/2}(\vb-\mQ\widehat\vx).
\]
Then \cref{eq:acl-approximation-contract} follows by rearranging the
PageRank equation, and
$\vr^{\mathrm{mass}}\leq2\rho\mD\one=\epsappr\mD\one$.
Inverse positivity gives the one-sided solution and semantic bounds.
The support and work follow from the main theorems.
If $\rho d_v\geq1$, zero has residual $\eunit{v}\leq\rho\mD\one$.
At $\alpha=1$ the exact one-coordinate RPPR solution has the same certificate.
\end{proof}

The corollary supplies the numerical approximation property used in ACL
local clustering, in addition to the semantic guarantee. Applying a
particular sweep-cut theorem still requires its graph, seed, and parameter
hypotheses; those hypotheses are not replaced by an objective-gap statement.

\subsection{An approximate support as a linear-system envelope}

For any nonempty set $\gU\subseteq\gV$, define the principal PPR solution
and its zero extension by
\begin{equation*}
 \vu_\gU:=\mQ_{\gU\gU}^{-1}\vb_\gU,
 \qquad \widetilde\vu_\gU\in\R^n.
\end{equation*}
An empty set has the zero solution. The following comparison does not
require $\gU$ to contain the full RPPR support or to be reachable by pivots.

\begin{lemma}[Approximate-envelope linearization]
\label{lem:approximate-envelope}
Let
\begin{equation*}
 \delta_\gU:=\max_{i\notin \gU}(\vx_\rho^*)_i/\sqrt{d_i},
\end{equation*}
where the maximum is zero if $\gU=\gV$. Then
\begin{equation*}
 \vzero\leq\widetilde\vu_\gU\leq\vx_0^*,\qquad
 \norm{\mD^{-1/2}(\vx_0^*-\widetilde\vu_\gU)}_\infty
       \leq\rho+\delta_\gU.
\end{equation*}
\end{lemma}
\begin{proof}
Write $\gU^c=\gV\setminus\gU$.
Outside $\gU$, \cref{eq:rppr-bias} gives
$(\vx_0^*)_i\leq(\rho+\delta_\gU)\sqrt{d_i}$.
On $\gU$, the difference
$\vw_\gU=(\vx_0^*)_\gU-\vu_\gU$ satisfies
\[
 \mQ_{\gU\gU}\vw_\gU=-\mQ_{\gU\gU^c}(\vx_0^*)_{\gU^c}\geq\vzero.
\]
Both nonnegativity of $\vu_\gU$ and its comparison with $(\vx_0^*)_\gU$
follow from inverse positivity. The outside bound and Stieltjes signs imply
\[
 \mQ_{\gU\gU}\vw_\gU
 \leq-\mQ_{\gU\gU^c}(\rho+\delta_\gU)\vomega_{\gU^c}
 \leq\mQ_{\gU\gU}(\rho+\delta_\gU)\vomega_\gU.
\]
The last inequality uses $\mQ\vomega=\alpha\vomega\geq\vzero$.
Multiplication by the nonnegative principal inverse proves the inside
bound. The empty-set case is the already established outside bound.
\end{proof}

\begin{lemma}[A sparse objective certificate gives a set certificate]
\label{lem:face-energy-envelope}
If a nonnegative $\vz$ supported on $\gU$ satisfies
$F_\rho(\vz)-F_\rho(\vx_\rho^*)\leq\Gamma_\gU$, then
\begin{equation*}
 \delta_\gU\leq\sqrt{2\Gamma_\gU/\alpha},\qquad
 \norm{\mD^{-1/2}(\vx_0^*-\widetilde\vu_\gU)}_\infty
       \leq\rho+\sqrt{2\Gamma_\gU/\alpha}.
\end{equation*}
\end{lemma}
\begin{proof}
Strong convexity bounds $\norm{\vz-\vx_\rho^*}_2$ by
$\sqrt{2\Gamma_\gU/\alpha}$. Since $\vz$ vanishes outside $\gU$ and
$d_i\geq1$, this bounds $\delta_\gU$. Apply
\cref{lem:approximate-envelope}.
\end{proof}

The lemma applies to either main algorithm's output. For the randomized
algorithm, one may also use the exact reachable-face gap certified in the
proof of \cref{thm:local-rppr}; the numerical point need not be retained.

\begin{algorithm}[t]
\caption{PPR with a set-only handoff}
\label{alg:two-stage-ppr}
\begin{algorithmic}[1]
\STATE Choose $\rho,\varepsilon_{\mathrm{env}},\varepsilon_{\mathrm{lin}}>0$ with
       $\rho+\varepsilon_{\mathrm{env}}+\varepsilon_{\mathrm{lin}}\leq\epsppr$.
\STATE Run a main RPPR algorithm to objective gap
       $\epsobj=\alpha\varepsilon_{\mathrm{env}}^2/2$.
\STATE Retain only its output support $\gU$ (or its certified reachable face).
       Discard the changing numerical state. If $\gU=\emptyset$, return zero.
\STATE Build the principal PPR system on $\gU$ using original degrees and
       solve it to Euclidean error at most $\varepsilon_{\mathrm{lin}}$.
\STATE Project its numerical solution onto $\R_+^\gU$ and zero-extend to
       $\widehat\vx$. Return $\widehat\vpi=\mD^{1/2}\widehat\vx$.
\end{algorithmic}
\end{algorithm}

\begin{theorem}[Deterministic and randomized two-stage completion]
\label{thm:two-stage-completion}
For $0<\epsppr<1$, \cref{alg:two-stage-ppr} attains the semantic target
whenever its two certified stages return. With equal budgets
$\rho=\varepsilon_{\mathrm{env}}=\varepsilon_{\mathrm{lin}}=\epsppr/3$, there
is a deterministic implementation with
$\softO(1/(\epsppr\sqrt\alpha))$ fully charged work and a Las Vegas
implementation with the same expected work. The second stage may use
ordinary preconditioned conjugate gradient deterministically or a
residual-certified SDD solve randomly.
\end{theorem}
\begin{proof}
\Cref{lem:face-energy-envelope} bounds the error of the exact principal
solution by $\rho+\varepsilon_{\mathrm{env}}$.
Orthant projection cannot increase distance from that nonnegative solution.
A Euclidean error at most $\varepsilon_{\mathrm{lin}}$ therefore contributes
at most that amount to the degree-normalized maximum norm. This proves
accuracy, including the empty-set case.

Building the principal matrix from only the retained labels costs
$\widetilde{\mathcal{O}}(\vol(\gU))$; no prior solver state is needed. Its energy
satisfies
\[
 \norm{\widetilde\vu_\gU}_{\mQ}^2
 =\vb_\gU^\trans\mQ_{\gU\gU}^{-1}\vb_\gU\leq\alpha.
\]
Preconditioned CG, initialized at zero, on the degree system
$\mH_{\gU\gU}\vy_\gU=\mD_\gU^{1/2}\vb_\gU$, with preconditioner $\mD_\gU$,
is equivalent to CG on $\mQ_{\gU\gU}$, whose spectrum is in $[\alpha,1]$.
Its energy minimization over the current Krylov space, together with a
Chebyshev residual polynomial on $[\alpha,1]$, bounds the Euclidean error
after $j$ steps by
\[
 2\left(\frac{1-\sqrt\alpha}{1+\sqrt\alpha}\right)^j.
\]
Indeed the polynomial's maximum absolute residual bounds the energy error
relative to its initial value, which is at most $\sqrt\alpha$; converting
to Euclidean error divides by $\sqrt\alpha$.
The polynomial is obtained by the same construction as
\cref{lem:chebyshev-inverse}, with upper endpoint one. Exact zero residual
terminates CG immediately, including the degenerate spectral case.
Thus $\mathcal{O}(\alpha^{-1/2}\log(2/\varepsilon_{\mathrm{lin}}))$ iterations suffice.
Each costs $\mathcal{O}(\vol(\gU))$ arithmetic and repeated incidence reads, in
addition to local indexing logarithms. This gives the deterministic tail
charge
\begin{equation*}
 \softO\!\left(\frac{\vol(\gU)}{\sqrt\alpha}
                  \log\frac2{\varepsilon_{\mathrm{lin}}}\right).
\end{equation*}

For the randomized tail, choose a relative SDD energy tolerance no larger
than $\sqrt\alpha\,\varepsilon_{\mathrm{lin}}$. A good call then has active
residual norm at most $\alpha\varepsilon_{\mathrm{lin}}$.
Accept only after verifying this residual inequality by a charged active-row
scan. Every accepted call has Euclidean error at most
$\varepsilon_{\mathrm{lin}}$, since $\mQ_{\gU\gU}\succeq\alpha\mI$.
The resulting expected tail charge is
\begin{equation*}
 \softO\!\left(\vol(\gU)
          \log\frac2{\sqrt\alpha\,\varepsilon_{\mathrm{lin}}}\right).
\end{equation*}
The tolerance can be rounded down dyadically to avoid computing a root.
Independent certified retries give Las Vegas completion; alternatively,
split a requested failure budget between the two stages and report failure
explicitly if a cap is reached.

Finally, $\vol(\gU)\leq1/\rho$. Combining either tail charge with its
corresponding Stage~I theorem and the equal budgets proves the result.
All reconstruction, projection, state disposal, and output are included.
\end{proof}

For the shortest execution, return the numerical Stage~I answer directly
using \cref{eq:direct-ppr-budgets}. The set-only variant is useful when
support discovery and final linear accuracy are implemented by different
routines. Its correctness follows from a quantitative envelope certificate,
so exact identification of the regularized support is unnecessary.

\subsection{Discussion}\label{sec:conclusion}
We have given deterministic and randomized local algorithms that attain
the accelerated running-time target for regularized PageRank. The concurrent
work of \citet{cui2026accelerating} attains the same randomized scale;
our principal distinction is the deterministic bound with only
polylogarithmic overhead and no SDD oracle. An explicit
regularization-bias and objective-error conversion gives the corresponding
degree-normalized PPR guarantee. The two methods control locality in
different ways: deterministic continuation bounds cumulative work along
an accelerated trajectory, while randomized threshold batching bounds the
number of support-safe restricted systems. Both bounds apply to total
work in the stated local-access model.

The deterministic method also has a specified bounded-arithmetic
realization. Developing an effective fixed-precision implementation and
comparing practical constants remain useful next steps.

Two algorithmic questions follow directly from the results. First, can a
deterministic method achieve the same support-adaptive minimum as the
randomized method? Second, can the randomized method reuse numerical work
across nested active sets while maintaining all boundary information
locally? These improvements may help determine the optimal joint
dependence on regularization, teleportation, and accuracy, which is not
settled by the present upper bounds. The locality counterexamples for
unmodified accelerated methods also motivate identifying which constraints
or graph structures make classical recurrences efficient.

\section*{Research Assistance and Responsibility}
Baojian Zhou directed the research and is the sole author of this manuscript.
OpenAI's ChatGPT-5.6 Sol and ChatGPT-6-Astra were used through Codex and
related research tools to assist with the research and manuscript
preparation. Their assistance included proposing and refining proof
strategies, deriving and examining inequalities, searching for
counterexamples, comparing results with primary sources, writing
computational verification programs, and drafting and revising the text
and LaTeX. This use included substantive mathematical development and
critical examination of proofs, not only language editing. The randomized
and deterministic arguments were developed in research notes before being
revised and assembled into this manuscript.

Model-assisted reviews and computational checks were used to look for
errors and test intermediate claims and special cases. Neither agreement
between models nor successful computational tests establishes a theorem
or constitutes formal verification or external peer review. Claims of
correctness rest on the mathematical arguments presented in the paper,
not on the authority of a model response.

The author retains full responsibility for the correctness and completeness
of the mathematical statements and proofs, the accuracy of citations and
attribution, the description of computational evidence, and the final
submitted text. The AI systems are acknowledged as research tools and are
not listed as authors.

\paragraph{Development record.}
The repository preserves an early randomized proof in
\href{https://github.com/baojian/hybrid-local-solver/commit/3fa84552e6582604a214021f20310df1c6290c9e}{commit \texttt{3fa8455}}
(August~30, 2026, UTC+8), including the threshold-batch Cholesky theorem
and the charged RPPR algorithm in the \texttt{active\_edge\_lcp} note.
The deterministic continuation proof is preserved in
\href{https://github.com/baojian/hybrid-local-solver/commit/0f0ac32d70290f2c6877b4cf1faa7da85e2835ee}{commit \texttt{0f0ac32}}
(September~6, 2026, UTC+8), in the original sources of the independent
deterministic note. Both arguments were integrated into the active
manuscript by
\href{https://github.com/baojian/hybrid-local-solver/commit/7e1b4e77351b361e0c700ce3186c1534d2023f0e}{commit \texttt{7e1b4e7}}
(September~6, 2026, UTC+8).
GitHub-hosted workflow and pull-request records corroborate the presence
of these historical snapshots before September~10; the detailed evidence
and its limitations are recorded in the repository's dated release audit.
These are development records, not claims of earlier public disclosure
or external proof certification. Later exposition, certificates, and
implementation refinements should be distinguished from the early proofs.

\bibliography{references}

@inproceedings{fountoulakis2022open,
  title     = {Open Problem: Running Time Complexity of Accelerated
               {$\ell_1$}-Regularized {PageRank}},
  author    = {Fountoulakis, Kimon and Yang, Shenghao},
  booktitle = {Proceedings of the 35th Conference on Learning Theory},
  series    = {Proceedings of Machine Learning Research},
  volume    = {178},
  pages     = {5630--5632},
  publisher = {PMLR},
  year      = {2022},
  url       = {https://proceedings.mlr.press/v178/open-problem-fountoulakis22a.html}
}

@article{beck2009fast,
  title     = {A Fast Iterative Shrinkage-Thresholding Algorithm for Linear
               Inverse Problems},
  author    = {Beck, Amir and Teboulle, Marc},
  journal   = {SIAM Journal on Imaging Sciences},
  volume    = {2},
  number    = {1},
  pages     = {183--202},
  year      = {2009},
  doi       = {10.1137/080716542},
  url       = {https://doi.org/10.1137/080716542}
}

@inproceedings{allenzhu2017linear,
  title     = {Linear Coupling: An Ultimate Unification of Gradient and
               Mirror Descent},
  author    = {Allen-Zhu, Zeyuan and Orecchia, Lorenzo},
  booktitle = {8th Innovations in Theoretical Computer Science Conference
               (ITCS 2017)},
  series    = {Leibniz International Proceedings in Informatics},
  volume    = {67},
  pages     = {3:1--3:22},
  year      = {2017},
  publisher = {Schloss Dagstuhl -- Leibniz-Zentrum f{\"u}r Informatik},
  doi       = {10.4230/LIPIcs.ITCS.2017.3},
  url       = {https://doi.org/10.4230/LIPIcs.ITCS.2017.3}
}

@article{lin2018catalyst,
  title   = {Catalyst Acceleration for First-Order Convex Optimization: From
             Theory to Practice},
  author  = {Lin, Hongzhou and Mairal, Julien and Harchaoui, Zaid},
  journal = {Journal of Machine Learning Research ({JMLR})},
  volume  = {18},
  number  = {212},
  pages   = {1--54},
  year    = {2018},
  url     = {https://jmlr.org/papers/v18/17-748.html}
}

@inproceedings{wang2024revisiting,
  title     = {Revisiting Local Computation of {PageRank}: Simple and Optimal},
  author    = {Wang, Hanzhi and Wei, Zhewei and Wen, Ji-Rong and Yang, Mingji},
  booktitle = {Proceedings of the 56th Annual ACM Symposium on Theory of
               Computing},
  pages     = {911--922},
  publisher = {Association for Computing Machinery},
  year      = {2024},
  doi       = {10.1145/3618260.3649661},
  url       = {https://doi.org/10.1145/3618260.3649661}
}

@inproceedings{bai2024faster,
  title     = {Faster Local Solvers for Graph Diffusion Equations},
  author    = {Bai, Jiahe and Zhou, Baojian and Yang, Deqing and Xiao, Yanghua},
  booktitle = {Advances in Neural Information Processing Systems},
  volume    = {37},
  year      = {2024},
  doi       = {10.52202/079017-0085},
  url       = {https://proceedings.neurips.cc/paper_files/paper/2024/hash/0506ad3d1bcc8398a920db9340f27fe4-Abstract-Conference.html}
}

@inproceedings{zhou2024iterative,
  title     = {Iterative Methods via Locally Evolving Set Process},
  author    = {Zhou, Baojian and Sun, Yifan and Babanezhad Harikandeh, Reza
               and Guo, Xingzhi and Yang, Deqing and Xiao, Yanghua},
  booktitle = {Advances in Neural Information Processing Systems},
  volume    = {37},
  year      = {2024},
  doi       = {10.52202/079017-4494},
  url       = {https://proceedings.neurips.cc/paper_files/paper/2024/hash/fffe5a7804c40465ef2432386850c2c7-Abstract-Conference.html}
}

@inproceedings{huang2025accelerated,
  title     = {Accelerated Evolving Set Processes for Local {PageRank}
               Computation},
  author    = {Huang, Binbin and Luo, Luo and Xiao, Yanghua and Yang, Deqing
               and Zhou, Baojian},
  booktitle = {Advances in Neural Information Processing Systems},
  volume    = {38},
  year      = {2025},
  url       = {https://proceedings.neurips.cc/paper_files/paper/2025/hash/946ecab300b0695fe24b53a92e632935-Abstract-Conference.html}
}

@inproceedings{martinezrubio2023accelerated,
  title     = {Accelerated and Sparse Algorithms for Approximate Personalized
               {PageRank} and Beyond},
  author    = {Mart{\'i}nez-Rubio, David and Wirth, Elias and Pokutta, Sebastian},
  booktitle = {Proceedings of the 36th Conference on Learning Theory},
  series    = {Proceedings of Machine Learning Research},
  volume    = {195},
  pages     = {2852--2876},
  publisher = {PMLR},
  year      = {2023},
  url       = {https://proceedings.mlr.press/v195/martinez-rubio23b.html}
}

@article{fountoulakis2026complexity,
  title         = {Complexity of Classical Acceleration for
                   {$\ell_1$}-Regularized {PageRank}},
  author        = {Fountoulakis, Kimon and Mart{\'i}nez-Rubio, David},
  journal       = {arXiv preprint arXiv:2602.21138},
  year          = {2026},
  eprint        = {2602.21138},
  archiveprefix = {arXiv},
  primaryclass  = {math.OC},
  doi           = {10.48550/arXiv.2602.21138},
  url           = {https://arxiv.org/abs/2602.21138},
  note          = {Version 2}
}

@article{wei2026simple,
  title         = {A Simple Active-Set Method for {PageRank}-Based Local Graph
                   Clustering},
  author        = {Wei, Zhewei and Yang, Mingji},
  journal       = {arXiv preprint arXiv:2608.16339},
  year          = {2026},
  eprint        = {2608.16339},
  archiveprefix = {arXiv},
  primaryclass  = {cs.DS},
  doi           = {10.48550/arXiv.2608.16339},
  url           = {https://arxiv.org/abs/2608.16339},
  note          = {Version 1}
}

@article{fountoulakis2019variational,
  title   = {Variational Perspective on Local Graph Clustering},
  author  = {Fountoulakis, Kimon and Roosta-Khorasani, Farbod and Shun, Julian
             and Cheng, Xiang and Mahoney, Michael W.},
  journal = {Mathematical Programming},
  volume  = {174},
  number  = {1--2},
  pages   = {553--573},
  year    = {2019},
  doi     = {10.1007/s10107-017-1214-8},
  url     = {https://doi.org/10.1007/s10107-017-1214-8}
}

@article{karkkainen2003augmented,
  title   = {Augmented {Lagrangian} Active Set Methods for Obstacle Problems},
  author  = {K{\"a}rkk{\"a}inen, Tommi and Kunisch, Karl and Tarvainen, Pasi},
  journal = {Journal of Optimization Theory and Applications},
  volume  = {119},
  number  = {3},
  pages   = {499--533},
  year    = {2003},
  doi     = {10.1023/B:JOTA.0000006687.57272.b6}
}

@inproceedings{andersen2006local,
  title     = {Local Graph Partitioning Using {PageRank} Vectors},
  author    = {Andersen, Reid and Chung, Fan R. K. and Lang, Kevin J.},
  booktitle = {47th Annual IEEE Symposium on Foundations of Computer Science
               (FOCS 2006)},
  pages     = {475--486},
  publisher = {IEEE Computer Society},
  year      = {2006},
  doi       = {10.1109/FOCS.2006.44},
  url       = {https://doi.org/10.1109/FOCS.2006.44}
}

@article{ha2021statistical,
  title   = {Statistical Guarantees for Local Graph Clustering},
  author  = {Ha, Wooseok and Fountoulakis, Kimon and Mahoney, Michael W.},
  journal = {Journal of Machine Learning Research ({JMLR})},
  volume  = {22},
  number  = {148},
  pages   = {1--54},
  year    = {2021},
  url     = {https://jmlr.org/papers/v22/20-029.html}
}

@inproceedings{morris2003evolving,
  title     = {Evolving Sets and Mixing},
  author    = {Morris, Ben and Peres, Yuval},
  booktitle = {Proceedings of the 35th Annual ACM Symposium on Theory of
               Computing},
  pages     = {279--286},
  publisher = {Association for Computing Machinery},
  year      = {2003},
  doi       = {10.1145/780542.780585},
  url       = {https://doi.org/10.1145/780542.780585}
}

@inproceedings{andersen2009finding,
  title     = {Finding Sparse Cuts Locally Using Evolving Sets},
  author    = {Andersen, Reid and Peres, Yuval},
  booktitle = {Proceedings of the 41st Annual ACM Symposium on Theory of
               Computing},
  pages     = {235--244},
  publisher = {Association for Computing Machinery},
  year      = {2009},
  doi       = {10.1145/1536414.1536449},
  url       = {https://doi.org/10.1145/1536414.1536449}
}

@inproceedings{spielman2004nearly,
  title     = {Nearly-Linear Time Algorithms for Graph Partitioning, Graph
               Sparsification, and Solving Linear Systems},
  author    = {Spielman, Daniel A. and Teng, Shang-Hua},
  booktitle = {Proceedings of the 36th Annual ACM Symposium on Theory of
               Computing},
  pages     = {81--90},
  publisher = {Association for Computing Machinery},
  year      = {2004},
  doi       = {10.1145/1007352.1007372},
  url       = {https://doi.org/10.1145/1007352.1007372}
}

@inproceedings{koutis2011nearly,
  title     = {A Nearly-$m\log n$ Time Solver for {SDD} Linear Systems},
  author    = {Koutis, Ioannis and Miller, Gary L. and Peng, Richard},
  booktitle = {2011 IEEE 52nd Annual Symposium on Foundations of Computer
               Science},
  pages     = {590--598},
  publisher = {IEEE Computer Society},
  year      = {2011},
  doi       = {10.1109/FOCS.2011.53},
  url       = {https://doi.org/10.1109/FOCS.2011.53}
}

@article{spielman2013local,
  title   = {A Local Clustering Algorithm for Massive Graphs and Its
             Application to Nearly Linear Time Graph Partitioning},
  author  = {Spielman, Daniel A. and Teng, Shang-Hua},
  journal = {SIAM Journal on Computing ({SICOMP})},
  volume  = {42},
  number  = {1},
  pages   = {1--26},
  year    = {2013},
  doi     = {10.1137/080744888},
  url     = {https://doi.org/10.1137/080744888}
}

@article{andersen2007using,
  title   = {Using {PageRank} to Locally Partition a Graph},
  author  = {Andersen, Reid and Chung, Fan R. K. and Lang, Kevin J.},
  journal = {Internet Mathematics},
  volume  = {4},
  number  = {1},
  pages   = {35--64},
  year    = {2007},
  doi     = {10.1080/15427951.2007.10129139},
  url     = {https://doi.org/10.1080/15427951.2007.10129139}
}

@inproceedings{wei2024absolute,
  title = {Approximating Single-Source Personalized {PageRank} with Absolute Error Guarantees},
  author = {Wei, Zhewei and Wen, Ji-Rong and Yang, Mingji},
  booktitle = {27th International Conference on Database Theory (ICDT 2024)},
  series = {Leibniz International Proceedings in Informatics},
  volume = {290},
  pages = {9:1--9:19},
  year = {2024},
  doi = {10.4230/LIPIcs.ICDT.2024.9},
  url = {https://drops.dagstuhl.de/entities/document/10.4230/LIPIcs.ICDT.2024.9},
  note = {Full version: \url{https://arxiv.org/abs/2401.01019v1}}
}

@article{bertram2026undirected,
  title = {Personalized {PageRank} Estimation in Undirected Graphs},
  author = {Bertram, Christian and Jensen, Mads Vestergaard},
  journal = {arXiv preprint arXiv:2602.10843},
  year = {2026},
  eprint = {2602.10843},
  archiveprefix = {arXiv},
  url = {https://arxiv.org/abs/2602.10843},
  note = {Version 1, February 11, 2026}
}

@article{jiang2026nearoptimality,
  title = {Near-Optimality for Single-Source Personalized {PageRank}},
  author = {Jiang, Xinpeng and Liu, Haoyu and Luo, Siqiang and Xiao, Xiaokui},
  journal = {Proceedings of the ACM on Management of Data},
  volume = {4},
  number = {2},
  articleno = {110},
  numpages = {55},
  pages = {110:1--110:55},
  year = {2026},
  doi = {10.1145/3801906},
  url = {https://arxiv.org/abs/2507.14462},
  note = {Preprint version: arXiv:2507.14462v5, April 12, 2026}
}

@inproceedings{kwok2026asymmetric,
  title = {On Solving Asymmetric Diagonally Dominant Linear Systems in Sublinear Time},
  author = {Kwok, Tsz Chiu and Wei, Zhewei and Yang, Mingji},
  booktitle = {17th Innovations in Theoretical Computer Science Conference (ITCS 2026)},
  series = {Leibniz International Proceedings in Informatics},
  volume = {362},
  pages = {89:1--89:25},
  year = {2026},
  doi = {10.4230/LIPIcs.ITCS.2026.89},
  url = {https://drops.dagstuhl.de/entities/document/10.4230/LIPIcs.ITCS.2026.89}
}

@inproceedings{chen2021diffusion,
  author = {Chen, Li and Peng, Richard and Wang, Di},
  title = {$\ell_2$-norm Flow Diffusion in Near-Linear Time},
  booktitle = {2021 IEEE 62nd Annual Symposium on Foundations of Computer Science (FOCS)},
  pages = {540--549},
  publisher = {IEEE},
  year = {2021},
  doi = {10.1109/FOCS52979.2021.00060},
  url = {https://arxiv.org/abs/2105.14629}
}

@inproceedings{vladu2025barrier,
  author = {Vladu, Adrian},
  title = {Breaking the Barrier of Self-Concordant Barriers: Faster Interior Point Methods for {M}-Matrices},
  booktitle = {Proceedings of the 57th Annual ACM Symposium on Theory of Computing},
  pages = {2213--2224},
  publisher = {Association for Computing Machinery},
  year = {2025},
  doi = {10.1145/3717823.3718255},
  url = {https://arxiv.org/abs/2504.20619}
}

@inproceedings{lin2024constraints,
  title = {Faster Accelerated First-order Methods for Convex Optimization with Strongly Convex Function Constraints},
  author = {Lin, Zhenwei and Deng, Qi},
  booktitle = {Advances in Neural Information Processing Systems},
  volume = {37},
  year = {2024},
  url = {https://proceedings.neurips.cc/paper_files/paper/2024/file/8d8e060d9a3312ae12f42adf0da6ec7c-Paper-Conference.pdf}
}

@inproceedings{thorup2026pagerank,
  author = {Thorup, Mikkel and Wang, Hanzhi and Wei, Zhewei and Yang, Mingji},
  title = {{PageRank} Centrality in Directed Graphs with Bounded In-Degree},
  booktitle = {Proceedings of the 2026 Annual ACM-SIAM Symposium on Discrete Algorithms (SODA)},
  pages = {3820--3840},
  publisher = {Society for Industrial and Applied Mathematics},
  year = {2026},
  doi = {10.1137/1.9781611978971.140},
  url = {https://epubs.siam.org/doi/abs/10.1137/1.9781611978971.140}
}

@misc{thorup2026instance,
  author = {Thorup, Mikkel and Wang, Hanzhi},
  title = {Instance-Optimality of Bidirectional {PageRank} Estimation},
  year = {2026},
  eprint = {2512.16087},
  archivePrefix = {arXiv},
  primaryClass = {cs.DS},
  note = {Version 6, August 3, 2026. To appear in FOCS 2026},
  url = {https://arxiv.org/abs/2512.16087v6}
}

@misc{li2026resistance,
  author = {Li, Rong-Hua and Yang, Yichun},
  title = {Improved Algorithm for Counting Spanning Trees by $\ell_1$-Regularized Resistance},
  year = {2026},
  eprint = {2609.03574},
  archivePrefix = {arXiv},
  primaryClass = {cs.DS},
  note = {Version 2, September 4, 2026},
  url = {https://arxiv.org/abs/2609.03574v2}
}

@misc{cui2026accelerating,
  author = {Guanyu Cui and Zhewei Wei and Mingji Yang},
  title = {Accelerating the Local Push Primitive for {PageRank} Computation},
  year = {2026},
  eprint = {2609.12076},
  archivePrefix = {arXiv},
  primaryClass = {cs.DS},
  note = {Version 1, submitted September 10, 2026},
  url = {https://arxiv.org/abs/2609.12076v1}
}

\clearpage
\appendix
\section{Detailed Related Work and Normalization Comparisons}
\label{app:related-work}
\subsection{Local PageRank, regularization, and accuracy criteria}

The local-push work of \citet{andersen2006local,andersen2007using} is the
basic numerical precedent. Its residual invariant and positive mass decrease
lead to $\mathcal{O}(1/(\alpha\epsappr))$ adjacency work and a sparse approximation.
The same diffusion supports local graph partitioning through degree-normalized
sweep cuts. Earlier and complementary local clustering methods use truncated
random walks or evolving random sets
\citep{spielman2013local,morris2003evolving,andersen2009finding}.
Their conductance guarantees and their numerical work guarantees are distinct
objects. Our theorems concern RPPR objective accuracy and PPR solution error;
we do not infer a new conductance theorem merely from those numerical bounds.

\citet{fountoulakis2019variational} connect local PageRank to a weighted
$\ell_1$-regularized quadratic. Their KKT structure, nonnegative optimum,
monotone ISTA trajectory, and support-volume theorem explain why an ordinary
proximal method can be implemented locally. The objective-gap recurrence
in the proof of their Theorem~3, after equation~(23), together with
curvature at least $\alpha$ and $\mathcal{O}(k_*+V_*)$ work per local step, gives
the $\widetilde{\mathcal{O}}(V_*/\alpha)$ row in \cref{tab:lineage}.
The stronger mass identity in \cref{eq:rppr-order-mass} is a direct
consequence of the same structure.
\citet[Lemma~4 and Section~5]{ha2021statistical} study the regularization path and its statistical
properties; monotonicity as the regularizer decreases is also central to
our continuation argument. We provide the elementary order proof to fix
its precise role and normalization.

Several PageRank conventions coexist. To compare with non-lazy formulations,
let $\mP=\mA\mD^{-1}$ and set
\begin{equation}
 \alpha_{\mathrm{nl}}=\frac{2\alpha}{1+\alpha},\qquad
 \mL_{\alpha_{\mathrm{nl}}}=\mD-(1-\alpha_{\mathrm{nl}})\mA,
 \qquad \vy=\mD^{-1/2}\vx.
 \label{eq:related-nonlazy-map}
\end{equation}
The lazy and non-lazy systems then have the same PageRank vector, and
\begin{equation*}
 \mD^{1/2}\mQ\mD^{1/2}
   =\frac{1+\alpha}{2}\mL_{\alpha_{\mathrm{nl}}}.
\end{equation*}
In particular, for the degree-coordinate objective used by
\citet{wei2026simple,cui2026accelerating},
\begin{equation*}
 \begin{aligned}
 \psi_{\alpha_{\mathrm{nl}},\rho}(\vy)
   &:=\tfrac12\vy^\trans\mL_{\alpha_{\mathrm{nl}}}\vy
      -\alpha_{\mathrm{nl}}\eunit{v}^\trans\vy
      +\alpha_{\mathrm{nl}}\rho\norm{\mD\vy}_1,\\
 F_\rho(\mD^{1/2}\vy)
   &=\tfrac{1+\alpha}{2}\psi_{\alpha_{\mathrm{nl}},\rho}(\vy).
 \end{aligned}
\end{equation*}
Thus $\rho$ is unchanged and objective tolerances differ by a factor in
$[1/2,1]$. The rescaled symmetric lazy system of
\citet{zhou2024iterative} instead multiplies both $\mQ$ and $\vb$ by
$2/(1+\alpha)$, leaving $\alpha$ and the solution unchanged.
These algebraic conversions do not identify an objective tolerance with a
normalized residual or semantic output error; those require
\cref{eq:objective-to-ppr,eq:ppr-residual-certificate}.

\subsection{Accelerated and SDD-based active sets}

\citet{martinezrubio2023accelerated} give two sparse algorithms for convex
quadratics with an $M$-matrix Hessian, including RPPR. Their CDPR algorithm
computes the exact optimum and retains conjugate directions as its space
expands. Their ASPR algorithm solves successive restricted problems by
accelerated projected gradient, repairs approximate points to preserve safe
expansion, and attains additive objective accuracy. This established that
acceleration can coexist with support-sensitive computation, with a
tradeoff in the dependence on support size.

For a precise comparison, write $k_*=|\mathcal S_\rho^*|$,
$V_*=\vol(\mathcal S_\rho^*)$, and
$\widetilde V_*:=\nnz(\mQ_{\mathcal S_\rho^*,\mathcal S_\rho^*})$.
Their external-volume convention counts nonzeros in full columns of the
Hessian, including diagonals; for $0<\alpha<1$ it is
$V_*+k_*=\Theta(V_*)$ on our graphs. Their internal volume is
$\widetilde V_*\leq2V_*$. At $\alpha=1$, the matrix is diagonal and the
exact solution is handled separately.
Theorem~8 of that paper gives ASPR work
$\widetilde{\mathcal{O}}(k_*\widetilde V_*/\sqrt\alpha+k_*V_*)$;
Theorem~4 gives CDPR work $\mathcal{O}(k_*^3+k_*V_*)$ and space $\mathcal{O}(k_*^2)$.
The two bounds have different output guarantees: CDPR is exact,
whereas ASPR permits a specified additive objective error.

\citet[Theorem~1.3]{wei2026simple} (arXiv version~1, August~17, 2026) use nearly-linear SDD solves on growing
active sets. Their RPPR work is
$\widetilde{\mathcal{O}}(k_*\widetilde V_*+k_*V_*)=\widetilde{\mathcal{O}}(k_*V_*)$,
with probability specified by the caller. This removes the explicit
$1/\sqrt\alpha$ factor in the supplied-face solve. Their output additionally
has a nonnegative residual bounded by $2\rho$ in degree-normalized units,
which gives the ACL approximation property needed by their clustering
application. They also obtain $\widetilde{\mathcal{O}}(1/\epsappr^2)$ work for ACL
approximate PPR. This tradeoff can improve the teleportation dependence
substantially when the requested locality scale is coarse.
Their deterministic-solver remark on page~4 gives a deterministic version
with additional $k_*^{o(1)}$ or $(1/\epsappr)^{o(1)}$ factors. Therefore
the distinction here is the particular
$\widetilde{\mathcal{O}}(1/(\rho\sqrt\alpha))$ bound with only polylogarithmic
overhead, not the mere availability of a deterministic active-set method.

\Cref{tab:lineage} summarizes these bounds under the stated output and
work conventions.

Our randomized analysis limits the number of restricted systems that require
processing. Threshold batching retains safe support expansion and replaces
the cardinality-only count by a spectral depth bound, while still benefiting
from a small $k_*$. All principal systems use original degrees, and every
boundary refresh is included. The SDD solver itself is supplied by
\citet{koutis2011nearly}, building on the nearly-linear solver literature
\citep{spielman2004nearly}; our batch-depth argument bounds how many
locally discovered systems need solving. Our deterministic algorithm
establishes the accelerated local bound through a separate recurrence,
without requiring a deterministic nearly-linear SDD primitive.

These results do not establish a universal lower bound with
$1/\sqrt\alpha$ dependence. Even the existing $\widetilde{\mathcal{O}}(\rho^{-2})$
SDD bound is smaller than $\widetilde{\mathcal{O}}(1/(\rho\sqrt\alpha))$ in
some parameter regimes. The support-adaptive minimum in
\cref{eq:randomized-adaptive-work} records this distinction.

\subsection{Concurrent accelerated local push}

\citet{cui2026accelerating} give a directly overlapping randomized result
in arXiv:2609.12076v1, submitted September~10, 2026. Their Theorems~1.2
and~1.3 (PDF pages~3--4) provide accelerated ACL approximation and RPPR
objective accuracy in the same local graph model. The normalization map
in \cref{eq:related-nonlazy-map} leaves $\rho$ unchanged and multiplies
their degree-coordinate objective by $(1+\alpha)/2$. Consequently a
target gap $\epsobj$ here corresponds to their gap
$2\epsobj/(1+\alpha)$: normalization does not separate the results.

Their proof of Theorem~1.3 (PDF page~18) bounds the number of face solves
by
\[
 \mathcal{O}(\min\{k_*,\alpha^{-1/2}\log(1/(\alpha\eta))\})
\]
after this parameter conversion, where $\eta$ is their stopping tolerance.
Each phase has nearly-linear local solve and boundary cost. Since
$\widetilde V_*\leq2V_*$, this also implies
$\softO(V_*\min\{k_*,\alpha^{-1/2}\})$. This is a consequence of their
proof, not an additional advantage of our support-adaptive bound.
Their Lemma~3.10 (PDF pages~18--19) also treats sparse sources with
additive input work for ACL approximation.

The proof mechanisms differ. Their Lemmas~3.7 and~3.9 compare potential
decreases across consecutive blocks of expansions. Our
\cref{thm:threshold-batch-depth} uses a block-bidiagonal Cholesky
comparison and Chebyshev inverse decay. Both use thresholded, support-safe
active sets and supplied-face SDD solves. We provide explicit residual
certification and Las Vegas retries, without claiming that such wrapping
is intrinsically unavailable to their method.

The deterministic-solver remark on PDF page~4 introduces a further
subpolynomial factor when their randomized SDD solver is replaced.
Our deterministic continuation attains
$\softO(1/(\rho\sqrt\alpha))$ with only polylogarithmic overhead and
without an SDD primitive; \cref{app:bounded-arithmetic} separately gives
its rational-input realization. This is the principal running-time
distinction from the concurrent work. Repository provenance is recorded
below; a development timestamp is not a claim of earlier public disclosure.

\subsection{Classical acceleration and evolving-support methods}

FISTA \citep{beck2009fast} and linear coupling
\citep{allenzhu2017linear} illustrate the general principle that combining
primal and auxiliary-sequence progress yields faster convergence.
\citet[Section~3]{fountoulakis2022open} explicitly name both as motivation
for the PageRank question. The Euclidean energy of our deterministic
recurrence has this familiar role. The additional PageRank-metric energy,
its projection-sector inequality, and the selected-flow work charge address
the cost of realizing that iteration sequence locally.

The closest analysis of transient activity is
\citet{fountoulakis2026complexity} (arXiv version~2, April~8, 2026). They study standard FISTA on the
slightly over-regularized objective $F_{2\rho}$ and use the support at
$\rho$ to give a uniform margin outside an analytical core. Under an
explicit confinement condition, their Theorem~4.3 gives
\begin{equation*}
 \mathcal{O}\!\left(\frac1{\rho\sqrt\alpha}\log\frac\alpha{\epsobj}
       +\frac{\sqrt{\vol(\gB)}}{\rho\alpha^{3/2}}\right),
\end{equation*}
where $\gB$ contains all spurious activations. Their Theorem~4.4 provides
a graph-structural sufficient condition.
The over-regularization comparison already avoids dependence on a global
minimum complementarity margin. Our use of two regularizers builds on that
idea. The further step is an unconditional cumulative-volume estimate for
our constrained accelerated trajectory, using a diffuse source and the
second energy; no containing boundary set $\gB$ is assumed.

Their Proposition~D.4, summarized in Proposition~4.7, proves a worst-case
slowdown for standard FISTA on $m$-edge stars seeded at a leaf. At
$\rho_0=(1-\alpha)/(m(1+\alpha)+(1-\alpha))$, FISTA requires at least $2m$
degree-weighted work for every sufficiently small objective tolerance
depending only on $\alpha$, whereas ISTA reaches the same tolerance in
$\mathcal{O}((1/\alpha)\log(1/\epsobj))$ work independent of $m$.
For fixed $\alpha$, however, $\rho_0=\Theta(1/m)$, so this lower bound is
compatible with the $\widetilde{\mathcal{O}}(1/(\rho\sqrt\alpha))$ target.

Locally evolving-set methods make the cumulative-volume viewpoint explicit
\citep{zhou2024iterative}, including local gradient, Chebyshev, heavy-ball,
and SOR variants. Their accelerated guarantees involve trajectory-dependent
residual-reduction quantities. AESP \citep{huang2025accelerated} combines
local inexact proximal solves with an accelerated outer process in the
spirit of Catalyst \citep{lin2018catalyst}. For $\alpha<1/2$, its
Theorem~3.6 gives the semantic PPR guarantee with the bound
\begin{equation*}
 \softO\!\left(\min\left\{\frac m{\sqrt\alpha},
                  \frac{R_{\mathrm{A}}^2}{\sqrt\alpha\,\epsppr^2}
             \right\}\right),
\end{equation*}
where $R_{\mathrm{A}}$ is the source paper's maximum ratio of initial
weighted gradient masses across proximal subproblems (its equation~(7)).
The source does not give a graph-independent constant bound for this ratio;
its discussion proposes a simplex constraint as a way to control mass.
Thus mass-constrained acceleration has an explicit precedent. Our proof
uses a degree-scaled box together with a source-dependent mass cap, and
proves the projection comparison and cumulative work needed for the
$1/(\rho\sqrt\alpha)$ RPPR guarantee. Our continuation changes the
$\ell_1$ regularization scale; it is not itself a sequence of
Catalyst-shifted quadratic objectives.

\subsection{Other local diffusion and PPR estimation guarantees}

\citet{bai2024faster} give a general local diffusion framework. Their
Theorem~3.3 and Corollary~3.6 recover the unaccelerated
$\mathcal{O}(1/(\alpha\epsppr))$ PageRank scale; their Section~6 leaves accelerated
LocalSOR and LocalCH guarantees open. Thus a general local solver framework
and practical acceleration are precedents, but do not supply the
unconditional accelerated work estimate proved here.

A particularly close accuracy comparison is the SSPPR-D task of
\citet{wei2024absolute}. It uses the same degree-normalized solution-error
criterion as \cref{eq:semantic-ppr-target}. Theorem~20 of their full
preprint gives expected query work
\begin{equation}
 \widetilde{\mathcal{O}}\!\left(\frac1{\epsppr}
       \sqrt{\sum_{t\in\gV}\frac{\pi(v,t)}{d_t}}\right)
 \label{eq:related-sspprd}
\end{equation}
with $\alpha$ treated as constant. Their Section~1.4 remark restores a
linear $1/\alpha$ factor when it varies, and Appendix~D explicitly assumes
that the $\Theta(m)$ graph preprocessing required by Randomized Backward
Search has already been performed. The source-dependent factor in
\cref{eq:related-sspprd} can be valuable. This guarantee shares our
solution-error criterion; the relevant distinctions are teleportation
dependence, preprocessing, and Monte Carlo versus certified output.

Recent PPR estimation bounds use several other accuracy criteria.
\citet[Section~1 and equation~(2)]{bertram2026undirected} characterize
undirected PPR estimation with constant teleportation and thresholded
relative error, with variants for optional random-vertex, sorted-neighbor,
and adjacency queries. Their Theorem~4.2.2 gives a worst-graph,
average-source $\Omega(\min\{m,1/\delta\})$ lower bound for that
criterion. It is not a lower bound with explicit $\alpha$ dependence for
\cref{eq:semantic-ppr-target}.
\citet[Definitions~1.1--1.2]{jiang2026nearoptimality} study fixed-teleportation
single-source absolute and relative error on directed graphs. Absolute
error in PageRank probabilities is stronger than degree-normalized absolute
error on the same graph, but the directed lower-bound constructions do
not transfer automatically to our undirected task. Similarly,
\citet{wang2024revisiting} concern PageRank contributions and single-node
centrality on directed graphs; \citet{thorup2026pagerank} sharpen the
single-node bounds for bounded in-degree, again with fixed teleportation.
\citet{thorup2026instance} establish instance-optimality up to logarithmic
factors for adaptive bidirectional estimation in specified directed graph
classes. Their target is one vertex's global PageRank, averaged over uniform
starting vertices, with fixed teleportation and probabilistic relative error
(Sections~1--2 and Theorem~4.1 of arXiv version~6, August~3, 2026).
These lower bounds do not determine the joint parameter dependence of
either task studied here.

\citet{kwok2026asymmetric} extend sublinear linear-system algorithms to
asymmetric diagonally dominant matrices. Their Section~1.2 specifies a
supplied scalar query $\vt^\trans\vx^*$, with theorem-specific matrix,
vector, and sampling access. This differs from discovering and returning
an entire sparse approximation. Repeated scalar queries would also need a
charged procedure for finding the relevant target labels. Separately,
\citet{lin2024constraints} analyze accelerated primal--dual methods and
finite sparsity identification for an $\ell_1$ objective constrained by a
strongly convex function, including a PageRank example. Iteration and
identification bounds for that constrained formulation do not themselves
bound cumulative graph-row work for $F_\rho$.

\subsection{Obstacle methods and the role of the dual}

The nonnegative quadratic and its flow interpretation have substantial
algorithmic precedents. \citet[Theorem~1.1]{chen2021diffusion} solve
$\ell_2$ flow diffusion with nonnegativity constraints in randomized
$\mathcal{O}(m\log^8 n\log(1/\varepsilon))$ work on the supplied whole graph;
their Section~1.4 identifies strong locality as a further direction.
\citet[Theorem~2 and Corollary~3]{vladu2025barrier} give faster interior
point methods for symmetric $M$-matrix quadratics over the nonnegative
orthant, with work proportional to
$\widetilde{\mathcal{O}}(n^{1/3}\operatorname{nnz}(\mQ)\log(1/\varepsilon))$
including the stated norm and condition-number logarithms. These results
apply to the relevant optimization geometry, but their input-scale work
does not supply local support discovery or boundary accounting. Our claim
concerns those costs and the stated joint parameter bound, rather than a
new obstacle formulation or the first acceleration of constrained diffusion.

\citet{li2026resistance} apply related local obstacle ideas to
$\ell_1$-regularized resistance and spanning-tree counting. Their
Definition~3 uses the unshifted graph Laplacian and an unweighted linear
penalty. Theorem~13 (arXiv version~2, September~4, 2026) gives
$\widetilde{\mathcal{O}}(\lambda^{-3})$ work per source after
$\widetilde{\mathcal{O}}(m)$ graph preprocessing, with a sparse potential and
coordinatewise accuracy. Their nested SDD solves and preprocessed boundary
oracle are close methodological precedents; the objective, preprocessing,
and parameter bound differ from the RPPR problem studied here.

Stieltjes obstacle problems and primal--dual active-set methods are classical
\citep{karkkainen2003augmented}. Nonnegative principal inverses explain safe
pivots, and standard Fenchel duality gives the grounded electrical-flow
formulation in \cref{prop:grounded-flow-dual}. The outside obstacle slack is
exactly a violated dual vertex inequality. In the non-lazy coordinates of
\cref{eq:related-nonlazy-map}, the residual
$\vr^{\mathrm{nl}}=\eunit{v}-\alpha_{\mathrm{nl}}^{-1}\mL_{\alpha_{\mathrm{nl}}}\vy$
satisfies
\[
 \mL_{\alpha_{\mathrm{nl}}}\vy
       -\alpha_{\mathrm{nl}}(\eunit{v}-\rho\mD\one)
 =\alpha_{\mathrm{nl}}(\rho\mD\one-\vr^{\mathrm{nl}}).
\]
This directly identifies the residual activation used by the SDD active-set
method with local separation of a dual inequality.

We use this geometry as a shared language for certification and discovery.
A global obstacle or flow formulation by itself does not supply a local
algorithm: constructing all constraints or testing all vertices would
already read the graph. The complexity proofs must additionally control
which rows are opened, how inactive boundary records are maintained, and
how often numerical state is rebuilt. Those charges are explicit in both
algorithms. The deterministic proof also shows that a rounded implementation
can preserve the required inequalities without exact support identification.

\section{Bounded Arithmetic for the Deterministic Algorithm}
\label{app:bounded-arithmetic}
The exact-real model does not bound the encoding length of a repeatedly
updated rational state. We therefore give a specified rounded algorithm,
prove its guarantees for the actual stored iterates, and account for all
integer lengths. This appendix does not assume a positive activation margin.
Throughout, the constant-output cases have been removed, so
$0<\alpha<1$, $0<r<1/d_v\leq1$, and $U=4r<4$.

\subsection{A dimension-independent perturbation bound}

Fix a valid stage. At the actual stored states $\vxi,\vz\in\gK_r$,
form the ideal $\vy,\vq$ by \cref{eq:det-iteration}. Model a rounded step as
\begin{equation*}
 \begin{aligned}
 \vp_0&=\operatorname{Proj}_{\gK_r}(\vq+\vu),&
 \vp&=\vp_0+\ve_p,&
 \vxi^+&=\chi\vxi+\theta\vp+\ve_x,\\
 |\vu|&\leq\kappa_{\mathrm{raw}}\vomega,&
 -\kappa_{\mathrm{proj}}\vomega&\leq\ve_p\leq\vzero,&
 -\kappa_{\mathrm{pri}}\vomega&\leq\ve_x\leq\vzero.
 \end{aligned}
\end{equation*}
We require $\vp,\vxi^+\geq\vzero$. The downward errors then preserve
box and mass feasibility. The ideal projection $\vp_0$ is a comparison
object; the implementation need not enumerate all of its positive entries.

\begin{lemma}[Two energies under directed rounding]
\label{lem:bounded-two-energies}
The energies $E$ and $B$ in \cref{eq:det-first-energy,eq:det-second-energy}
satisfy, at the actual stored states,
\begin{equation}
 \begin{gathered}
 E^+\leq\chi E+e_{\mathrm{step}},\qquad
 B^+\leq\chi B+e_{\mathrm{step}},\\
 e_{\mathrm{step}}:=2\mu_{\mathrm{c}}\kappa_{\mathrm{raw}}
       +\tfrac52\kappa_{\mathrm{pri}}
       +\tfrac52(\theta+\mu_{\mathrm{c}})\kappa_{\mathrm{proj}}.
 \end{gathered}
 \label{eq:bounded-energy-errors}
\end{equation}
For uniform error bounds put $\Gamma=e_{\mathrm{step}}/\theta$. Starting
from zero, the rounded trajectory obeys
\begin{equation*}
 \begin{aligned}
 J_r(\vxi_k)-J_r(\vxi^*)&\leq\chi^k+\Gamma,\\
 \norm{\mQ(\vxi_k-\vxi^*)}_2^2&\leq18\alpha^2r+4\Gamma.
 \end{aligned}
\end{equation*}
\end{lemma}
\begin{proof}
Symmetry and $|\mQ|\vomega=\vomega$ imply
\begin{equation}
 \vomega^\trans|\mQ\va|\leq\vomega^\trans|\va|,
 \qquad |\mQ\va|\leq\kappa\vomega
       \quad\hbox{if }|\va|\leq\kappa\vomega.
 \label{eq:bounded-weighted-contractions}
\end{equation}
Every vector in $\gK_r$ has weighted mass at most one. In particular,
\cref{lem:det-correction-domain} gives $\vxi^*,\vt\in\gK_r$, with
$\vomega^\trans\vt=m_h/\alpha=m_r\leq1$.
The two comparison vectors therefore satisfy
\[
 |(\vp_0-\vxi^*)^\trans\vu|\leq2\kappa_{\mathrm{raw}},
 \qquad
 |(\vp_0-\vt)^\trans\mQ\vu|\leq2\kappa_{\mathrm{raw}}.
\]
The normal of the projection of $\vq+\vu$ has the required sector
signs in both metrics. Replacing that normal by $\vq-\vp_0$ changes
the comparison estimate by at most $2\mu_{\mathrm{c}}\kappa_{\mathrm{raw}}$.
This follows directly from the squared-distance identity in the proof of
\cref{lem:det-comparison}; no nonexpansiveness in a second metric is used.

Let $\vxi_0^+=\chi\vxi+\theta\vp_0$ and
$\vd=\vxi^+-\vxi_0^+=\theta\ve_p+\ve_x\leq\vzero$.
Put $\kappa_d=\theta\kappa_{\mathrm{proj}}+\kappa_{\mathrm{pri}}$.
The lost weighted mass is at most one, so
$\norm{\vd}_2^2\leq\kappa_d\vomega^\trans|\vd|\leq\kappa_d$.
Also
$\vomega^\trans|\mQ\vxi_0^+-\vh|\leq1+m_h\leq2$.
Expanding either $J_r$ or $\mathcal A_r$, their positive mass penalties
decrease under $\vd\leq\vzero$; the remaining linear increase is at
most $2\kappa_d$, by \cref{eq:bounded-weighted-contractions}.
The quadratic increase is at most $\kappa_d/2$, using
$\mQ\preceq\mI$ and $\norm{\mQ\vd}_2\leq\norm{\vd}_2$.

Similarly $\norm{\ve_p}_2^2\leq\kappa_{\mathrm{proj}}$.
Expanding the Euclidean or $\mQ$ mirror-distance term costs at most
$5\mu_{\mathrm{c}}\kappa_{\mathrm{proj}}/2$, because the absolute weighted mass
of the difference between two nonnegative vectors of weighted mass at most
one is at most two.
Adding these contributions proves \cref{eq:bounded-energy-errors}.
Summing the geometric recurrences proves the objective estimate.
The proof of \cref{lem:det-response} with the additional energy $\Gamma$
gives a bound $8\lambda_r m_h+2\Gamma$ on
$\norm{\mQ\vxi_k-\vh}_2^2$, and hence the stated response bound.
\end{proof}

Uniform error per coordinate alone would permit dimension-dependent error.
The estimates above instead use the mass of the actual feasible vectors
and their actual downward losses. Unexposed zero coordinates are not
assigned fictitious independent rounding errors.

\begin{lemma}[Local work survives rounding]
\label{lem:bounded-kinetic-volume}
If $\theta\kappa_{\mathrm{raw}}\leq\lambda_r/4$ and
$\Gamma\leq\alpha^2r$, the actual rounded kinetic supports satisfy
\begin{equation}
 \sum_{k<K}\vol(\supp\vp_k)\leq360K/r.
 \label{eq:bounded-kinetic-volume}
\end{equation}
\end{lemma}
\begin{proof}
Use the comparison set $\gC_r$ from \cref{eq:det-comparison-margin}.
On a selected coordinate $p_i>\xi_i^*$, also $p_{0,i}>\xi_i^*$.
Its lower normal vanishes; the upper and mass normals and $\ve_p\leq0$
can only reduce the selected flow. The raw error contributes at most
$\nu d_i$, where $\nu=\theta\kappa_{\mathrm{raw}}$.
The telescope from \cref{eq:det-selected-flow} therefore becomes
\begin{equation*}
 \frac{\lambda_r}2 V_{\mathrm{out}}
 \leq\sqrt{(V_{\mathrm{out}}+2K/r)H_2}
          +\nu(V_{\mathrm{out}}+2K/r),
 \qquad H_2\leq K(18\alpha^2r+4\Gamma).
\end{equation*}
Downward primal error affects the response estimate; it introduces no
additional term into the identity for the raw kinetic vector at the next
actual state. With $Y=V_{\mathrm{out}}+2K/r$ and $\nu\leq\lambda_r/4$,
the displayed inequality gives
$Y\leq4K/r+4\sqrt{YH_2}/\lambda_r
 \leq4K/r+Y/2+8H_2/\lambda_r^2$.
Consequently $Y\leq8K/r+16H_2/\lambda_r^2\leq360K/r$.
\end{proof}

\subsection{A rounded implementation and its sparse reporter}

Retain the stage parameters $\theta$ and $0<\delta\leq\lambda_r/2$ from
\cref{alg:deterministic-rppr}. Choose a dyadic grid width
$h_{\mathrm{q}}=1/H$, where $H$ is a power of two,
with
\begin{equation*}
 h_{\mathrm{q}}\leq
 \min\{1/8,\theta\tau/256,\delta/2,\lambda_r/8\},
 \qquad \tau=\alpha\delta^2/8.
\end{equation*}
Use the coarsest such grid at the first stage and refine it by halving as
needed at later stages. The grids are thus nested. Let
\begin{equation*}
 T=1/\theta,\qquad K=Tq_{\mathrm{blk}},\qquad
 q_{\mathrm{blk}}:=\min\{j\in\sN:2^j\geq2/\tau\}.
\end{equation*}
Here $T$ is an integer, and $(1-\theta)^T<1/2$ implies
$\chi^K\leq\tau/2$. This finite schedule uses no stored growing exact
power. Its iteration count is $\mathcal{O}(\log(1/\tau)/\sqrt\alpha)$.

Store densities
\begin{equation*}
 \xi_i/\omega_i=\sigma X_i,\qquad Z_i=z_i/\omega_i,
 \qquad \sigma,X_i,Z_i\in h_{\mathrm{q}}\sN_0.
\end{equation*}
At iteration boundaries $1/2\leq\sigma\leq1$.
Maintain primal neighbor records $L_i$, initially
$L_i=\sum_{j\sim i}X_j$, together with exact neighbor sums of $Z_j$
and the current baseline densities. For notational convenience set
$c=(1-\alpha)/2$ and $q_0=(1+\alpha)/2$ in this appendix.
The normalized response used by the reporter is
\begin{equation}
 \widetilde{\mathscr R}_i=(q_0-\mu_{\mathrm{c}})X_i-cL_i/d_i.
 \label{eq:bounded-approximate-response}
\end{equation}
The source is always computed exactly from the stored baseline:
\begin{equation}
 \frac{h_i}{\omega_i}
 =\frac{\alpha\one_{\{i=v\}}}{d_i}
    -q_0\frac{\overline x_i}{\omega_i}
    +\frac c{d_i}\sum_{j\sim i}\frac{\overline x_j}{\omega_j}.
 \label{eq:bounded-exact-source}
\end{equation}
The seed contribution is zero off its one known coordinate.

Use \cref{eq:det-raw-key} with
$\mathscr R_i$ replaced by \cref{eq:bounded-approximate-response} and
with exact kinetic and source terms. The finite reporter computes the exact
box/mass multiplier $\gamma$ for this represented raw vector. Enumerate
only the \emph{closed} tail
\begin{equation}
 k_i\geq\frac{\lambda_r/\theta+\gamma+h_{\mathrm{q}}}{\sigma}.
 \label{eq:bounded-closed-tail}
\end{equation}
At each emitted vertex set
$p_i/\omega_i=\lfloor p_{0,i}/\omega_i\rfloor_{h_{\mathrm{q}}}$, where
$\lfloor t\rfloor_{h_{\mathrm{q}}}=h_{\mathrm{q}}\lfloor t/h_{\mathrm{q}}\rfloor$.
Because $h_{\mathrm{q}}\leq U$, \cref{eq:bounded-closed-tail} is exactly the
condition for a positive stored projection. Equality must be included.
Enumerating positive ideal coordinates smaller than one grid unit would
incur work absent from \cref{eq:bounded-kinetic-volume}; the closed-tail
reporter avoids that enumeration.

Remove the old source and kinetic exceptions. Set
\begin{equation*}
 \sigma'=\lfloor\chi\sigma\rfloor_{h_{\mathrm{q}}},\qquad
 X_i\gets\left\lfloor
 X_i+\frac{\theta(p_i/\omega_i)}{\sigma'}
 \right\rfloor_{h_{\mathrm{q}}}\quad\hbox{at emitted vertices only}.
\end{equation*}
A changed $X_j$ scatters its \emph{actual} dyadic increment to every
neighbor record $L_i$ with $i\sim j$. Set $\vz\gets\vp$ and form the
exact kinetic neighbor sums by scanning these same emitted rows.
Since $\chi\sigma\geq1/4$ and $h_{\mathrm{q}}\leq1/8$, the scale stays
positive. At the old boundary $X_i\leq2U$. Scale rounding therefore loses
at most $2Uh_{\mathrm{q}}$ in physical primal density, and rounding the update
loses at most another $h_{\mathrm{q}}$.

If $\sigma'<1/2$, perform a scalar rebase on all retained records:
\begin{equation*}
 X_i\gets\lfloor\sigma'X_i\rfloor_{h_{\mathrm{q}}},\qquad
 L_i\gets\lfloor\sigma'L_i\rfloor_{h_{\mathrm{q}}},\qquad
 \sigma\gets1.
\end{equation*}
Both replacements use their pre-rebase values. Do not scatter these $X_i$
changes: the neighbor records have already been scaled. Do not scale the
kinetic or baseline neighbor sums. Rebuild all keys from the resulting
records. If $\sigma'\geq1/2$, set $\sigma\gets\sigma'$ instead and
refresh only touched ordinary keys and the new exceptions. No projection
query occurs during an incomplete update. The optional rebase adds at most
one grid unit of downward primal error. Thus
\begin{equation}
 \kappa_{\mathrm{proj}}=h_{\mathrm{q}},\qquad
 \kappa_{\mathrm{pri}}\leq(2U+2)h_{\mathrm{q}}\leq10h_{\mathrm{q}}.
 \label{eq:bounded-downward-errors}
\end{equation}
Nonnegativity is enforced by construction.

\begin{lemma}[Controlled neighbor-response error]
\label{lem:bounded-neighbor-error}
The stored neighbor records satisfy
\begin{equation}
 \frac1{d_i}\left|L_i-\sum_{j\sim i}X_j\right|\leq4h_{\mathrm{q}},
 \qquad \kappa_{\mathrm{raw}}\leq2h_{\mathrm{q}}/\theta.
 \label{eq:bounded-neighbor-error}
\end{equation}
\end{lemma}
\begin{proof}
Let $e_i^L=L_i-\sum_{j\sim i}X_j$. Actual scattered increments preserve
this difference between rebases, including when records are first exposed.
At a rebase, denote its scalar floor losses by $\ell_j$ and $\ell_i^L$.
Then
\[
 e_i^{L,+}=\sigma'e_i^L+\sum_{j\sim i}\ell_j-\ell_i^L,
 \qquad 0\leq\ell_j,\ell_i^L<h_{\mathrm{q}}.
\]
Because $d_i\geq1$ and $\sigma'<1/2$, an old bound of $4h_{\mathrm{q}}$
yields a new bound at most
$\sigma'4h_{\mathrm{q}}+2h_{\mathrm{q}}\leq4h_{\mathrm{q}}$.
The initial error is zero. The only approximate response term is this
primal neighbor sum; its effect on the raw density is at most
\[
 \frac{|u_i|}{\omega_i}
 \leq\frac{\sigma c|e_i^L|}{\theta(1+\theta)d_i}
 \leq\frac{2h_{\mathrm{q}}}\theta.
\]
Retaining the stage's exposed records ensures that this induction does not
silently discard a previously nonzero neighbor contribution.
\end{proof}

Substituting \cref{eq:bounded-downward-errors,eq:bounded-neighbor-error}
into \cref{eq:bounded-energy-errors}, and using $\theta\leq1/2$, gives
\begin{equation}
 \Gamma\leq29h_{\mathrm{q}}/\theta<\tau/8,
 \qquad \theta\kappa_{\mathrm{raw}}\leq2h_{\mathrm{q}}\leq\lambda_r/4.
 \label{eq:bounded-certified-budgets}
\end{equation}
Also $\tau\leq\alpha^2r$. Hence
\cref{lem:bounded-two-energies,lem:bounded-kinetic-volume} certify a gap
below $\tau$ for the actual candidate and a cumulative volume at most
$360K/r$. These statements allow approximate primal neighbor records;
they never identify them with an exact matrix product.

\paragraph{Charging rebases.}
A scale interval starts at one. Each step decreases scale by at most
$\theta+h_{\mathrm{q}}$, so crossing below $1/2$ takes at least
$1/[2(\theta+h_{\mathrm{q}})]$ steps. Since $h_{\mathrm{q}}\leq\theta$,
there are at most $4\theta K+1=\mathcal{O}(1+\log(1/\tau))$ rebases.
Each is a scalar and tree pass over at most $\mathcal{O}(1+K/r)$ retained records;
it inspects no graph adjacency list. Thus rebases contribute only an
additional logarithmic factor. Ordinary scatters and the final candidate
scan retain their explicit cumulative-volume charges.

\subsection{Integer encodings and a finite projection search}

It remains to show that the reporter's arithmetic does not accumulate a
product of degree denominators or one new denominator per iteration.
Write
\[
 \alpha=A_0/B_\alpha,\qquad\theta=1/T,\qquad
 h_{\mathrm{q}}=1/H,\qquad\sigma=S/H,
\]
and express dyadic densities and neighbor records as integers over $H$.
Let $M_i=\sum_{j\sim i}Z_j$. Multiplying the raw key by its degree gives
\begin{equation*}
 \begin{aligned}
 W_i:=d_i k_i={}&-
   \frac{(q_0-\mu_{\mathrm{c}})d_iX_i-cL_i}{\theta(1+\theta)}\\
 &+\frac1\sigma\left\{
   \chi d_iZ_i-\frac{(q_0-\mu_{\mathrm{c}})d_iZ_i-cM_i}{1+\theta}
            +\frac{d_i h_i/\omega_i}{\theta}\right\}.
 \end{aligned}
\end{equation*}
The degree denominator cancels from every term. By
\cref{eq:bounded-exact-source}, all $W_i$ have denominators dividing
$2B_\alpha HT(T+1)S$; the physical weighted keys $\sigma W_i$ have
one common denominator dividing $2B_\alpha H^2T(T+1)$.
Include the represented denominator of $r$ for the common shift and the
box height. A comparison of two unweighted keys uses their two individual
degrees, whereas subtree weighted sums use this one common denominator.
No least common multiple of encountered degrees is formed.

For a direct implementation, partition the records into two disjoint
balanced trees. Clearing the fixed stage denominators expresses the physical
pre-shift density as
\begin{equation*}
 \frac{S\mathfrak b_i+H\mathfrak e_i}{G_0d_i},
\end{equation*}
where $\mathfrak b_i,\mathfrak e_i$ are integers, $G_0$ is a fixed
positive integer, and $\mathfrak e_i=0$ outside the current source and
kinetic exceptions. The ordinary tree stores $\mathfrak b_i/d_i$ with
global scale $S$; the exception tree stores
$\mathfrak n_i/d_i=(S\mathfrak b_i+H\mathfrak e_i)/d_i$ with scale one.
A registry places each vertex in exactly one tree and records the old key
needed for deletion. Each subtree stores cardinality, degree sum, and
numerator sum. All updates and weighted sums therefore use integers.

A physical clipped-mass query is four tail queries: two trees, each at
the lower and upper clipping thresholds. Its breakpoints, in units scaled
by $G_0$, are the ordered sequences
\[
 \{S\mathfrak b_i/d_i\},\quad\{\mathfrak n_i/d_i\},
 \quad\{S\mathfrak b_i/d_i-G_0U\},
 \quad\{\mathfrak n_i/d_i-G_0U\}.
\]
Let $N$ be the number of retained records. The argument of
\cref{lem:det-reporter} applies to these four lists.
Successively binary search each list by rank, narrowing the bracket until
it has no breakpoint from that list in its interior. Previously excluded
breakpoints remain excluded. There are $\mathcal{O}(\log(N+2))$ queries per list,
each using $\mathcal{O}(\log(N+2))$ integer tree operations. An affine interval
then gives the root by one rational division. Equality returns immediately;
a zero cap returns the zero projection. The reporter still costs
$\mathcal{O}(\log^2(N+2))$ operations plus its stored positive output, independently
of activation margins.

First test the clipped mass at $\gamma=0$, as in \cref{lem:det-reporter};
if it is at most $m_r$, take $\gamma=0$. For the active-mass case,
write $a_i^{\mathrm{den}}$ for a represented raw density after the common
$\lambda_r/\theta$ shift. Let $\gI$ be the free coordinates in the final
interval and $\gU_0$ the saturated coordinates. If $\gI$ is nonempty, the
last division gives
\begin{equation*}
 \gamma=\frac{\sum_{i\in\gI}d_i a_i^{\mathrm{den}}
               +U\sum_{i\in\gU_0}d_i-m_r}
              {\sum_{i\in\gI}d_i}.
\end{equation*}
An empty free set is handled by the exact breakpoint comparison.
The only new divisor is a degree sum. The ideal projected densities are
immediately rounded to the grid, so that divisor never enters the next
stored vector. This is a finite ordered search, not repeated numerical
bisection toward an unknown sign margin.

Normalized primal densities are at most $2U$ at iteration boundaries and
remain bounded by an absolute constant during a rebase step. Neighbor sums
add at most the encoding length of a degree or an exposed-record count.
The grids are nested across stages, so the old baseline is exact on the
new grid. It follows that every integer temporary has length
\begin{equation}
 \mathcal{O}\!\left(L_{\mathrm{in}}+\log H+\log T+
       \log(N+2)+\log(d_{\max}+2)\right).
 \label{eq:bounded-integer-length}
\end{equation}
Here $L_{\mathrm{in}}$ includes the binary lengths of the rational input
numerators and denominators and the maximum length of an encountered vertex
label; $d_{\max}$ is the maximum encountered degree. Intermediate cross products change only the absolute
constant. In particular, there is no linear-in-$K$ denominator growth.

\subsection{Rounded handoff and the bit-complexity theorem}

At completion of a stage, materialize the actual candidate and perform one
\emph{exact} projected-gradient step as in \cref{eq:det-repair-map}.
This scan must use the actual stored primal densities; the approximate
records $L_i$ are not a valid replacement for its matrix product.
The projected-gradient point $\vp$ obeys
$\norm{\vp-\vx_r^*}_2\leq\delta/2$ and has a valid subgradient of
norm at most $\delta/2$, by \cref{eq:det-pg-certificate}. Return the
rounded downward correction
\begin{equation}
 u_i/\omega_i=
    \bigl\lfloor[p_i/\omega_i-\delta]_+\bigr\rfloor_{h_{\mathrm{q}}},
 \qquad \overline\vx^+=\max\{\overline\vx,\vu\}.
 \label{eq:bounded-terminal-repair}
\end{equation}
The total density error is at most
$\delta/2+\delta+h_{\mathrm{q}}\leq2\delta$, and $\vu\leq\vx_r^*$.
On a retained positive coordinate, uniform clipping helps the subsolution
inequality, while the additional grid loss changes the $\mQ$ product by
at most $h_{\mathrm{q}}\omega_i$. Thus
\[
 (\mQ\vu-\vb)_i
 \leq(-\lambda_r+\delta/2-\alpha\delta+h_{\mathrm{q}})\omega_i\leq0.
\]
At a zero coordinate, $(\mQ\vu-\vb)_i\leq0$ follows directly from
the Stieltjes signs. The coordinatewise
maximum preserves safety, as in \cref{lem:det-repair}. Consequently the
actual dyadic baseline satisfies all of \cref{eq:det-repaired-invariant}:
source diffuseness, the support bound, and gap at most $2\delta^2/r$.
It supplies the invariant for the next stage and the requested final gap.
Newly positive neighbors are scanned only after this repair.

\begin{theorem}[Bounded-arithmetic deterministic local acceleration]
\label{thm:bounded-deterministic-rppr}
For rational inputs $\alpha,\rho,\epsobj>0$ in the main problem with
$0<\rho<1/d_v$, the rounded continuation algorithm described in this
appendix returns
$\vzero\leq\widehat\vx\leq\vx_\rho^*$ with objective gap at most
$\epsobj$. Its support volume is at most $1/\rho$. Its fully charged
local operation count is $\softO(1/(\rho\sqrt\alpha))$.

Let $L_{\mathrm{par}}=\log(2/(\alpha\rho\min\{1,\epsobj\}))$ and set
\begin{equation*}
 B:=1+L_{\mathrm{in}}+L_{\mathrm{par}}
       +\log\!\left(2+\frac{L_{\mathrm{par}}}{\rho\sqrt\alpha}\right)
       +\log(d_{\max}+2).
\end{equation*}
All numerical integers have $\mathcal{O}(B)$ bits. With schoolbook integer arithmetic
and exact division, the total bit cost is
$\softO(B^2/(\rho\sqrt\alpha))$. The output stores rational densities
and original degrees. The theorem permits arbitrarily small positive
activation margins and exact threshold equalities. For $\rho\geq1/d_v$,
one degree query identifies the exact zero solution.
\end{theorem}
\begin{proof}
For $\alpha=1$, the exact one-coordinate solution has rational density
$[1/d_v-\rho]_+$ and satisfies the stated bounds directly. Assume
$0<\alpha<1$ for the continuation argument.
The initial baseline and the stage induction are the same as in
\cref{alg:deterministic-rppr}, with \cref{eq:bounded-terminal-repair} at
each handoff. The block schedule, \cref{eq:bounded-certified-budgets}, and
\cref{lem:bounded-two-energies} give the required candidate gap.
The rounded repair closes the induction and the final choice of $\delta$
proves accuracy and support safety.

The rounded cumulative volume, finite reporter, sparse source refreshes,
and scalar rebases have just been charged. Their sum over the halving
schedule uses \cref{eq:det-geometric-work} and is
$\softO(1/(\rho\sqrt\alpha))$. At any stage the number of retained
records is $N=\mathcal{O}(1+L_{\mathrm{par}}/(\rho\sqrt\alpha))$.
Choosing the coarsest allowed nested grids gives
$\log H=\mathcal{O}(L_{\mathrm{par}})$, and $\log T=\mathcal{O}(\log(1/\alpha))$.
Hence \cref{eq:bounded-integer-length} is $\mathcal{O}(B)$.
All comparisons, products, sums, floors, and divisions can be performed
with $\mathcal{O}(B^2)$ bit operations. Label processing is included in $L_{\mathrm{in}}$.
No operation uses an unrepresented real oracle, so multiplication by this
encoding cost proves the bit bound.
\end{proof}

\section{Explicit Seed Distributions and Their Input Costs}
\label{app:seed-scope}

The main problem begins with one known seed label. This appendix uses a
different, explicit input interface: a list of $h$ distinct vertex labels
and positive weights $(i,s_i)$, with $\sum_i s_i=1$. Reading this list,
checking its mass, querying its degrees, and building an indexed source map
are charged. These labels are initially known; all subsequent graph labels
must still be discovered through charged incidences. The graph assumptions
and original-degree normalization are unchanged. The results in this
appendix count exact-real words; the bounded-arithmetic theorem in
\cref{app:bounded-arithmetic} retains its point-source scope.

\subsection{Randomized acceleration with additive seed-input work}

The threshold-batch proof has a useful block-source extension. Set
\begin{equation}
 \gB_0:=\{i:s_i>\rho d_i\},\qquad
 \vc_\rho=\alpha\mD^{-1/2}\vs-\alpha\rho\vomega.
 \label{eq:general-seed-initial-block}
\end{equation}
Thus $\gB_0$ is exactly the set of positive coordinates of $\vc_\rho$.
It is computable from the input list and its degree replies. If it is
empty, zero satisfies the full obstacle KKT system and is the optimum.
Otherwise
\begin{equation}
 \vol(\gB_0)<1/\rho,\qquad
 \gB_0\subseteq\mathcal S_\rho^*,\qquad
 \norm{(\vc_\rho)_{\gB_0}}_2\leq\norm{\vb}_2\leq\alpha.
 \label{eq:general-seed-initial-bounds}
\end{equation}
The first inequality follows by summing $\rho d_i<s_i$.
For support containment, a positive load at a zero optimum coordinate
would contradict the off-diagonal signs and KKT. For the norm bound, use
$d_i\geq1$ and $\sum_i s_i^2\leq(\sum_i s_i)^2=1$.

\begin{theorem}[Explicit-distribution randomized RPPR]
\label{thm:general-seed-randomized}
For the explicit seed list above, put
$k_*=|\mathcal S_\rho^*|$ and $V_*=\vol(\mathcal S_\rho^*)$.
There is a Las Vegas local algorithm with expected fully charged work
\begin{equation*}
 \softO\!\left(h+V_*\min\{k_*,\alpha^{-1/2}\}\right)
 \leq\softO\!\left(h+
          \min\{\rho^{-2},\rho^{-1}\alpha^{-1/2}\}\right)
\end{equation*}
that returns $\vzero\leq\widehat\vx\leq\vx_\rho^*$ with objective
gap at most $\epsobj$ and
$\vzero\leq\vb-\mQ\widehat\vx\leq2\alpha\rho\vomega$.
The support volume is at most $1/\rho$. Logarithms may include the
explicit input length $h$; there is no ambient graph-size charge.
\end{theorem}
\begin{proof}
The obstacle reduction, least-supersolution order, KKT signs, and weighted
mass proof in \cref{thm:obstacle-reduction,lem:rppr-order-mass} use only
$\vs\geq\vzero$ and $\one^\trans\vs=1$, except for the assertion that
one active component contains the one seed. That assertion is unnecessary
here. Start the face sequence at $\gB_0$. Its exact face solution is
strictly positive because its load is strictly positive and the principal
inverse is nonnegative with positive diagonal. The proof of
\cref{thm:safe-batched-pivots} then applies to all later batches.
Outside $\gB_0$, the initial load is nonpositive. Thus a negative slack
outside the current face can occur only on its boundary, even when an
unadmitted source label lies farther away. No scan over all source labels
is required at later phases.

The proof of \cref{thm:threshold-batch-depth} uses the singleton initial
block at exactly two places. The causal inequality for block one needs
$(\vc_\rho)_i\leq0$ outside the initial block. The seed-tail estimate needs
$\norm{\vg_0}_2\leq\alpha$. Both hold by
\cref{eq:general-seed-initial-block,eq:general-seed-initial-bounds}, with
$\vg_0=(\vc_\rho)_{\gB_0}$. All other steps allow arbitrary block
sizes: the Stieltjes Cholesky factor, its block-bidiagonal truncation,
inverse ordering, row-map norm, and Chebyshev propagation depend on block
indices rather than on the cardinality of block zero. The support-volume
bound remains $V_*\leq1/\rho$. Therefore the same energy-depth theorem
holds verbatim with initial block $\gB_0$.

Run the certified threshold-batch algorithm from this initial face.
Its energy bound still follows from $\norm{\vb}_2\leq\alpha$.
The numerical wrapper, boundary error estimate, and output projection
therefore have the same guarantees. There are at most
$\min\{k_*-|\gB_0|+1,J+1\}\leq\min\{k_*,J+1\}$ face solves.
Building the source map costs $\widetilde{\mathcal{O}}(h)$ once. Each later source
lookup costs $\mathcal{O}(\log(h+2))$ and occurs only in a charged face or boundary
operation. Exposing the initial block costs at most $V_*$, and all later
face and boundary records obey the same volume bound. The source map
adds $\mathcal{O}(h)$ storage and is not rescanned per phase.

Finally apply \cref{cor:randomized-safe-repair}. Its proof uses the
same KKT signs, $\mQ\vomega=\alpha\vomega$, and support-volume bound,
all of which hold for this distribution. This gives the stronger output
at the same work up to logarithms. If the initial block is empty, zero
already has residual at most $\alpha\rho\vomega$.
At $\alpha=1$, coordinatewise soft thresholding of the explicit source
list is exact in $\mathcal{O}(h)$ work. These branches complete the proof.
\end{proof}

In particular, direct regularization and the bias bound give a general-seed
PPR approximation in expected work
$\widetilde{\mathcal{O}}(h+1/(\epsppr\sqrt\alpha))$.
The result is a direct multi-source algorithm. It does not use linear
superposition of separate regularized optima.

\subsection{The deterministic extension and its different input charge}

\begin{proposition}[Explicit-distribution deterministic continuation]
\label{prop:general-seed-deterministic}
With the same explicit source interface, the exact-real deterministic
method has work
\begin{equation}
 \softO\!\left(\frac{h+1/\rho}{\sqrt\alpha}\right)
 \label{eq:general-seed-deterministic-work}
\end{equation}
for the same RPPR objective and safe-subsolution output guarantees.
\end{proposition}
\begin{proof}
Read the source list and compute
$r_0=\max_i s_i/d_i$. If $\rho\geq r_0$, return the exact zero
solution. At $\alpha=1$, soft thresholding the explicit source list gives
the exact solution in $\mathcal{O}(h)$ work. Otherwise start the continuation schedule
at $r_{\mathrm{old}}=r_0$ and zero baseline. Its first stage has $r\geq r_0/2$, so
$\vb\leq2\alpha r\vomega$. All subsequent invariants follow as before.
The correction domain, both energies, and the selected-flow argument use
only nonnegativity and unit source mass. They therefore give the same
$76K/r$ cumulative kinetic-volume estimate.

The change is in the fixed source-record set. It now lies in
$\supp\vs\cup\supp\overline\vx\cup\partial(\supp\overline\vx)$,
with $\mathcal{O}(h+1/r)$ records. Refreshing these exceptions costs
$\widetilde{\mathcal{O}}(h+1/r)$ per iteration. Thus a stage costs
$\widetilde{\mathcal{O}}(K(h+1/r))$.
There are $\mathcal{O}(1+\log(1/\rho))$ stages, each with
$K=\widetilde{\mathcal{O}}(1/\sqrt\alpha)$, and
$\sum_j1/r_j\leq4/\rho$. This proves
\cref{eq:general-seed-deterministic-work}, including repeated source
refreshes rather than assigning them an additive input charge.
The terminal repair evaluates the supplied source records as well as the
candidate support and boundary. This adds $\mathcal{O}(h)$ scalar work without
opening unselected source rows.
For arbitrary requested accuracy, use the final repair schedule; if a
positive residual upper bound is desired even at a loose objective target,
run with tolerance $\min\{\epsobj,\alpha/4\}$ to omit the optional
loose-accuracy zero branch. The extra precision is only logarithmic.
\end{proof}

\subsection{What superposition alone would give}

For comparison, linearity of unregularized PPR yields the following
weaker direct reduction from point-source routines.

\begin{proposition}[Superposition with an explicit mixture factor]
\label{prop:sparse-source-superposition}
Let $\gK=\supp\vs$. Suppose a point-source routine at $v$ has
semantic tolerance $\varepsilon_v$ and work
$\widetilde{\mathcal{O}}(1/(\sqrt\alpha\varepsilon_v))$, apart from constant
input work. Then the mixture
$\widehat\vpi=\sum_{v\in\gK}s_v\widehat\vpi^{(v)}$ has error at most
$\sum_vs_v\varepsilon_v$. Optimizing these budgets gives work
\begin{equation*}
 \softO\!\left(h+
       \frac{(\sum_{v\in\gK}\sqrt{s_v})^2}
            {\sqrt\alpha\epsppr}\right).
\end{equation*}
The mixture factor lies in $[1,h]$ and equals $h$ for a uniform source.
The work guarantee is deterministic or expected according to the
point-source routine used.
\end{proposition}
\begin{proof}
Linearity and the triangle inequality give the error estimate.
Minimizing $\sum_v1/\varepsilon_v$ subject to
$\sum_vs_v\varepsilon_v\leq\epsppr$ gives
$\varepsilon_v=\epsppr/(\sqrt{s_v}\sum_u\sqrt{s_u})$.
Cauchy--Schwarz gives the factor's range. If $\epsppr\geq1$, zero is
sufficient. Otherwise discard weights at most $\epsppr/(2h)$ and allocate
error budget $\epsppr/2$ among the retained terms, without renormalizing.
The discarded terms contribute at most $\epsppr/2$, since every point-source
PPR vector has degree-normalized maximum norm at most one. Each retained
weight admits a dyadic upper estimate of its square root within a factor
two using $\mathcal{O}(\log(h/\epsppr))$ squared comparisons. Using these estimates
in the budget formula preserves accuracy and increases work by at most a
constant factor. Tolerances at least one use zero outputs; the per-source
constant work is covered by $h$. Merging sparse outputs costs only their
total size and local dictionary logarithms.
\end{proof}

RPPR itself is nonlinear in the source. At $\alpha=1$ its optimum is
\[
 (\vx_\rho^*(\vs))_i
   =\max\{s_i/\sqrt{d_i}-\rho\sqrt{d_i},0\}.
\]
For example, on a single edge with $\rho=1/2$, the two point-source
regularized optima are $(1/2,0)$ and $(0,1/2)$, whereas the optimum for
their uniform seed mixture is zero. Their averaged optima are therefore
not the mixture optimum. The direct extension in
\cref{thm:general-seed-randomized} works by an initial positive-load
block and the same depth theorem, rather than by composing those optima.

\section{Grounded Electrical-Flow Interpretation}
\label{app:grounded-flow-dual}
The KKT slack from \cref{sec:obstacle-active-sets} has an exact network
interpretation. Form an augmented graph by
adding a ground vertex $\mathsf g$.  Give every original edge conductance
$(1-\alpha)/2$ and connect each $i\in\gV$ to $\mathsf g$ with conductance
$\alpha d_i$; when $\alpha=1$, omit the zero-conductance original edges.  Fix
an orientation, let $\mB$ be the signed incidence matrix with the ground
column deleted, and let $\mW$ be the diagonal matrix of positive
conductances.  Define
\begin{equation}
 \mH:=\mD^{1/2}\mQ\mD^{1/2}
 =\alpha\mD+\frac{1-\alpha}{2}(\mD-\mA),
 \qquad
 \vc_\rho^{\mathrm d}:=\mD^{1/2}\vc_\rho=\alpha(\vs-\rho\mD\one).
 \label{eq:grounded-system}
\end{equation}
Then $\mB^\trans\mW\mB=\mH$.

\begin{proposition}[Grounded electrical-flow dual]
\label{prop:grounded-flow-dual}
Under the change of variables $\vy=\mD^{-1/2}\vx$, the obstacle problem is
\begin{equation}
 \min_{\vy\geq\vzero}
 \left\{
  \frac12\norm{\mW^{1/2}\mB\vy}_2^2-(\vc_\rho^{\mathrm d})^\trans\vy
 \right\}.
 \label{eq:grounded-flow-primal}
\end{equation}
Its Fenchel dual, written as a minimum-energy problem, is
\begin{equation}
 \min_{\vf}\ \frac12\vf^\trans\mW^{-1}\vf
 \qquad\text{subject to}\qquad
 \mB^\trans\vf\geq\vc_\rho^{\mathrm d}.
 \label{eq:grounded-flow-dual}
\end{equation}
The two minimum values have opposite signs.  At an optimal primal--dual pair,
\begin{equation*}
 \vf^*=\mW\mB\vy^*,
 \qquad
 \vlambda^*:=\mB^\trans\vf^*-\vc_\rho^{\mathrm d}
 =\mD^{1/2}\vw^*\geq\vzero,
 \qquad
 y_i^*\lambda_i^*=0.
\end{equation*}
More generally, let $\vx^\gU$ be the zero-extended exact solution on a
reachable active set $\gU$, with slack $\vw^\gU$ as defined in
\cref{subsec:reachable-faces}. Set $\vy^\gU:=\mD^{-1/2}\vx^\gU$.
The induced flow $\vf^\gU:=\mW\mB\vy^\gU$ is tight on $\gU$ and
\begin{equation*}
 (\mB^\trans\vf^\gU-\vc_\rho^{\mathrm d})_j=\sqrt{d_j}\,w_j^\gU
 \quad(j\notin \gU).
\end{equation*}
Consequently, $w_j^\gU<0$ is exactly a violated dual vertex constraint.
\end{proposition}

\begin{proof}
The identities in \cref{eq:grounded-system} give
$\vy^\trans\mH\vy=\norm{\mW^{1/2}\mB\vy}_2^2$ and transform
\cref{eq:obstacle-problem} into \cref{eq:grounded-flow-primal}.  The conjugate
identity
\[
 \frac12\norm{\mW^{1/2}\vz}_2^2
 =\max_{\vf}
 \left\{\vf^\trans\vz-\frac12\vf^\trans\mW^{-1}\vf\right\}
\]
and minimization over $\vy\geq\vzero$ make the inner value finite precisely
when $\mB^\trans\vf-\vc_\rho^{\mathrm d}\geq\vzero$.  Fenchel--Rockafellar duality applies
because the quadratic is continuous everywhere and the orthant is nonempty,
with the sign reversal in the displayed minimization convention.  Stationarity in the conjugate pair
gives $\vf^*=\mW\mB\vy^*$; minimization over the orthant gives
complementarity.  Finally,
\[
 \mB^\trans\mW\mB\vy-\vc_\rho^{\mathrm d}
 =\mH\vy-\vc_\rho^{\mathrm d}
 =\mD^{1/2}(\mQ\vx-\vc_\rho)
 =\mD^{1/2}\vw,
\]
which proves both slack identities.
\end{proof}

Thus the active-set method can be interpreted as local constraint
generation for \cref{eq:grounded-flow-dual}: a restricted solve constructs
the flow induced by the current potentials, and boundary slack tests identify
violated vertex inequalities. The implementation uses the local support
and certificate properties of \cref{sec:obstacle-active-sets} to perform
these tests; the batch-depth bound in \cref{sec:threshold-batch-depth}
controls the number of restricted solves. Constructing the full dual or
testing every vertex would require graph-wide access.

\end{document}